%% file: main.tex
\documentclass[conference]{IEEEtran}
\usepackage{cite}
\usepackage{xspace}
\usepackage{xparse}
\usepackage{url}
\usepackage{expl3}
\usepackage{array}
\usepackage{collcell}
\usepackage{comment}
\usepackage{amsmath}
\usepackage{color}
\usepackage[table]{xcolor}
\usepackage{tikz}
\usetikzlibrary{patterns}
\usepackage{pgfplots}
\usepackage{graphicx}
\usepackage{caption}
\usepackage[skip=3pt,font=footnotesize]{subcaption}
\usepackage[super]{nth}
\usepackage{siunitx}
\usepackage{mfirstuc}

\usepackage{amsfonts}
\usepackage{booktabs}

\usepackage{balance}

\input{macros}

\begin{document}

\title{Optimally Hiding Object Sizes\\with Constrained Padding}

\author{\IEEEauthorblockN{Andrew C. Reed}
\IEEEauthorblockA{University of North Carolina at Chapel Hill\\
reed@cs.unc.edu}
\and
\IEEEauthorblockN{Michael K. Reiter}
\IEEEauthorblockA{Duke University\\
michael.reiter@duke.edu}}

\IEEEoverridecommandlockouts
\makeatletter\def\@IEEEpubidpullup{6.5\baselineskip}\makeatother

\maketitle

\begin{abstract}
  Among the most challenging traffic-analysis attacks to confound are
  those leveraging the sizes of objects downloaded over the network.
  In this paper we systematically analyze this problem under realistic
  constraints regarding the padding overhead that the \objStoreTerm is
  willing to incur.  We give algorithms to compute privacy-optimal
  padding schemes---specifically that minimize the network observer's
  information gain from a downloaded object's padded size---in several
  scenarios of interest: \perObjPaddingTerm, in which the
  \objStoreTerm responds to each request for an object with the same
  padded copy; \perReqPaddingTerm, in which the \objStoreTerm pads an
  object anew each time it serves that object; and a scenario unlike
  the previous ones in that the \objStoreTerm is unable to leverage a
  known distribution over the object queries.  We provide
  constructions for privacy-optimal padding in each case, compare them
  to recent contenders in the research literature, and evaluate their
  performance on practical datasets.
\end{abstract}

\input{introduction}

\input{related_work}

\input{design}
\input{security}
\input{performance}
\input{discussion}

\input{conclusion}

\balance
\bibliographystyle{IEEEtranS}
\bibliography{full,main}

\end{document}

%% file: macros.tex
\newif\ifdraft
\drafttrue
\ifdraft
  \newcommand{\todocolor}[1]{\textcolor{red}{#1}}
\else
  \newcommand{\todocolor}[1]{}
\fi

\newcommand{\mike}[1]{\todocolor{[[Mike: #1]]}}

\newcommand{\bits}{\ensuremath{\mathrm{bits}}\xspace}
\newcommand{\bytes}{\ensuremath{\mathrm{B}}\xspace}

\newcommand{\gigabytes}{\ensuremath{\mathrm{GB}}\xspace}

\newcommand{\gigahertz}{\ensuremath{\mathrm{GHz}}\xspace}
\newcommand{\secs}{\ensuremath{\mathrm{seconds}}\xspace}

\newcommand{\secref}[1]{\mbox{Sec.~\ref{#1}}\xspace}
\newcommand{\staticsecref}[1]{\mbox{Sec.~{#1}}}
\newcommand{\secsref}[2]{\mbox{Sec.~\ref{#1}--\ref{#2}}\xspace}
\newcommand{\figref}[1]{\mbox{Fig.~\ref{#1}}}
\newcommand{\figsref}[2]{\mbox{Figs.~\ref{#1}--\ref{#2}}}

\newcommand{\tabref}[1]{\mbox{Table~\ref{#1}}}

\newcommand{\eqnref}[1]{(\ref{#1})\xspace}
\newcommand{\eqnsref}[2]{(\ref{#1})--(\ref{#2})\xspace}
\newcommand{\staticthmref}[1]{Theorem~{#1}}
\newcommand{\staticlmaref}[1]{Lemma~{#1}}
\newcommand{\staticchapref}[1]{Chapter~{#1}}

\newtheorem{prop}{Proposition}
\newcommand{\propref}[1]{\mbox{Prop.~\ref{#1}}\xspace}

\NewDocumentCommand{\infoEnt}{g}%
{%
	\IfNoValueTF{#1}%
	{\ensuremath{\mathbb{H}}\xspace}%
	{\ensuremath{\mathbb{H}\mathopen{}\left({#1}\right)\mathclose{}}\xspace}%
}
\NewDocumentCommand{\infoEntRecursive}{o}%
{%
	\IfNoValueTF{#1}%
	{\ensuremath{H}\xspace}%
	{\ensuremath{H\mathopen{}\left({#1}\right)\mathclose{}}\xspace}%
}
\newcommand{\condlEnt}[2]{\ensuremath{\infoEnt\mathopen{}\left({#1}\mid{#2}\right)\mathclose{}}\xspace}

\NewDocumentCommand{\mutInfo}{g g}%
{%
	\IfNoValueTF{#1}%
	{\ensuremath{\mathbb{I}}\xspace}%
	{\ensuremath{\mathbb{I}\mathopen{}\left({#1};{#2}\right)\mathclose{}}\xspace}%
}

\NewDocumentCommand{\mutInfoEst}{g g g}%
{%
	\IfNoValueTF{#2}%
	{\ensuremath{\tilde{\mathbb{I}}_{#1}}\xspace}%
	{\ensuremath{\tilde{\mathbb{I}}_{#1}\mathopen{}\left({#2};{#3}\right)\mathclose{}}\xspace}%
}

\NewDocumentCommand{\mutInfoInf}{g g}%
{%
	\IfNoValueTF{#1}%
	{\ensuremath{\mathbb{I}_{\infty}}\xspace}%
	{\ensuremath{\mathbb{I}_{\infty}\mathopen{}\left({#1};{#2}\right)\mathclose{}}\xspace}%
}

\NewDocumentCommand{\func}{o}%
{%
	\IfNoValueTF{#1}%
	{\ensuremath{f}\xspace}%
	{\ensuremath{f\mathopen{}\left({#1}\right)\mathclose{}}\xspace}%
}

\NewDocumentCommand{\funcTwoDer}{o}%
{%
	\IfNoValueTF{#1}%
	{\ensuremath{f''}\xspace}%
	{\ensuremath{f''\mathopen{}\left({#1}\right)\mathclose{}}\xspace}%
}

\NewDocumentCommand{\entTerm}{o}%
{%
	\IfNoValueTF{#1}%
	{\ensuremath{h}\xspace}%
	{\ensuremath{h\mathopen{}\left({#1}\right)\mathclose{}}\xspace}%
}
\NewDocumentCommand{\entTermTwoDer}{o}%
{%
	\IfNoValueTF{#1}%
	{\ensuremath{h''}\xspace}%
	{\ensuremath{h''\mathopen{}\left({#1}\right)\mathclose{}}\xspace}%
}

\NewDocumentCommand{\floor}{o}%
{%
	\IfNoValueTF{#1}%
	{\ensuremath{\mathsf{floor}}\xspace}%
	{\ensuremath{\mathsf{floor}\mathopen{}\left({#1}\right)\mathclose{}}\xspace}%
}

\newcommand{\prob}[1]{\ensuremath{\mathbb{P}\mathopen{}\left({#1}\right)\mathclose{}}\xspace}
\newcommand{\expv}[1]{\ensuremath{\mathbb{E}\mathopen{}\left({#1}\right)\mathclose{}}\xspace}
\newcommand{\cprob}[3]{\ensuremath{\mathbb{P}\mathopen{}#1(}{#2}\ensuremath{\;#1|} \ifmmode{\;}\fi {#3}\ensuremath{#1)\mathclose{}}\xspace}

\newcommand{\bigO}[1]{\ensuremath{O\mathopen{}\left({#1}\right)\mathclose{}}\xspace}
\NewDocumentCommand{\nats}{o o}%
{%
	\IfNoValueTF{#1}%
                {\ensuremath{\mathbb{N}}\xspace}
                {\IfNoValueTF{#2}%
                        {\ensuremath{[{#1}]}\xspace}
                        {\ensuremath{[{#1},{#2}]}\xspace}
        }
}

\newcommand{\reals}{\ensuremath{\mathbb{R}}\xspace}

\newcommand{\genericNat}{\ensuremath{i}\xspace}
\newcommand{\genericNatAlt}{\ensuremath{j}\xspace}
\newcommand{\setSize}[1]{\ensuremath{\#{#1}}\xspace}
\newcommand{\genericProb}{\ensuremath{p}\xspace}

\newcommand{\boolTrue}{\ensuremath{\mathit{true}}\xspace}
 
\newcommand{\rvNotation}[1]{\ensuremath{\mathsf{\expandafter\MakeUppercase{#1}}}\xspace}
\newcommand{\valNotation}[1]{\ensuremath{\MakeLowercase{#1}}\xspace}
\newcommand{\setNotation}[1]{\ensuremath{\MakeUppercase{#1}}\xspace}
\newcommand{\fnNotation}[1]{\ensuremath{\expandafter\MakeUppercase{#1}}\xspace}

\ExplSyntaxOn
\cs_new:Npn \intensity #1
   {\fp_eval:n {\MinIntensity + (1.0-\MinIntensity) * ((({#1}-\value{MinNumber}+1)/(\value{MaxNumber}-\value{MinNumber}+1))^(4/5))}
   }
\ExplSyntaxOff
\newcommand{\MinIntensity}{0.0}   
\newcounter{MinNumber}
\newcounter{MaxNumber}
\newcommand{\ApplyGradientX}[1]{\cellcolor[gray]{\intensity{#1}}{\parbox{2em}{\centering \raggedleft{#1}}}}
\newcommand{\ApplyGradientY}[1]{\cellcolor[gray]{\intensity{#1}}{\parbox{1.5em}{\centering \raggedleft{#1}}}}
\newcolumntype{X}{>{\collectcell\ApplyGradientX}c<{\endcollectcell}}
\newcolumntype{Y}{>{\collectcell\ApplyGradientY}c<{\endcollectcell}}

\newcommand{\objStoreTerm}{object store\xspace}
\newcommand{\perObjPaddingTerm}{per-object padding\xspace}
\newcommand{\PerObjPaddingTerm}{Per-object padding\xspace}
\newcommand{\perReqPaddingTerm}{per-request padding\xspace}
\newcommand{\PerReqPaddingTerm}{Per-request padding\xspace}
\newcommand{\perObjAlgLong}{\underline{P}er-\underline{O}bject \underline{P}adding\xspace}
\newcommand{\perObjAlg}{POP\xspace}
\newcommand{\perReqAlgLong}{\underline{P}er-\underline{R}equest \underline{P}adding\xspace}
\newcommand{\perReqAlg}{PRP\xspace}
\newcommand{\perReqAlgInit}{\ensuremath{\textrm{\perReqAlg}_{\textrm{init}}}\xspace}
\newcommand{\perReqAlgInc}{\ensuremath{\textrm{\perReqAlg}_{\textrm{inc}}}\xspace}
\newcommand{\noDistAlgLong}{\underline{P}adding \underline{w}ith\underline{o}ut a \underline{D}istribution\xspace}
\newcommand{\noDistAlg}{PwoD\xspace}

\newcommand{\kAnonymityK}{\ensuremath{k}\xspace}
\newcommand{\lDiversityL}{\ensuremath{\ell}\xspace}
\newcommand{\kAnonymity}{\kAnonymityK-anonymity\xspace}
\newcommand{\lDiversity}{\lDiversityL-diversity\xspace}

\newlength{\figureheight}
\newcommand{\figurewidth}{\columnwidth}

\definecolor{curvecolor}{rgb}{0.129411764705882,0.380392156862745,0.549019607843137}

\newcommand{\objSize}[1]{\ensuremath{|{#1}|}\xspace}
\newcommand{\objSizeIdx}{\ensuremath{k}\xspace}
\newcommand{\objStore}[1]{\ensuremath{\mathsf{obj}_{#1}}\xspace}
\NewDocumentCommand{\objSizeOrdinal}{o}%
{%
	\IfNoValueTF{#1}%
        {\ensuremath{\psi}\xspace}%
        {\ensuremath{\psi\mathopen{}\left({#1}\right)\mathclose{}}\xspace}%
}

\newcommand{\privDataRV}{\ensuremath{\rvNotation{S}}\xspace}
\newcommand{\privDataVal}{\ensuremath{\valNotation{s}}\xspace}
\newcommand{\privDataValAlt}{\ensuremath{\valNotation{s}'}\xspace}
\newcommand{\privDataValMin}{\ensuremath{\valNotation{s}_{\mathrm{min}}}\xspace}
\newcommand{\privDataSubset}{\ensuremath{\setNotation{S}^{\ast}}\xspace}
\newcommand{\privDataDomain}{\ensuremath{\setNotation{S}}\xspace}
\newcommand{\pubDataRV}{\ensuremath{\rvNotation{X}}\xspace}
\newcommand{\pubDataVal}{\ensuremath{\valNotation{x}}\xspace}
\newcommand{\distortPubDataRV}{\ensuremath{\rvNotation{Y}}\xspace}
\newcommand{\distortPubDataAlg}[1]{\ensuremath{\lceil{#1}\rceil}\xspace}
\newcommand{\distortPubDataVal}{\ensuremath{\valNotation{y}}\xspace}
\newcommand{\distortPubDataValAlt}{\ensuremath{\valNotation{y}'}\xspace}
\newcommand{\distortPubDataSubset}[1]{\ensuremath{\setNotation{Y}({#1})}\xspace}
\newcommand{\distortionBound}{\ensuremath{d}\xspace}
\newcommand{\distortionFn}{\ensuremath{\fnNotation{D}}\xspace}
\newcommand{\padFactor}{\ensuremath{c}\xspace}

\newcommand{\loReal}{\ensuremath{z}\xspace}
\newcommand{\hiReal}{\ensuremath{z'}\xspace}
\newcommand{\xferReal}{\ensuremath{\epsilon}\xspace}
\newcommand{\rangeReals}{\ensuremath{R}\xspace}
\newcommand{\partitionBlock}[1]{\ensuremath{\setNotation{b}_{#1}}\xspace}
\newcommand{\partitionBlockProb}[1]{\ensuremath{p_{#1}}\xspace}

\newcommand{\blahutParam}{\ensuremath{\beta}\xspace}
\newcommand{\blahutCutoff}{\ensuremath{\Delta}\xspace}
\newcommand{\blahutIdx}{\ensuremath{t}\xspace}
\NewDocumentCommand{\blahutProb}{g o}%
{%
	\IfNoValueTF{#2}%
        {\ensuremath{u_{#1}}\xspace}%
        {\ensuremath{u_{#1}\mathopen{}\left({#2}\right)\mathclose{}}\xspace}%
}
\NewDocumentCommand{\blahutCProb}{g o o}%
{%
	\IfNoValueTF{#2}%
        {\ensuremath{v_{#1}}\xspace}%
        {\ensuremath{v_{#1}\mathopen{}\left({#2}, {#3}\right)\mathclose{}}\xspace}%
}

\NewDocumentCommand{\anchor}{o}%
{%
	\IfNoValueTF{#1}%
        {\ensuremath{a}\xspace}%
        {\ensuremath{a\mathopen{}\left({#1}\right)\mathclose{}}\xspace}%
}

\NewDocumentCommand{\anchorReach}{o}%
{%
	\IfNoValueTF{#1}%
        {\ensuremath{\setNotation{b}}\xspace}%
        {\ensuremath{\setNotation{b}\mathopen{}\left({#1}\right)\mathclose{}}\xspace}%
}

\NewDocumentCommand{\anchors}{o}%
{%
	\IfNoValueTF{#1}%
        {\ensuremath{\setNotation{A}}\xspace}%
        {\ensuremath{\setNotation{A}\mathopen{}\left({#1}\right)\mathclose{}}\xspace}%
}

\newcommand{\padme}{\textsc{Padm\'{e}}\xspace}
\newcommand{\padmeObjSizeExponent}{\ensuremath{E}\xspace}
\newcommand{\padmeObjSizeExponentBits}{\ensuremath{S}\xspace}
\newcommand{\padmeObjSizeMax}{\ensuremath{M}\xspace}
\newcommand{\padmeBase}{\ensuremath{b}\xspace}
\newcommand{\alpaca}{\textsc{ALPaCA}\xspace}
\newcommand{\pAlpaca}{\textsc{P-\alpaca}\xspace}
\newcommand{\dAlpaca}{\textsc{D-\alpaca}\xspace}
\newcommand{\dAlpacaBinSize}{\ensuremath{\sigma}\xspace}
\newcommand{\exampleObject}[1]{\ensuremath{\mathrm{P{#1}}}\xspace}
\newcommand{\adversaryConfidence}{\ensuremath{\tau}\xspace}

%% file: introduction.tex
\section{Introduction}
\label{sec:intro}

The transmission of objects in a way that hides the sizes of objects
transmitted from a network observer---either to hide which of many
potential objects is transmitted, or as an ingredient in hiding which
sender and receiver are communicating---is a longstanding problem in
traffic-analysis defense.  Indeed, the sizes of objects transmitted
has been shown to single-handedly enable fingerprinting websites or
webpages, even in the presence of otherwise sophisticated defenses
against this practice (e.g.,~\cite{dyer:2012:peek-a-boo}).  Despite
the utility of object sizes in traffic analysis, the expense of hiding
object sizes is such that substantial threads of research on private
communication either do not even attempt to hide object sizes (e.g.,
in low-latency anonymous messaging~\cite{serjantov:2003:passive,
  danezis:2004:mixes, levine:2004:low-latency, zhu:2004:flow} or
protocols for private video downloads,
e.g.,~\cite{zhang:2019:streaming}) or restrict attention to
fixed-sized, small messages (e.g., anonymous microblogging
systems~\cite{corrigan-gibbs:2015:riposte, abraham:2020:blinder}).

In this paper we consider a fundamental and practical instance of this
problem, in which a benevolent \objStoreTerm enables clients to
retrieve its objects.  Each client's communication with the
\objStoreTerm is encrypted for that client, but the sizes of objects
it retrieves is nevertheless revealed to a network observer.  This
network observer might also be one of the clients of the \objStoreTerm
and so can retrieve objects himself.  We further allow this observer
(and the \objStoreTerm, potentially) to know the frequency of requests
for each object.  Being benevolent, the \objStoreTerm is willing to
pad objects before sending them, so that their sizes do not directly
disclose to the network observer the objects others retrieve.
However, the benevolence of the \objStoreTerm extends only so far;
since padding objects consumes more of its bandwidth, the
\objStoreTerm is willing to pad objects only so much.  The question we
consider here is: how should the \objStoreTerm pad its objects subject
to this constraint, to best hide which object a client retrieves from
this network observer?

More specifically, for any object identifier $\privDataVal \in
\privDataDomain$, where \privDataDomain is the set of all object
identifiers, let \objStore{\privDataVal} denote the object with
identifier \privDataVal at the \objStoreTerm, and let
$\objSize{\objStore{\privDataVal}} \in \nats$ denote the size of
\objStore{\privDataVal}.  Consider random variables \privDataRV,
\pubDataRV, and \distortPubDataRV, which take on the object identifier
\privDataVal in a request, the corresponding object's actual size
\objSize{\objStore{\privDataVal}}, and the object's padded size when
returned, respectively.  \privDataRV is distributed according to a
public probability distribution, and so
\[
\prob{\pubDataRV = \pubDataVal} =
\sum_{\privDataVal \in \privDataDomain :
  \objSize{\objStore{\privDataVal}} = \pubDataVal} \prob{\privDataRV =
  \privDataVal}
\]
is also public for each size $\pubDataVal \in \nats$.  The goal of the
object store is to select a padding scheme
\distortPubDataAlg{\cdot} that pads each object
\objStore{\privDataVal} to a (possibly randomized) size
\distortPubDataAlg{\objStore{\privDataVal}} before sending it, in
which case
\[
\prob{\distortPubDataRV = \distortPubDataVal} =
\sum_{\privDataVal \in \privDataDomain}
\cprob{\big}{\distortPubDataAlg{\objStore{\privDataVal}} = \distortPubDataVal}{\privDataRV = \privDataVal}
\prob{\privDataRV = \privDataVal}
\]

We presume that objects are served in full and cannot be compressed,
and so
\begin{equation}
  \prob{\distortPubDataAlg{\objStore{\privDataVal}} <
    \objSize{\objStore{\privDataVal}}} = 0
  \label{eqn:noCompression}
\end{equation}
for all $\privDataVal \in \privDataDomain$.  Moreover, as mentioned
above, the \objStoreTerm is willing to pad each object only so much.
We capture this constraint by requiring, for a specified constant
$\padFactor \ge 1$, that
\begin{equation}
  \prob{\distortPubDataAlg{\objStore{\privDataVal}} > \padFactor \times
    \objSize{\objStore{\privDataVal}}} = 0
\label{eqn:limitExpansion}
\end{equation}
for all $\privDataVal \in \privDataDomain$.  An \objStoreTerm might
prefer \eqnref{eqn:limitExpansion} to requiring only that objects be
expanded by a factor of at most \padFactor in expectation---i.e., that
$\expv{\frac{\distortPubDataAlg{\objStore{\privDataRV}}}{\objSize{\objStore{\privDataRV}}}}
\le \padFactor$, where the expectation is taken with respect to the
distribution of \privDataRV and any random choices of
\distortPubDataAlg{\cdot}---for fairness.  Limiting padding overhead
only in expectation would still permit some objects to be expanded by
more than a factor of \padFactor, which might cause some clients to
unduly suffer if they retrieve that object more than others or do so
over a bandwidth-limited (or priced) connection.

Having adopted constraints
\eqnsref{eqn:noCompression}{eqn:limitExpansion}, the \objStoreTerm
cannot hide \textit{all} objects retrieved from the network observer.
For example, an object \objStore{\privDataVal} for which $\padFactor
\times \objSize{\objStore{\privDataValAlt}} <
\objSize{\objStore{\privDataVal}}$ or $\padFactor \times
\objSize{\objStore{\privDataVal}} <
\objSize{\objStore{\privDataValAlt}}$ for every other object
\objStore{\privDataValAlt} will be the only object padded to a value
in the range
$\nats[\objSize{\objStore{\privDataVal}}][\padFactor\times
  \objSize{\objStore{\privDataVal}}]$.  So, \objStore{\privDataVal}
will be identifiable to the network observer when returned to a
client.  Rather than give up and return to today's status quo,
however, the benevolent \objStoreTerm strives to protect client
privacy as well as it can (subject to
\eqnsref{eqn:noCompression}{eqn:limitExpansion}).  In this paper, we
presume the measure of client privacy the \objStoreTerm seeks to
minimize is the mutual information between \privDataRV and
\distortPubDataRV, also referred to as the \textit{information gain},
i.e.,
\begin{equation}
\mutInfo{\privDataRV}{\distortPubDataRV}
= \infoEnt{\privDataRV} - \condlEnt{\privDataRV}{\distortPubDataRV}
\label{eqn:mutInfo}
\end{equation}
where \infoEnt denotes entropy.  That is, the \objStoreTerm seeks to
choose a padding scheme \distortPubDataAlg{\cdot} to minimize
\mutInfo{\privDataRV}{\distortPubDataRV} subject to constraints
\eqnsref{eqn:noCompression}{eqn:limitExpansion}.

In this context, our algorithmic contributions, detailed in
\secref{sec:algs}, are as follows:
\begin{itemize}%[nosep,leftmargin=1em,labelwidth=*,align=left]
\item In a \textit{\perObjPaddingTerm} scenario in which the
  \objStoreTerm pads each object a single time and serves this single
  padded object in response to repeated requests (including possibly
  one from the network observer),
  \distortPubDataAlg{\objStore{\privDataVal}} is fixed across
  retrievals of \objStore{\privDataVal} and so
  \distortPubDataAlg{\cdot} is, in effect, a function.  In this case,
  we characterize the privacy-optimal scheme
  \distortPubDataAlg{\cdot}, i.e., that minimizes
  \mutInfo{\privDataRV}{\distortPubDataRV}, by showing that it is
  nondecreasing in the sense that if
  $\objSize{\objStore{\privDataVal}} \le
  \objSize{\objStore{\privDataValAlt}}$ then
  $\distortPubDataAlg{\objStore{\privDataVal}} \le
  \distortPubDataAlg{\objStore{\privDataValAlt}}$.  We use this
  finding to give an explicit algorithm to choose
  \distortPubDataAlg{\cdot} to minimize
  \mutInfo{\privDataRV}{\distortPubDataRV} subject to the above
  constraints.  This algorithm computes \distortPubDataAlg{\cdot}
  in \bigO{\setSize{\privDataDomain}^2} time, where
  \setSize{\privDataDomain} is the cardinality of \privDataDomain.

\item In a \textit{\perReqPaddingTerm} scenario in which the
  \objStoreTerm pads each object anew before serving it each time, we
  observe a connection between our problem and \textit{rate-distortion
    minimization} as originally considered by
  Shannon~\cite{shannon:1959:theorems}.  By expressing our problem as
  an instance of rate-distortion minimization, we can apply classic
  results to compute the privacy-optimal \distortPubDataAlg{\cdot}
  numerically.

\item Both of the above contributions require the \objStoreTerm to
  accurately predict the distribution of \privDataRV, i.e., the
  distribution of requests it will receive, in order to compute the
  optimal padding scheme \distortPubDataAlg{\cdot}.  In cases where
  the \objStoreTerm cannot do so, e.g., because the adversary can
  affect this distribution, we give an algorithm to solve for a
  padding scheme \distortPubDataAlg{\cdot} that pads objects to
  achieve a measure that upper-bounds
  \mutInfo{\privDataRV}{\distortPubDataRV} even for an adversarially
  chosen distribution.  Perhaps surprisingly, this algorithm computes
  \distortPubDataAlg{\cdot} in \bigO{\setSize{\privDataDomain}} time
  and so is the most scalable of our algorithms.
\end{itemize}

In \secref{sec:security}, we empirically evaluate these algorithms
using two real-world datasets to compare the security they provide to
recently published algorithms for similar goals, for both
\perObjPaddingTerm and \perReqPaddingTerm~\cite{cherubin:2017:alpaca,
  nikitin:2019:purbs}.  Our evaluation shows that in terms of both
information gain for the adversary and the adversary's practical
ability to detect object retrievals as being in classes of interest,
our \perReqPaddingTerm algorithm provided better security than the
\perReqPaddingTerm contender, and similarly, our \perObjPaddingTerm
algorithm dominated its contenders.  Even our algorithm that does not
leverage a distribution for \privDataRV remained competitive, while
being intrinsically robust to any mistakes in estimating that
distribution that would cripple the other algorithms.

We report on the performance of our algorithms on these datasets in
\secref{sec:performance}.  Only our algorithm for finding the
privacy-optimal \perReqPaddingTerm scheme presented scaling
challenges for large numbers of objects and large values of
\padFactor.  However, we show that once it computed a solution, this
solution could be incrementally adapted in response to object (size)
changes much faster than computing the distribution used by
\distortPubDataAlg{\cdot} from scratch.

Finally, we discuss limitations and possible extensions of our results
in \secref{sec:discussion}.  We conclude in \secref{sec:conclusion}.
All of our source code, to include our algorithm implementations,
datasets, and test code, is available on
GitHub\footnote{https://github.com/andrewreed/constrained-padding}.

%% file: related_work.tex
\section{Related Work}
\label{sec:related}

\paragraph{Mutual information as a privacy measure}
Mutual information \mutInfo{\privDataRV}{\distortPubDataRV} between a
random variable \distortPubDataRV that \textit{will} be disclosed and
secret information \privDataRV that should \textit{not} be disclosed
has been used to measure the information leakage from \privDataRV to
\distortPubDataRV in numerous contexts for over 70 years
(e.g.,~\cite{shannon:1949:secrecy, wyner:1975:wire-tap,
  leung-yan-cheong:1978:gaussian, csiszar:1978:broadcast,
  yamamoto:1994:coding, gopala:2008:secrecy, gunduz:2008:lossless,
  sankar:2013:databases, makhdoumi:2014:privfun,
  ding:2019:submodprivfun, yin:2020:meter}).  Our contribution, we
believe, lies in leveraging mutual information specifically to
optimize the padding applied to objects to hide the object returned
over a network, subject to padding size constraints, a problem for
which heuristic solutions continue to be published
(e.g.,~\cite{cherubin:2017:alpaca, nikitin:2019:purbs}, as we will
discuss in \secref{sec:security}).  As we will show, in some cases we
can adapt known methods for minimizing mutual information in other
contexts, and in others we develop novel and very efficient algorithms
for doing so.

\paragraph{Padding to achieve other privacy measures}
Despite the longevity and pervasiveness of mutual information as a
privacy measure, it is not without its critics
(e.g.,~\cite{smith:2009:foundations, issa:2020:operational}).  Other
measures of privacy for padding security have been studied in contexts
similar to ours~\cite{liu:2012:k-indistinguishable, liu:2014:pptp},
notably adaptations of measures initially proposed for ensuring
privacy for statistical databases, namely
\kAnonymity~\cite{samarati:2001:k-anonymity, sweeney:2002:k-anonymity}
and its generalization
\lDiversity~\cite{machanavajjhala:2007:l-diversity}.  Aside from
having critics of their own (e.g.,~\cite{li:2007:t-closeness}), these
measures are incomparable to mutual information: subject to padding
constraints, minimizing mutual information does not necessarily
achieve the maximum \lDiversityL for \lDiversity, and maximizing
\lDiversityL does not necessarily minimize mutual information.  And
while the works known to us~\cite{liu:2012:k-indistinguishable,
  liu:2014:pptp} are more ambitious than ours in attempting to address
correlated retrievals over multiple flows (necessitated by their focus
on web applications), the complexity of doing so renders them unable
to provably optimize the tradeoff between privacy and overhead.
\iffalse
Simplifying their goals to a single object retrieval as in our
scenario enables the maximization of \lDiversityL subject to a padding
overhead constraint \padFactor to be encoded as a linear program that
can be solved by a general solver, but our experiments doing so
resulted solving times that were impractical \mike{Is that true?}, as
we will discuss in \secref{sec:eval}.
\fi

\paragraph{Leakage based on communication volume}
Kellaris et al.~\cite{kellaris:2016:generic} analyzed an
``outsourced'' (i.e., untrusted) \objStoreTerm that returns some
subset of its objects in response to \textit{range queries} on their
(encrypted) search keys by clients.  They evaluated basic sources of
leakage that practical defenses might permit, of which one is
\textit{communication volume}.  This form of leakage occurs when the
\objStoreTerm observes \textit{how many} objects it returns to the
client (but not which ones), as systems leveraging ORAMs
(e.g.,~\cite{goldreich:1996:orams, asharov:2019:locality,
  chakraborti:2019:range}) would typically leak.  (See also
Naveed~\cite{naveed:2015:fallacy}.)  While communication volume
leakage bears some similarity to our problem, our study differs in the
threat model (our \objStoreTerm is trusted, theirs is not), what is
sensitive (search terms in their case, returned objects in ours), and
the nature of the results (they presented attacks, whereas we present
defenses).

Generic defenses against communication-volume leakage in the model of
Kellaris et al.\ have been explored recently
(e.g.,~\cite{kamara:2018:suppression, kamara:2019:computationally}).
Lossless defenses (as we require here) provide strong privacy but
retrieve an object via multiple fixed-sized retrievals---of total size
larger than the original object, and so themselves padded---over
multiple rounds of interaction.  In contrast, our design does not
require multiple rounds or otherwise alter the communication pattern
of object retrievals (aside from padding them), and focuses on
limiting bandwidth overhead to a maximum multiplicative overhead per
object while achieving the best privacy that limit allows against a
network observer.

\paragraph{Leakage based on access patterns}
The second basic source of leakage analyzed by Kellaris et
al.~\cite{kellaris:2016:generic} is \textit{access patterns}, in which
the \objStoreTerm observes \textit{which} objects it returns to the
client, as would be typical of systems based on searchable symmetric
and structured encryption (e.g.,~\cite{chang:2005:remote,
  chase:2010:structured, cash:2013:boolean, kamara:2018:sql}) or on
deterministic and order-preserving encryption
(e.g.,~\cite{popa:2011:cryptdb, arasu:2013:cipherbase}).  Various
other works have studied access-pattern leakage and its detrimental
effects on privacy against an untrusted \objStoreTerm
(e.g.,~\cite{islam:2012:access, cash:2015:leakage,
  grubbs:2019:learning, kornaropoulos:2019:neighbor}).  Defenses
against this risk tend to incorporate fake queries
(e.g.,~\cite{pang:2012:obfuscating, mavroforakis:2015:modular,
  grubbs:2020:pancake}), again which we eschew here due to their
overheads, or ORAMs, whose overheads are even worse
(e.g.,~\cite{chang:2016:evaluation}).  Still, most defenses against
access-pattern leakage typically assume all objects are the same size
(e.g.,~\cite{grubbs:2020:pancake}), so that object lengths do not leak
information.  It is exactly this assumption that we seek to relax.

\paragraph{Webpage fingerprinting}
The context within which traffic analysis has been most often
considered recently is webpage fingerprinting.  In this context, a
network observer seeks to determine which webpage (or which website) a
user is accessing based on features visible to the observer.  The
variety of features that the observer might leverage is
vast~\cite{yan:2018:feature}, but it has been shown that communication
volume is particularly powerful for fingerprinting
webpages~\cite{dyer:2012:peek-a-boo}.  Webpage fingerprinting has been
attempted within, among others, HTTPS
traffic~\cite{cheng:1998:traffic, danezis:2003:traffic,
  miller:2014:clinic, gonzalez:2016:profiling, alan:2019:diversity};
encrypted web-proxy traffic~\cite{sun:2002:statistical,
  hintz:2002:fingerprinting}; SSH proxy
tunnels~\cite{bissias:2005:privacy, liberatore:2006:inferring};
netflow records~\cite{coull:2007:privacy, yen:2009:browser}; packet
headers~\cite{macia-fernandez:2010:isp}; and Tor
traffic~\cite{herrmann:2009:website, panchenko:2011:website,
  cai:2012:touching, wang:2013:improved, juarez:2014:critical,
  wang:2014:effective, hayes:2016:k-fingerprinting,
  cherubin:2017:alpaca}.  Many (though not all) proposed defenses
against webpage fingerprinting exploit protocol features in TCP and/or
HTTP (e.g.,~\cite{sun:2002:statistical, luo:2011:httpos}).  While
heuristically padding web objects has been considered
(e.g.,~\cite{sun:2002:statistical, cherubin:2017:alpaca}), we know of
no work in this context or others that shows how to pad objects so as
to maximize privacy subject to a bandwidth overhead constraint.

%% file: design.tex
\section{Algorithms}
\label{sec:algs}

In this section we develop algorithms for calculating the padding
scheme \distortPubDataAlg{\cdot} that optimally achieves privacy
subject to padding overhead constraints
\eqnsref{eqn:noCompression}{eqn:limitExpansion}.  We address multiple
scenarios: ``\perObjPaddingTerm,'' in which objects are padded once
and then provided in response to requests distributed according to the
known distribution on \privDataRV (\secref{sec:algs:per-object});
``\perReqPaddingTerm,'' in which objects can be padded anew in
response to each request, with each request again distributed
according to the known distribution on \privDataRV
(\secref{sec:algs:per-request}); and a third scenario in which the
distribution on \privDataRV is unknown to the \objStoreTerm
(\secref{sec:algs:unknown-dist}).  In all cases, our target is to
minimize the mutual information
\mutInfo{\privDataRV}{\distortPubDataRV} of \privDataRV and the padded
object sizes \distortPubDataRV revealed to the attacker.  Once
\distortPubDataAlg{\cdot} has been calculated, each invocation
\distortPubDataAlg{\objStore{\privDataVal}} involves simply sampling
from the distribution for \distortPubDataRV conditioned on the event
$\privDataRV = \privDataVal$ and then padding \objStore{\privDataVal}
to the sampled size, and so each invocation is very efficient.  As
such, the primary cost we focus on in this paper is the cost of
calculating the distribution of \distortPubDataRV conditioned on
\privDataRV.

\begin{table*}[t]
  \caption{Selected objects used in algorithm illustrations in \secref{sec:algs}.}
  \label{tbl:objects}
  \begin{center}
    \begin{tabular}{clrr}
      \toprule
      Label & \multicolumn{1}{c}{URL (accessed Apr 25, 2021)} & \multicolumn{1}{c}{Size (\bytes)} & \multicolumn{1}{c}{\parbox{4.5em}{\centering Downloads per day}} \\
      \midrule
      \exampleObject{0} & \url{https://images.unsplash.com/photo-1572095426476-808d659b4ea3} & 2493855 & 2.53\hspace{1.25em} \\
      \exampleObject{1} & \url{https://images.unsplash.com/reserve/qstJZUtQ4uAjijbpLzbT_LO234824.JPG} & 3833489 & 27.92\hspace{1.25em} \\
      \exampleObject{2} & \url{https://images.unsplash.com/photo-1583582829797-b2990fb9946b} & 7929946 & 5.41\hspace{1.25em} \\
      \exampleObject{3} & \url{https://images.unsplash.com/photo-1591672524177-261a7744a2b6} & 13322074 & 12.41\hspace{1.25em} \\
      \exampleObject{4} & \url{https://images.unsplash.com/photo-1579832888877-74d7a790df36} & 13589747 & 1.09\hspace{1.25em} \\
      \exampleObject{5} & \url{https://images.unsplash.com/photo-1558136015-7002a0f5e58d} & 16235142 & 5.54\hspace{1.25em}  \\
      \exampleObject{6} & \url{https://images.unsplash.com/photo-1586030307451-dfc64907aaa5} & 16719886 & 10.65\hspace{1.25em} \\
      \exampleObject{7} & \url{https://images.unsplash.com/photo-1558729923-720bbb76a430} & 19437984 & 5.07\hspace{1.25em} \\
      \exampleObject{8} & \url{https://images.unsplash.com/photo-1528233090455-e245a0c50575} & 25905442 & 2.27\hspace{1.25em} \\
      \exampleObject{9} & \url{https://images.unsplash.com/photo-1559422721-1ed9b8d28236} & 34389677 & 4.23\hspace{1.25em} \\
      \bottomrule
    \end{tabular}
  \end{center}
\end{table*}

For each algorithm we provide, we illustrate the padding scheme 
\distortPubDataAlg{\cdot} it produces for objects selected from one 
of the datasets we utilize in our evaluations in 
\secsref{sec:security}{sec:performance}. We defer
detailed discussion of these datasets to that section, but the objects
selected for illustration in this section are shown in
\tabref{tbl:objects}.  We selected these objects to illustrate the
differences among algorithms effectively; we do not claim that the
selected objects are representative of the dataset from which we drew
them.  Central to two of our algorithms is knowing the
frequency with which each object is requested, so that we can estimate
the distribution of \privDataRV.  For the objects listed in
\tabref{tbl:objects}, this information is provided in the ``Downloads
per day'' column.

We reiterate that our threat model permits an attacker to observe the
sizes of objects returned in response to others' requests, and to
query the \objStoreTerm itself to observe padded objects.  However, we
assume that these are the only features available to the attacker.  In
particular, the sizes of all requests (as observable on the network)
are the same, and neither others' requests nor the contents of the
objects returned to these requests are visible to the attacker.  We
further assume that the \textit{contents} of objects returned to
others' requests have no observable effect on network-level features
available to the attacker; i.e., the return of two different objects,
if padded to the same size, will be indistinguishable to the attacker.

\subsection{\xcapitalisewords{\PerObjPaddingTerm}}
\label{sec:algs:per-object}

The case of \perObjPaddingTerm differs from \perReqPaddingTerm in that
\distortPubDataAlg{\objStore{\privDataVal}} is invoked only once per
identifier \privDataVal.  All queries for \privDataVal then return
this once-padded object.  As such, we can view
\distortPubDataAlg{\cdot} as a deterministic function in this case.
\PerObjPaddingTerm is appropriate when the expense of padding anew for
each query is deemed too costly, or if objects will be diffused via
content-distribution networks (CDNs) outside the \objStoreTerm's
control.

A classic result
(e.g.,~\cite[\staticthmref{2.4.1}]{cover:2006:infotheory})
regarding mutual information is that in addition to
\eqnref{eqn:mutInfo},
\begin{equation}
\mutInfo{\privDataRV}{\distortPubDataRV}
= \infoEnt{\distortPubDataRV} - \condlEnt{\distortPubDataRV}{\privDataRV}
\label{eqn:mutInfoSymmetric}
\end{equation}
As such, when $\distortPubDataRV =
\distortPubDataAlg{\objStore{\privDataRV}}$ is a deterministic
function of \privDataRV, as in this case, then
$\condlEnt{\distortPubDataRV}{\privDataRV} = 0$ and so
$\mutInfo{\privDataRV}{\distortPubDataRV} =
\infoEnt{\distortPubDataRV}$.  Therefore, to minimize
\mutInfo{\privDataRV}{\distortPubDataRV} it suffices to minimize
\infoEnt{\distortPubDataRV}.

In the remainder of this section, we develop an algorithm for
computing the optimal function \distortPubDataAlg{\cdot} for
\perObjPaddingTerm.  We first prove an important property about the
optimal \distortPubDataAlg{\cdot} in
\secref{sec:algs:per-object:nondecreasing} and then provide an algorithm to
compute the optimal \distortPubDataAlg{\cdot} efficiently in
\secref{sec:algs:per-object:algorithm}.

\subsubsection{The Privacy-Optimal \distortPubDataAlg{\cdot} is Nondecreasing}
\label{sec:algs:per-object:nondecreasing}

In this section we prove that any privacy-optimal scheme
\distortPubDataAlg{\cdot} for \perObjPaddingTerm is
\textit{nondecreasing} in object size, in the sense that
$\objSize{\objStore{\privDataVal}} <
\objSize{\objStore{\privDataValAlt}} \Rightarrow
\distortPubDataAlg{\objStore{\privDataVal}} \le
\distortPubDataAlg{\objStore{\privDataValAlt}}$.  To do so, consider
any function \distortPubDataAlg{\cdot} satisfying padding constraints
\eqnsref{eqn:noCompression}{eqn:limitExpansion} that is \textit{not}
nondecreasing, i.e., for which there are objects
\objStore{\privDataVal} and \objStore{\privDataValAlt} with
$\objSize{\objStore{\privDataVal}} <
\objSize{\objStore{\privDataValAlt}}$ but for which
$\distortPubDataAlg{\objStore{\privDataVal}} >
\distortPubDataAlg{\objStore{\privDataValAlt}}$.  Since
$\objSize{\objStore{\privDataVal}} <
\objSize{\objStore{\privDataValAlt}}$, increasing
\distortPubDataAlg{\objStore{\privDataValAlt}} to
\distortPubDataAlg{\objStore{\privDataVal}} will not violate our
padding constraints (in particular, \eqnref{eqn:limitExpansion}).
Similarly, decreasing \distortPubDataAlg{\objStore{\privDataVal}} to
\distortPubDataAlg{\objStore{\privDataValAlt}} will not violate our
constraints (in particular, \eqnref{eqn:noCompression}).  Below we
show that one of these alternatives will decrease
\infoEnt{\distortPubDataRV} and, therefore,
\mutInfo{\privDataRV}{\distortPubDataRV}, showing that
\distortPubDataAlg{\cdot} cannot be privacy-optimal.  To do so, we use
the following proposition.

\begin{prop}
\label{prop:transfer}
Let \func be a function that is defined on the interval 
$\rangeReals \subseteq \reals$, and that has a negative second 
derivative. For all $\loReal, \hiReal \in \rangeReals$ such that 
$\loReal \le \hiReal$ and for any $\xferReal > 0$ such that 
$\loReal - \xferReal \in \rangeReals$ and $\hiReal +
\xferReal \in \rangeReals$:
\[
\func[\loReal] + \func[\hiReal] > \func[\loReal-\xferReal] +
\func[\hiReal+\xferReal]
\]
\end{prop}

\begin{proof}
  Since a negative second derivative implies a decreasing first derivative,
  \[
  \frac{\func[\loReal] - \func[\loReal-\xferReal]}{\loReal - (\loReal - \xferReal)} > \frac{\func[\hiReal+\xferReal] - \func[\hiReal]}{(\hiReal+\xferReal) - \hiReal} 
  \]
  and the result follows by rearranging terms.
\end{proof}

To use \propref{prop:transfer}, consider a partition
$\{\partitionBlock{\distortPubDataVal}\}_{\distortPubDataVal \in
  \nats}$ of the objects based on their padded sizes, i.e., so that
$\partitionBlock{\distortPubDataVal} = \{\privDataVal \in
\privDataDomain : \distortPubDataAlg{\objStore{\privDataVal}} =
\distortPubDataVal\}$.  If we let
$\partitionBlockProb{\distortPubDataVal} = \sum_{\privDataVal \in
  \partitionBlock{\distortPubDataVal}} \prob{\privDataRV =
  \privDataVal}$ and $\entTerm[\genericProb] = -\genericProb \log_2
\genericProb$ for $\genericProb \in [0,1]$, then
\begin{equation}
\infoEnt{\distortPubDataRV} = \sum_{\distortPubDataVal \in \nats}
\entTerm[\partitionBlockProb{\distortPubDataVal}]
\label{eqn:infoEnt}
\end{equation}
Note that the second derivative of \entTerm[\genericProb] is
$\entTermTwoDer[\genericProb] = -\frac{1}{\genericProb \ln (2)}$,
which is negative for $\genericProb \in [0,1]$.

Now, return to considering a function \distortPubDataAlg{\cdot} that
is \textit{not} nondecreasing, i.e., there are objects
\objStore{\privDataVal} and \objStore{\privDataValAlt} with
$\objSize{\objStore{\privDataVal}} <
\objSize{\objStore{\privDataValAlt}}$ but $\distortPubDataVal >
\distortPubDataValAlt$ where $\distortPubDataVal =
\distortPubDataAlg{\objStore{\privDataVal}}$ and
$\distortPubDataValAlt =
\distortPubDataAlg{\objStore{\privDataValAlt}}$.  If
$\partitionBlockProb{\distortPubDataVal} \le
\partitionBlockProb{\distortPubDataValAlt}$, then setting $\xferReal =
\prob{\privDataRV = \privDataVal}$ and applying
\propref{prop:transfer} yields
$\entTerm[\partitionBlockProb{\distortPubDataVal}] +
\entTerm[\partitionBlockProb{\distortPubDataValAlt}] >
\entTerm[\partitionBlockProb{\distortPubDataVal} - \xferReal] +
\entTerm[\partitionBlockProb{\distortPubDataValAlt} + \xferReal]$.  In
other words, decreasing \distortPubDataAlg{\objStore{\privDataVal}}
from \distortPubDataVal to \distortPubDataValAlt reduces
\infoEnt{\distortPubDataRV}.  Otherwise (i.e.,
$\partitionBlockProb{\distortPubDataVal} >
\partitionBlockProb{\distortPubDataValAlt}$), setting $\xferReal =
\prob{\privDataRV = \privDataValAlt}$ and applying
\propref{prop:transfer} shows that increasing
\distortPubDataAlg{\objStore{\privDataValAlt}} from
\distortPubDataValAlt to \distortPubDataVal reduces
\infoEnt{\distortPubDataRV}.  Either case reveals that the initial
function \distortPubDataAlg{\cdot}, which is not nondecreasing, cannot
be privacy-optimal, either, and so any privacy-optimal function
\distortPubDataAlg{\cdot} must be nondecreasing.

\subsubsection{Finding the Privacy-Optimal \distortPubDataAlg{\cdot}}
\label{sec:algs:per-object:algorithm}

That the privacy-optimal \distortPubDataAlg{\cdot} must be
nondecreasing allows us to leverage dynamic programming to compute the
solution.  In general, to leverage dynamic programming, it is
necessary to express our optimization problem recursively, i.e., so
that the optimal solution is computed by combining optimal solutions
to overlapping subproblems
(e.g.,~\cite[\staticchapref{6}{}]{dasgupta:2006:algorithms}).

Let $\nats[\genericNat][\genericNatAlt] = \{\genericNat, \ldots,
\genericNatAlt\}$, $\nats[\genericNat] = \nats[1][\genericNat]$, and
\setSize{\privDataDomain} be the cardinality of set \privDataDomain.
Because the privacy-optimal \distortPubDataAlg{\cdot} is nondecreasing
as a function of object size, there is a total order $\objSizeOrdinal:
\nats[\setSize{\privDataDomain}] \rightarrow \privDataDomain$ of
\privDataDomain by object size---i.e., a bijection satisfying
$\objSize{\objStore{\objSizeOrdinal[\objSizeIdx]}} \le
\objSize{\objStore{\objSizeOrdinal[\objSizeIdx+1]}}$ for $\objSizeIdx
\in \nats[\setSize{\privDataDomain}-1]$---such that each block
\partitionBlock{\distortPubDataVal} of the partition induced by
\distortPubDataAlg{\cdot} is of the form
$\partitionBlock{\distortPubDataVal} =
\{\objSizeOrdinal[\objSizeIdx]\}_{\objSizeIdx \in
  \nats[\genericNat][\genericNatAlt]}$ where
$\objSize{\objStore{\objSizeOrdinal[\genericNatAlt]}} \le
\distortPubDataVal \le \padFactor \times
\objSize{\objStore{\objSizeOrdinal[\genericNat]}}$.  We equivalently
denote this block by
\partitionBlock{\nats[\genericNat][\genericNatAlt]}, and analogously
denote $\partitionBlockProb{\nats[\genericNat][\genericNatAlt]} =
\sum_{\privDataVal \in
  \partitionBlock{\nats[\genericNat][\genericNatAlt]}}
\prob{\privDataRV = \privDataVal}$.

Now consider the function \infoEntRecursive defined recursively as
$\infoEntRecursive[0] = 0$ and, for $0 < \genericNatAlt \le
\setSize{\privDataDomain}$,
\begin{equation}
\infoEntRecursive[\genericNatAlt] = \min_{\genericNat \le
  \genericNatAlt : \objSize{\objStore{\objSizeOrdinal[\genericNatAlt]}}
  \le \padFactor \times \objSize{\objStore{\objSizeOrdinal[\genericNat]}}}
\left(\infoEntRecursive[\genericNat-1] +
\entTerm[\partitionBlockProb{\nats[\genericNat][\genericNatAlt]}]\right)
\label{eqn:recursion}
\end{equation}
Then, computing \infoEntRecursive[\setSize{\privDataDomain}] computes
the privacy-optimal padding scheme \distortPubDataAlg{\cdot} for
the \perObjPaddingTerm case: for the value of \genericNat minimizing
the right hand side of \eqnref{eqn:recursion}, we set
$\distortPubDataAlg{\objStore{\objSizeOrdinal[\objSizeIdx]}} =
\objSize{\objStore{\objSizeOrdinal[\genericNatAlt]}}$ for all
$\objSizeIdx \in \nats[\genericNat][\genericNatAlt]$.  (Note that no
other objects will be padded to this size, as otherwise
\eqnref{eqn:recursion} could be minimized further, as shown in
\secref{sec:algs:per-object:nondecreasing}.) The recursion 
\eqnref{eqn:recursion} exhibits the properties needed for dynamic 
programming to be effective because \infoEntRecursive[\genericNat] 
can be leveraged in the calculation of \infoEntRecursive[\genericNatAlt] 
for every $\genericNatAlt > \genericNat$.  That is,
\infoEntRecursive[\genericNat] can be computed only once and saved to
be looked up when needed.  As a consequence, dynamic programming
(e.g.,~\cite[\staticchapref{6}{}]{dasgupta:2006:algorithms}) yields an
algorithm that runs in time \bigO{(\setSize{\privDataDomain})^2}.  In
the following sections, we refer to this algorithm as \perObjAlgLong
(\perObjAlg).

\begin{figure}[t]
  \setcounter{MinNumber}{0}%
  \setcounter{MaxNumber}{1}%
  \iffalse
  \begin{subfigure}[t]{\columnwidth}
    \centering
    %\hspace{-.25em}
    \input{figures/algorithm-demos/pop/c_1.tex}
    \caption{$\padFactor = 1$}
    \label{fig:pop-demo:padFactor_1}
  \end{subfigure}
  \fi
  
  \begin{subfigure}[t]{\columnwidth}
    \centering
    %\hspace{-.25em}
    \input{figures/algorithm-demos/pop/c_2.tex}
    \caption{$\padFactor = 2$}
    \label{fig:pop-demo:padFactor_2}
  \end{subfigure}
	
  \begin{subfigure}[t]{\columnwidth}
    \centering
    %\hspace{-.25em}
    \input{figures/algorithm-demos/pop/c_3.tex}
    \caption{$\padFactor = 3$}
    \label{fig:pop-demo:padFactor_3}
  \end{subfigure}
  \vspace*{-1.0em}
  \caption{\cprob{\big}{\distortPubDataRV = \distortPubDataVal}{\privDataRV = \privDataVal} produced by \perObjAlg (\secref{sec:algs:per-object}).}
  \label{fig:pop-demo}
\end{figure}

When applied to the objects listed in \tabref{tbl:objects}, \perObjAlg
produces the conditional distributions shown in \figref{fig:pop-demo}.
\figref{fig:pop-demo:padFactor_2} shows the case $\padFactor = 2$, and
\figref{fig:pop-demo:padFactor_3} shows the case $\padFactor = 3$.
(We believe these values for \padFactor are larger than would
typically be tolerated in practice, but we use large values here to
illustrate the effects of the algorithm in light of the small number,
but diverse sizes, of the objects in \tabref{tbl:objects}.)  Because
\perObjAlg produces a deterministic padding scheme
\distortPubDataAlg{\cdot}, each row includes only one nonzero entry.
It is perhaps most insightful to consider how these distributions
change as \padFactor moves from $\padFactor = 2$ to $\padFactor = 3$.
For example, \perObjAlg prescribes that \exampleObject{0} be padded to
the size of \exampleObject{1} when $\padFactor = 2$, leaving
\exampleObject{2} isolated.  In contrast, when $\padFactor = 3$,
\perObjAlg instead prescribes padding \exampleObject{1} to the size of
\exampleObject{2}, leaving \exampleObject{0} isolated.  It is simple
to confirm that under the padding constraints
\eqnsref{eqn:noCompression}{eqn:limitExpansion}, at least one of
\exampleObject{0}, \exampleObject{1}, and \exampleObject{2} must be
left isolated by a deterministic padding scheme when $\padFactor
\le 3$, but when allowed, \perObjAlg prefers to isolate
\exampleObject{0} because it is requested less often (see
\tabref{tbl:objects}).

\subsection{\xcapitalisewords{\PerReqPaddingTerm}}
\label{sec:algs:per-request}

In the \perReqPaddingTerm scenario, the padding scheme
\distortPubDataAlg{\cdot} can be calculated as a special case of
\textit{rate-distortion minimization} proposed by
Shannon~\cite{shannon:1959:theorems} (see
also~\cite[\staticsecref{IV}]{berger:1998:lossy}).  Specifically,
using our terminology, rate-distortion minimization solves for a
scheme \distortPubDataAlg{\cdot} minimizing
\mutInfo{\privDataRV}{\distortPubDataRV} subject to the constraint
\[
\expv{\distortionFn(\privDataRV, \distortPubDataAlg{\objStore{\privDataRV}})}
\le \distortionBound
\]
where $\distortionFn: \privDataDomain \times \nats \rightarrow [0,
  \infty]$ is a \textit{distortion function}, $\distortionBound \in
\reals$ is a positive constant, and the expectation is computed
relative to the distribution of \privDataRV and random choices made by
\distortPubDataAlg{\cdot}.  Written explicitly, this expected value is
\[
\expv{\distortionFn(\privDataRV, \distortPubDataAlg{\objStore{\privDataRV}})}
= \sum_{\privDataVal \in \privDataDomain} \sum_{\distortPubDataVal \in \nats} \prob{\privDataRV = \privDataVal} \prob{\distortPubDataAlg{\objStore{\privDataVal}} = \distortPubDataVal} \distortionFn(\privDataVal, \distortPubDataVal)
\]
Specifying $\distortionFn(\privDataVal, \distortPubDataVal) = \infty$
for any \privDataVal such that $\prob{\privDataRV = \privDataVal} > 0$
implies that a solution for rate-distortion minimization, if it
exists, must set $\prob{\distortPubDataAlg{\objStore{\privDataVal}} =
  \distortPubDataVal} = 0$.  As such, specifying
$\distortionFn(\privDataVal, \distortPubDataVal) = \infty$ for any
$\distortPubDataVal < \objSize{\objStore{\privDataVal}}$ and any
$\distortPubDataVal > \padFactor \times
\objSize{\objStore{\privDataVal}}$ suffices to enforce
\eqnsref{eqn:noCompression}{eqn:limitExpansion}.  Additionally
specifying $\distortionFn(\privDataVal, \distortPubDataVal) = 0$ for
$\objSize{\objStore{\privDataVal}} \le \distortPubDataVal \le
\padFactor \times \objSize{\objStore{\privDataVal}}$ then provides
maximum flexibility in choosing \distortPubDataAlg{\cdot} to minimize
\mutInfo{\privDataRV}{\distortPubDataRV}.

Having reduced our problem to rate-distortion minimization, we can
employ a known algorithm for that problem to solve ours.  Typically
this problem must be solved numerically because the optimization
problem is nonlinear.  An iterative algorithm that converges to the
optimal solution is due to Blahut~\cite{blahut:1972:capacity} and
independently Arimoto~\cite{arimoto:1972:capacity} (see
also~\cite{yeung:2008:coding}).  For our purposes, this algorithm
works by iteratively computing values $\{\langle
\blahutProb{\blahutIdx}[\distortPubDataVal],
\blahutCProb{\blahutIdx}[\distortPubDataVal][\privDataVal]\rangle\}_{\blahutIdx
  \ge 0}$ as follows for each $\privDataVal \in \privDataDomain$ and
each \distortPubDataVal, $\objSize{\objStore{\privDataVal}} \le
\distortPubDataVal \le \padFactor \times
\objSize{\objStore{\privDataVal}}$, where \blahutParam is a parameter:
\begin{align}
  \blahutCProb{\blahutIdx+1}[\distortPubDataVal][\privDataVal]
  & = \frac{\blahutProb{\blahutIdx}[\distortPubDataVal]\exp(-\blahutParam \times \distortionFn(\privDataVal, \distortPubDataVal))}{\sum_{\distortPubDataValAlt} \blahutProb{\blahutIdx}[\distortPubDataValAlt]\exp(-\blahutParam \times \distortionFn(\privDataVal, \distortPubDataValAlt))} \nonumber \\
  & = \left\{\begin{array}{cl}
  \frac{\blahutProb{\blahutIdx}[\distortPubDataVal]}{\sum_{\distortPubDataValAlt : \distortionFn(\privDataVal, \distortPubDataValAlt) = 0} \blahutProb{\blahutIdx}[\distortPubDataValAlt]} & \mbox{if $\distortionFn(\privDataVal, \distortPubDataVal) = 0$} \\
  0 & \mbox{otherwise}
  \end{array}\right. \label{eqn:blahutSecond} \\
  \blahutProb{\blahutIdx+1}[\distortPubDataVal]
  & = \sum_{\privDataVal\in\privDataDomain} \prob{\privDataRV = \privDataVal} \blahutCProb{\blahutIdx+1}[\distortPubDataVal][\privDataVal] \nonumber
\end{align}
\eqnref{eqn:blahutSecond} holds since in our case,
$\distortionFn(\privDataVal, \distortPubDataVal) \in \{0, \infty\}$
for all \privDataVal and \distortPubDataVal.

As \blahutIdx grows, the values
\blahutProb{\blahutIdx}[\distortPubDataVal] and
\blahutCProb{\blahutIdx}[\distortPubDataVal][\privDataVal] converge to
\prob{\distortPubDataRV = \distortPubDataVal} and
\cprob{\big}{\distortPubDataRV = \distortPubDataVal}{\privDataRV =
  \privDataVal}, respectively, for the optimal
\distortPubDataAlg{\cdot} for \perReqPaddingTerm, provided that
\blahutCProb{0}[\distortPubDataVal][\privDataVal] is initialized to be
nonzero for any \distortPubDataVal for which
\cprob{\big}{\distortPubDataRV = \distortPubDataVal}{\privDataRV =
  \privDataVal} in the privacy-optimal \distortPubDataAlg{\cdot} is
nonzero.  So, in all empirical tests reported in this paper, we
initialized \blahutCProb{0}[\distortPubDataVal][\privDataVal] by
computing every possible padded size for \objStore{\privDataVal},
i.e.,
\[
\distortPubDataSubset{\privDataVal} =
\left\{ \distortPubDataValAlt
\left\lvert\begin{array}{l}
\objSize{\objStore{\privDataVal}} \le \distortPubDataValAlt \le \padFactor \times \objSize{\objStore{\privDataVal}}~\wedge \\ 
\exists \privDataValAlt \in \privDataDomain :
(\prob{\privDataRV = \privDataValAlt} > 0
\wedge \objSize{\objStore{\privDataValAlt}} =
\distortPubDataValAlt)
\end{array}\right.\right\}
\]
and set
\begin{align*}
\blahutCProb{0}[\distortPubDataVal][\privDataVal] & =
\left\{\begin{array}{ll}
\frac{1}{\setSize{\distortPubDataSubset{\privDataVal}}}
& \mbox{if $\distortPubDataVal \in \distortPubDataSubset{\privDataVal}$} \\
0
& \mbox{otherwise}
\end{array}\right. \\
\blahutProb{0}[\distortPubDataVal] & = \sum_{\privDataVal \in
  \privDataDomain} \prob{\privDataRV = \privDataVal} \blahutCProb{0}[\distortPubDataVal][\privDataVal]
\end{align*}
We terminated the algorithm once
$\mutInfoEst{\blahutIdx}{\privDataRV}{\distortPubDataRV} -
\mutInfoEst{\blahutIdx+1}{\privDataRV}{\distortPubDataRV} <
\blahutCutoff$, where
\mutInfoEst{\blahutIdx}{\privDataRV}{\distortPubDataRV} is the
value of \mutInfo{\privDataRV}{\distortPubDataRV} obtained using
$\prob{\distortPubDataRV = \distortPubDataVal} =
\blahutProb{\blahutIdx}[\distortPubDataVal]$ and
$\cprob{\big}{\distortPubDataRV = \distortPubDataVal}{\privDataRV =
  \privDataVal} =
\blahutCProb{\blahutIdx}[\distortPubDataVal][\privDataVal]$ in
\eqnref{eqn:mutInfoSymmetric}, and where \blahutCutoff is a
threshold to indicate when the algorithm can cease iterating
(as the per-iteration improvement to \mutInfo{\privDataRV}{\distortPubDataRV}
has become sufficiently small). In the sections that follow, we refer
to this algorithm as \perReqAlgLong (\perReqAlg).

\begin{figure}[t]
  \setcounter{MinNumber}{0}%
  \setcounter{MaxNumber}{1}%
  \iffalse
  \begin{subfigure}[t]{\columnwidth}
    \centering
    %\hspace{-.25em}
    \input{figures/algorithm-demos/prp/c_1.tex}
    \caption{$\padFactor = 1$}
    \label{fig:prp-demo:padFactor_1}
  \end{subfigure}
  \fi
  
  \begin{subfigure}[t]{\columnwidth}
    \centering
    %\hspace{-.25em}
    \input{figures/algorithm-demos/prp/c_2.tex}
    \caption{$\padFactor = 2$}
    \label{fig:prp-demo:padFactor_2}
  \end{subfigure}
  
  \begin{subfigure}[t]{\columnwidth}
    \centering
    %\hspace{-.25em}
    \input{figures/algorithm-demos/prp/c_3.tex}
    \caption{$\padFactor = 3$}
    \label{fig:prp-demo:padFactor_3}
  \end{subfigure}
  \vspace*{-1.0em}
  \caption{\cprob{\big}{\distortPubDataRV = \distortPubDataVal}{\privDataRV = \privDataVal} produced by \perReqAlg (\secref{sec:algs:per-request}).}
  \label{fig:prp-demo}
\end{figure}

\perReqAlg produces the conditional distributions shown in
\figref{fig:prp-demo} when applied to the objects in
\tabref{tbl:objects}.  The most immediately noticeable difference from
\figref{fig:pop-demo} is that \figref{fig:prp-demo} includes values $<
1.0$, which permits \distortPubDataAlg{\cdot} to sample each response
size from the distribution in the row of the requested object.  (Note
that each row sums to $1.0$, i.e. the calculation of
\blahutCProb{\blahutIdx}[\distortPubDataVal][\privDataVal]
followed by \blahutProb{\blahutIdx}[\distortPubDataVal],
as per \eqnref{eqn:blahutSecond}, yields a valid 
probability distribution.)  Whereas \perObjAlg was forced to isolate
either \exampleObject{0} or \exampleObject{2} when $\padFactor \le 3$,
\perReqAlg spreads probability mass to ensure that no object is
isolated; i.e., for any feasible response size \distortPubDataVal,
there are at least two objects that response could possibly represent.
Moreover, \perReqAlg spreads probability mass in the most effective
way possible to minimize \mutInfo{\privDataRV}{\distortPubDataRV}.

\subsection{Unknown Query Distribution}
\label{sec:algs:unknown-dist}

There are cases in which it is inconvenient or even impossible for the
\objStoreTerm to estimate the distribution for \privDataRV.  This
might occur because the adversary can affect that distribution or
simply because the distribution is likely to vary in unpredictable
ways over time.  Issa et al.~\cite{issa:2020:operational} advocate
that the Sibson mutual information of order infinity, denoted
\mutInfoInf{\privDataRV}{\distortPubDataRV}, be used in place of
\mutInfo{\privDataRV}{\distortPubDataRV} to describe information
leakage when the distribution of the private random variable
\privDataRV is complex or outside the defender's control, where
\[
\mutInfoInf{\privDataRV}{\distortPubDataRV} =
\log_2 \sum_{\distortPubDataVal} ~\max_{\privDataVal : \prob{\privDataRV = \privDataVal} > 0} \cprob{\big}{\distortPubDataRV = \distortPubDataVal}{\privDataRV = \privDataVal}
\]
This alternative is attractive because
$\mutInfoInf{\privDataRV}{\distortPubDataRV} \ge
\mutInfo{\privDataRV}{\distortPubDataRV}$~\cite{sibson:1969:radius,
  verdu:2015:alpha} and, moreover,
$\mutInfoInf{\privDataRV}{\distortPubDataRV} =
\mutInfo{\privDataRV}{\distortPubDataRV}$ for distributions of
\privDataRV and \distortPubDataRV meeting certain
conditions~\cite[\staticlmaref{2}]{issa:2020:operational}.  As such,
\mutInfoInf{\privDataRV}{\distortPubDataRV} is a conservative estimate
of information leakage in the absence of knowledge of the distribution
for \privDataRV (except those \privDataVal for which
$\prob{\privDataRV = \privDataVal} > 0$) and, Issa et al.\ argue, has
other advantages as a measure of information leakage, as well.

In our scenario, constructing \distortPubDataAlg{\cdot} that minimizes
\mutInfoInf{\privDataRV}{\distortPubDataRV} subject to overhead
constraints \eqnsref{eqn:noCompression}{eqn:limitExpansion} is very
efficient.  First, note that for any \privDataValAlt for which
$\prob{\privDataRV = \privDataValAlt} > 0$, constraints
\eqnsref{eqn:noCompression}{eqn:limitExpansion} imply
\[
\sum_{\distortPubDataVal : \objSize{\objStore{\privDataValAlt}} \le \distortPubDataVal \le \padFactor \times \objSize{\objStore{\privDataValAlt}}} \cprob{\big}{\distortPubDataRV = \distortPubDataVal}{\privDataRV = \privDataValAlt} = 1
\]
and so
\[
\sum_{\distortPubDataVal : \objSize{\objStore{\privDataValAlt}} \le \distortPubDataVal \le \padFactor \times \objSize{\objStore{\privDataValAlt}}} ~\max_{\privDataVal : \prob{\privDataRV = \privDataVal} > 0} \cprob{\big}{\distortPubDataRV = \distortPubDataVal}{\privDataRV = \privDataVal} \ge 1
\]
As such, for any $\privDataSubset \subseteq \privDataDomain$ such that
\begin{align}
&\forall \privDataVal, \privDataValAlt \in \privDataSubset,
\privDataVal \neq \privDataValAlt: \nonumber \\
&~~~\nats[\objSize{\objStore{\privDataVal}}][\padFactor \times
  \objSize{\objStore{\privDataVal}}] \cap
\nats[\objSize{\objStore{\privDataValAlt}}][\padFactor \times
  \objSize{\objStore{\privDataValAlt}}] = \emptyset
\label{eqn:no-overlap}
\end{align}
we know that $\mutInfoInf{\privDataRV}{\distortPubDataRV} \ge
\log_2 \setSize{\privDataSubset}$.

%In particular, if we find the largest such
%\privDataSubset and construct \distortPrivDataAlg{\cdot} such that \eqnref{}
%is an equality, then \distortPrivDataAlg{\cdot} is privacy optimal

Now, for any $\privDataSubset \subseteq \privDataDomain$ define the
\textit{anchors} \anchors[\privDataSubset] inductively as follows:
$\anchors[\emptyset] = \emptyset$ and, for $\privDataSubset \neq
\emptyset$,
\begin{align}
  \anchors[\privDataSubset] & =
  \{\anchor[\privDataSubset]\} \cup \anchors[\privDataSubset \setminus \anchorReach[\privDataSubset]] & \mbox{where} \label{eqn:anchors} \\
  \anchor[\privDataSubset] & = \arg\max_{\privDataVal \in \privDataSubset} \objSize{\objStore{\privDataVal}} & \mbox{and} \label{eqn:anchor} \\
  \anchorReach[\privDataSubset] & = \{\privDataValAlt \in \privDataSubset :  \objSize{\objStore{\anchor[\privDataSubset]}} \le \padFactor \times \objSize{\objStore{\privDataValAlt}}\} \label{eqn:reach}
\end{align}
In words, the anchors of \privDataSubset include the largest object of
\privDataSubset (see \eqnref{eqn:anchor}) and the anchors of the set
remaining after removing those objects for which this anchor is no
more than \padFactor times larger (see
\eqnref{eqn:anchors}, \eqnref{eqn:reach}).

Consider calculating the anchors \anchors[\privDataDomain] of the full
set \privDataDomain, and for each \anchorReach[\privDataSubset]
calculated in the induction, pad each object in
\anchorReach[\privDataSubset] to the size of the anchor
\anchor[\privDataSubset]; i.e.,
\[
\distortPubDataAlg{\objStore{\privDataVal}} = \objSize{\objStore{\anchor[\privDataSubset]}}
\]
for each $\privDataVal \in \anchorReach(\privDataSubset)$.  This
padding scheme respects constraints
\eqnsref{eqn:noCompression}{eqn:limitExpansion} and yields
$\mutInfoInf{\privDataRV}{\distortPubDataRV} = \log_2
\setSize{\anchors[\privDataDomain]}$, since there are only
\setSize{\anchors[\privDataDomain]} values \distortPubDataVal to which
objects are padded and
\[
\max_{\privDataVal : \prob{\privDataRV =
    \privDataVal} > 0} \cprob{\big}{\distortPubDataRV =
  \distortPubDataVal}{\privDataRV = \privDataVal} = 1
\]
for each such \distortPubDataVal.  Moreover, since
\eqnref{eqn:no-overlap} is satisfied with $\privDataSubset =
\anchors[\privDataDomain]$, we know that this value of
\mutInfoInf{\privDataRV}{\distortPubDataRV} is the minimum.  This
algorithm, which produces \distortPubDataAlg{\cdot} in time linear in
\setSize{\privDataDomain}, is denoted \noDistAlgLong (\noDistAlg) in
the following sections.

\begin{figure}[t]
  \setcounter{MinNumber}{0}%
  \setcounter{MaxNumber}{1}%
  \iffalse
  \begin{subfigure}[t]{\columnwidth}
    \centering
    %\hspace{-.25em}
    \input{figures/algorithm-demos/pwod/c_1.tex}
    \caption{$\padFactor = 1$}
    \label{fig:pwod-demo:padFactor_1}
  \end{subfigure}
  \fi
  
  \begin{subfigure}[t]{\columnwidth}
    \centering
    %\hspace{-.25em}
    \input{figures/algorithm-demos/pwod/c_2.tex}
    \caption{$\padFactor = 2$}
    \label{fig:pwod-demo:padFactor_2}
  \end{subfigure}
  
  \begin{subfigure}[t]{\columnwidth}
    \centering
    %\hspace{-.25em}
    \input{figures/algorithm-demos/pwod/c_3.tex}
    \caption{$\padFactor = 3$}
    \label{fig:pwod-demo:padFactor_3}
  \end{subfigure}
  \vspace*{-1.0em}
  \caption{\cprob{\big}{\distortPubDataRV = \distortPubDataVal}{\privDataRV = \privDataVal} produced by \noDistAlg (\secref{sec:algs:unknown-dist}).}
  \label{fig:pwod-demo}
\end{figure}

Applying \noDistAlg to the objects in \tabref{tbl:objects} yields the
conditional distributions shown in \figref{fig:pwod-demo}.  As in the
case of \perObjAlg, \noDistAlg is forced to isolate either
\exampleObject{0} or \exampleObject{2} when $\padFactor \le 3$ (since
both produce deterministic padding schemes), and \noDistAlg does so
in the same way as \perObjAlg.  The distributions resulting when
$\padFactor = 3$, shown in \figref{fig:pwod-demo:padFactor_3}, are
also identical to those produced by \perObjAlg
(\figref{fig:pop-demo:padFactor_3}).  When $\padFactor = 2$, however,
\noDistAlg prescribes that \exampleObject{3}--\exampleObject{9} be
padded differently than \perObjAlg prescribes, as can be seen by
comparing \figref{fig:pwod-demo:padFactor_2} to
\figref{fig:pop-demo:padFactor_2}. This is due to the fact that
\noDistAlg iterates through an object list in reverse-order by size
and greedily assigns objects to the current anchor.
Thus, in \figref{fig:pwod-demo:padFactor_2} \noDistAlg sets \exampleObject{9}
as the first anchor and then pads \exampleObject{8} and \exampleObject{7}
to \exampleObject{9}'s size before arriving at \exampleObject{6} and
creating the next anchor.

%% file: figures/algorithm-demos/pop/c_2.tex
{\scriptsize
\setlength{\tabcolsep}{.125em}
\begin{tabular}{@{}cc|@{\hspace{.16667em}}|@{\hspace{1.5pt}}*{10}{Y}@{\hspace{1.5pt}}}
& & \multicolumn{10}{c}{\distortPubDataVal} \\
& \privDataVal ~~
& \multicolumn{1}{c}{\tiny 2493855}
& \multicolumn{1}{c}{\tiny 3833489}
& \multicolumn{1}{c}{\tiny 7929946}
& \multicolumn{1}{c}{\tiny 13322074}
& \multicolumn{1}{c}{\tiny 13589747}
& \multicolumn{1}{c}{\tiny 16235142}
& \multicolumn{1}{c}{\tiny 16719886}
& \multicolumn{1}{c}{\tiny 19437984}
& \multicolumn{1}{c}{\tiny 25905442}
& \multicolumn{1}{c}{\tiny 34389677}\\
\cline{2-12} \\[-2.35ex]
& \exampleObject{0} & 0.00 & 1.00 & 0.00 & 0.00 & 0.00 & 0.00 & 0.00 & 0.00 & 0.00 & 0.00 \\
& \exampleObject{1} & 0.00 & 1.00 & 0.00 & 0.00 & 0.00 & 0.00 & 0.00 & 0.00 & 0.00 & 0.00 \\
& \exampleObject{2} & 0.00 & 0.00 & 1.00 & 0.00 & 0.00 & 0.00 & 0.00 & 0.00 & 0.00 & 0.00 \\
& \exampleObject{3} & 0.00 & 0.00 & 0.00 & 0.00 & 0.00 & 0.00 & 0.00 & 0.00 & 1.00 & 0.00 \\
& \exampleObject{4} & 0.00 & 0.00 & 0.00 & 0.00 & 0.00 & 0.00 & 0.00 & 0.00 & 1.00 & 0.00 \\
& \exampleObject{5} & 0.00 & 0.00 & 0.00 & 0.00 & 0.00 & 0.00 & 0.00 & 0.00 & 1.00 & 0.00 \\
& \exampleObject{6} & 0.00 & 0.00 & 0.00 & 0.00 & 0.00 & 0.00 & 0.00 & 0.00 & 1.00 & 0.00 \\
& \exampleObject{7} & 0.00 & 0.00 & 0.00 & 0.00 & 0.00 & 0.00 & 0.00 & 0.00 & 1.00 & 0.00 \\
& \exampleObject{8} & 0.00 & 0.00 & 0.00 & 0.00 & 0.00 & 0.00 & 0.00 & 0.00 & 1.00 & 0.00 \\
& \exampleObject{9} & 0.00 & 0.00 & 0.00 & 0.00 & 0.00 & 0.00 & 0.00 & 0.00 & 0.00 & 1.00 \\
\end{tabular}
}

%% file: figures/algorithm-demos/pop/c_3.tex
{\scriptsize
\setlength{\tabcolsep}{.125em}
\begin{tabular}{@{}cc|@{\hspace{.16667em}}|@{\hspace{1.5pt}}*{10}{Y}@{\hspace{1.5pt}}}
& & \multicolumn{10}{c}{\distortPubDataVal} \\
& \privDataVal ~~
& \multicolumn{1}{c}{\tiny 2493855}
& \multicolumn{1}{c}{\tiny 3833489}
& \multicolumn{1}{c}{\tiny 7929946}
& \multicolumn{1}{c}{\tiny 13322074}
& \multicolumn{1}{c}{\tiny 13589747}
& \multicolumn{1}{c}{\tiny 16235142}
& \multicolumn{1}{c}{\tiny 16719886}
& \multicolumn{1}{c}{\tiny 19437984}
& \multicolumn{1}{c}{\tiny 25905442}
& \multicolumn{1}{c}{\tiny 34389677}\\
\cline{2-12} \\[-2.35ex]
& \exampleObject{0} & 1.00 & 0.00 & 0.00 & 0.00 & 0.00 & 0.00 & 0.00 & 0.00 & 0.00 & 0.00 \\
& \exampleObject{1} & 0.00 & 0.00 & 1.00 & 0.00 & 0.00 & 0.00 & 0.00 & 0.00 & 0.00 & 0.00 \\
& \exampleObject{2} & 0.00 & 0.00 & 1.00 & 0.00 & 0.00 & 0.00 & 0.00 & 0.00 & 0.00 & 0.00 \\
& \exampleObject{3} & 0.00 & 0.00 & 0.00 & 0.00 & 0.00 & 0.00 & 0.00 & 0.00 & 0.00 & 1.00 \\
& \exampleObject{4} & 0.00 & 0.00 & 0.00 & 0.00 & 0.00 & 0.00 & 0.00 & 0.00 & 0.00 & 1.00 \\
& \exampleObject{5} & 0.00 & 0.00 & 0.00 & 0.00 & 0.00 & 0.00 & 0.00 & 0.00 & 0.00 & 1.00 \\
& \exampleObject{6} & 0.00 & 0.00 & 0.00 & 0.00 & 0.00 & 0.00 & 0.00 & 0.00 & 0.00 & 1.00 \\
& \exampleObject{7} & 0.00 & 0.00 & 0.00 & 0.00 & 0.00 & 0.00 & 0.00 & 0.00 & 0.00 & 1.00 \\
& \exampleObject{8} & 0.00 & 0.00 & 0.00 & 0.00 & 0.00 & 0.00 & 0.00 & 0.00 & 0.00 & 1.00 \\
& \exampleObject{9} & 0.00 & 0.00 & 0.00 & 0.00 & 0.00 & 0.00 & 0.00 & 0.00 & 0.00 & 1.00 \\
\end{tabular}
}

%% file: figures/algorithm-demos/prp/c_2.tex
{\scriptsize
\setlength{\tabcolsep}{.125em}
\begin{tabular}{@{}cc|@{\hspace{.16667em}}|@{\hspace{1.5pt}}*{10}{Y}@{\hspace{1.5pt}}}
& & \multicolumn{10}{c}{\distortPubDataVal} \\
& \privDataVal ~~
& \multicolumn{1}{c}{\tiny 2493855}
& \multicolumn{1}{c}{\tiny 3833489}
& \multicolumn{1}{c}{\tiny 7929946}
& \multicolumn{1}{c}{\tiny 13322074}
& \multicolumn{1}{c}{\tiny 13589747}
& \multicolumn{1}{c}{\tiny 16235142}
& \multicolumn{1}{c}{\tiny 16719886}
& \multicolumn{1}{c}{\tiny 19437984}
& \multicolumn{1}{c}{\tiny 25905442}
& \multicolumn{1}{c}{\tiny 34389677}\\
\cline{2-12} \\[-2.35ex]
& \exampleObject{0} & 0.00 & 1.00 & 0.00 & 0.00 & 0.00 & 0.00 & 0.00 & 0.00 & 0.00 & 0.00 \\
& \exampleObject{1} & 0.00 & 1.00 & 0.00 & 0.00 & 0.00 & 0.00 & 0.00 & 0.00 & 0.00 & 0.00 \\
& \exampleObject{2} & 0.00 & 0.00 & 0.00 & 0.00 & 1.00 & 0.00 & 0.00 & 0.00 & 0.00 & 0.00 \\
& \exampleObject{3} & 0.00 & 0.00 & 0.00 & 0.00 & 0.19 & 0.00 & 0.00 & 0.00 & 0.81 & 0.00 \\
& \exampleObject{4} & 0.00 & 0.00 & 0.00 & 0.00 & 0.19 & 0.00 & 0.00 & 0.00 & 0.81 & 0.00 \\
& \exampleObject{5} & 0.00 & 0.00 & 0.00 & 0.00 & 0.00 & 0.00 & 0.00 & 0.00 & 1.00 & 0.00 \\
& \exampleObject{6} & 0.00 & 0.00 & 0.00 & 0.00 & 0.00 & 0.00 & 0.00 & 0.00 & 1.00 & 0.00 \\
& \exampleObject{7} & 0.00 & 0.00 & 0.00 & 0.00 & 0.00 & 0.00 & 0.00 & 0.00 & 0.86 & 0.14 \\
& \exampleObject{8} & 0.00 & 0.00 & 0.00 & 0.00 & 0.00 & 0.00 & 0.00 & 0.00 & 0.86 & 0.14 \\
& \exampleObject{9} & 0.00 & 0.00 & 0.00 & 0.00 & 0.00 & 0.00 & 0.00 & 0.00 & 0.00 & 1.00 \\
\end{tabular}
}

%% file: figures/algorithm-demos/prp/c_3.tex
{\scriptsize
\setlength{\tabcolsep}{.125em}
\begin{tabular}{@{}cc|@{\hspace{.16667em}}|@{\hspace{1.5pt}}*{10}{Y}@{\hspace{1.5pt}}}
& & \multicolumn{10}{c}{\distortPubDataVal} \\
& \privDataVal ~~
& \multicolumn{1}{c}{\tiny 2493855}
& \multicolumn{1}{c}{\tiny 3833489}
& \multicolumn{1}{c}{\tiny 7929946}
& \multicolumn{1}{c}{\tiny 13322074}
& \multicolumn{1}{c}{\tiny 13589747}
& \multicolumn{1}{c}{\tiny 16235142}
& \multicolumn{1}{c}{\tiny 16719886}
& \multicolumn{1}{c}{\tiny 19437984}
& \multicolumn{1}{c}{\tiny 25905442}
& \multicolumn{1}{c}{\tiny 34389677}\\
\cline{2-12} \\[-2.35ex]
& \exampleObject{0} & 0.00 & 1.00 & 0.00 & 0.00 & 0.00 & 0.00 & 0.00 & 0.00 & 0.00 & 0.00 \\
& \exampleObject{1} & 0.00 & 0.39 & 0.61 & 0.00 & 0.00 & 0.00 & 0.00 & 0.00 & 0.00 & 0.00 \\
& \exampleObject{2} & 0.00 & 0.00 & 0.72 & 0.00 & 0.00 & 0.00 & 0.00 & 0.28 & 0.00 & 0.00 \\
& \exampleObject{3} & 0.00 & 0.00 & 0.00 & 0.00 & 0.00 & 0.00 & 0.00 & 0.19 & 0.00 & 0.81 \\
& \exampleObject{4} & 0.00 & 0.00 & 0.00 & 0.00 & 0.00 & 0.00 & 0.00 & 0.19 & 0.00 & 0.81 \\
& \exampleObject{5} & 0.00 & 0.00 & 0.00 & 0.00 & 0.00 & 0.00 & 0.00 & 0.19 & 0.00 & 0.81 \\
& \exampleObject{6} & 0.00 & 0.00 & 0.00 & 0.00 & 0.00 & 0.00 & 0.00 & 0.19 & 0.00 & 0.81 \\
& \exampleObject{7} & 0.00 & 0.00 & 0.00 & 0.00 & 0.00 & 0.00 & 0.00 & 0.19 & 0.00 & 0.81 \\
& \exampleObject{8} & 0.00 & 0.00 & 0.00 & 0.00 & 0.00 & 0.00 & 0.00 & 0.00 & 0.00 & 1.00 \\
& \exampleObject{9} & 0.00 & 0.00 & 0.00 & 0.00 & 0.00 & 0.00 & 0.00 & 0.00 & 0.00 & 1.00 \\
\end{tabular}
}

%% file: figures/algorithm-demos/pwod/c_2.tex
{\scriptsize
\setlength{\tabcolsep}{.125em}
\begin{tabular}{@{}cc|@{\hspace{.16667em}}|@{\hspace{1.5pt}}*{10}{Y}@{\hspace{1.5pt}}}
& & \multicolumn{10}{c}{\distortPubDataVal} \\
& \privDataVal ~~
& \multicolumn{1}{c}{\tiny 2493855}
& \multicolumn{1}{c}{\tiny 3833489}
& \multicolumn{1}{c}{\tiny 7929946}
& \multicolumn{1}{c}{\tiny 13322074}
& \multicolumn{1}{c}{\tiny 13589747}
& \multicolumn{1}{c}{\tiny 16235142}
& \multicolumn{1}{c}{\tiny 16719886}
& \multicolumn{1}{c}{\tiny 19437984}
& \multicolumn{1}{c}{\tiny 25905442}
& \multicolumn{1}{c}{\tiny 34389677}\\
\cline{2-12} \\[-2.35ex]
& \exampleObject{0} & 0.00 & 1.00 & 0.00 & 0.00 & 0.00 & 0.00 & 0.00 & 0.00 & 0.00 & 0.00 \\
& \exampleObject{1} & 0.00 & 1.00 & 0.00 & 0.00 & 0.00 & 0.00 & 0.00 & 0.00 & 0.00 & 0.00 \\
& \exampleObject{2} & 0.00 & 0.00 & 1.00 & 0.00 & 0.00 & 0.00 & 0.00 & 0.00 & 0.00 & 0.00 \\
& \exampleObject{3} & 0.00 & 0.00 & 0.00 & 0.00 & 0.00 & 0.00 & 1.00 & 0.00 & 0.00 & 0.00 \\
& \exampleObject{4} & 0.00 & 0.00 & 0.00 & 0.00 & 0.00 & 0.00 & 1.00 & 0.00 & 0.00 & 0.00 \\
& \exampleObject{5} & 0.00 & 0.00 & 0.00 & 0.00 & 0.00 & 0.00 & 1.00 & 0.00 & 0.00 & 0.00 \\
& \exampleObject{6} & 0.00 & 0.00 & 0.00 & 0.00 & 0.00 & 0.00 & 1.00 & 0.00 & 0.00 & 0.00 \\
& \exampleObject{7} & 0.00 & 0.00 & 0.00 & 0.00 & 0.00 & 0.00 & 0.00 & 0.00 & 0.00 & 1.00 \\
& \exampleObject{8} & 0.00 & 0.00 & 0.00 & 0.00 & 0.00 & 0.00 & 0.00 & 0.00 & 0.00 & 1.00 \\
& \exampleObject{9} & 0.00 & 0.00 & 0.00 & 0.00 & 0.00 & 0.00 & 0.00 & 0.00 & 0.00 & 1.00 \\
\end{tabular}
}

%% file: figures/algorithm-demos/pwod/c_3.tex
{\scriptsize
\setlength{\tabcolsep}{.125em}
\begin{tabular}{@{}cc|@{\hspace{.16667em}}|@{\hspace{1.5pt}}*{10}{Y}@{\hspace{1.5pt}}}
& & \multicolumn{10}{c}{\distortPubDataVal} \\
& \privDataVal ~~
& \multicolumn{1}{c}{\tiny 2493855}
& \multicolumn{1}{c}{\tiny 3833489}
& \multicolumn{1}{c}{\tiny 7929946}
& \multicolumn{1}{c}{\tiny 13322074}
& \multicolumn{1}{c}{\tiny 13589747}
& \multicolumn{1}{c}{\tiny 16235142}
& \multicolumn{1}{c}{\tiny 16719886}
& \multicolumn{1}{c}{\tiny 19437984}
& \multicolumn{1}{c}{\tiny 25905442}
& \multicolumn{1}{c}{\tiny 34389677}\\
\cline{2-12} \\[-2.35ex]
& \exampleObject{0} & 1.00 & 0.00 & 0.00 & 0.00 & 0.00 & 0.00 & 0.00 & 0.00 & 0.00 & 0.00 \\
& \exampleObject{1} & 0.00 & 0.00 & 1.00 & 0.00 & 0.00 & 0.00 & 0.00 & 0.00 & 0.00 & 0.00 \\
& \exampleObject{2} & 0.00 & 0.00 & 1.00 & 0.00 & 0.00 & 0.00 & 0.00 & 0.00 & 0.00 & 0.00 \\
& \exampleObject{3} & 0.00 & 0.00 & 0.00 & 0.00 & 0.00 & 0.00 & 0.00 & 0.00 & 0.00 & 1.00 \\
& \exampleObject{4} & 0.00 & 0.00 & 0.00 & 0.00 & 0.00 & 0.00 & 0.00 & 0.00 & 0.00 & 1.00 \\
& \exampleObject{5} & 0.00 & 0.00 & 0.00 & 0.00 & 0.00 & 0.00 & 0.00 & 0.00 & 0.00 & 1.00 \\
& \exampleObject{6} & 0.00 & 0.00 & 0.00 & 0.00 & 0.00 & 0.00 & 0.00 & 0.00 & 0.00 & 1.00 \\
& \exampleObject{7} & 0.00 & 0.00 & 0.00 & 0.00 & 0.00 & 0.00 & 0.00 & 0.00 & 0.00 & 1.00 \\
& \exampleObject{8} & 0.00 & 0.00 & 0.00 & 0.00 & 0.00 & 0.00 & 0.00 & 0.00 & 0.00 & 1.00 \\
& \exampleObject{9} & 0.00 & 0.00 & 0.00 & 0.00 & 0.00 & 0.00 & 0.00 & 0.00 & 0.00 & 1.00 \\
\end{tabular}
}

%% file: security.tex
\section{Security Evaluation}
\label{sec:security}

In this section, we begin by describing the datasets that we created to
support all of our security and performance tests. Next we list the
padding algorithms to which we compare, and we describe how they work
when used in our setting. We then evaluate each algorithm's ability
to reduce \mutInfo{\privDataRV}{\distortPubDataRV} for the two datasets.
We conclude this section with a security assessment that evaluates
each algorithm's ability to hinder a network attacker in an operational 
setting.

\subsection{Datasets}
\label{sec:security:datasets}

For our tests we use two datasets: a dataset consisting of
NodeJS\footnote{https://nodejs.org/en/} packages, and a dataset consisting of
Unsplash\footnote{https://unsplash.com/} images. Details for each dataset are
as follows.

\paragraph{NodeJS Packages}
To create our dataset of NodeJS packages, we first referenced a 
list\footnote{https://github.com/nice-registry/all-the-package-repos, 
accessed Feb. 19, 2021.} of all packages available on the NodeJS Package Manager 
(NPM) registry. We then used NPM's Application Programming Interface
(API) to retrieve the weekly download statistics of each package, for the week of 
Feb. 13-19, 2021. Finally, we issued HTTP HEAD requests for each package to
obtain the tarball size (in bytes) of its most current version (as of Feb. 19, 2021).
In total, our NodeJS dataset contains the name, tarball size, and weekly
download statistics for 423,450 NodeJS packages\footnote{Due to limitations of the 
NPM API, we did not retrieve download statistics for scoped packages (i.e. packages
whose name begins with an @ symbol). Additionally, we did not include packages with 0 
weekly downloads as part of our dataset. Overall, our starting list contained
993,825 packages, and our final dataset used for testing is comprised of 423,450
packages.}.

\paragraph{Unsplash Lite}
To create our Unsplash Lite dataset, we first referenced Unsplash's freely available
``Unsplash Lite Dataset 1.1.0'' dataset\footnote{https://unsplash.com/data}. This 
dataset includes the URL of each image in the dataset, as well each image's cumulative
downloads since the image was uploaded to Unsplash. Given this information, we were 
able to issue HTTP HEAD requests for each image in the dataset, as well as compute 
each image's average downloads per day (based on the difference between the dataset's
creation date and the image's upload date). In total, our Unsplash Lite dataset
contains the URL, file size (in bytes), and average daily downloads for
24,997 images.

\subsection{Padding Algorithms to Which We Compare}
\label{sec:security:algos}

In this section we briefly introduce the other padding algorithms to
which we compare.  These algorithms have appeared within the past five
years in well-regarded, peer-reviewed venues focusing on privacy
technologies, and so we take them as representative of modern
approaches for padding objects to make their retrieval harder to
detect by a network observer.

\paragraph{\dAlpaca}
Cherubin et al.~\cite{cherubin:2017:alpaca} proposed padding algorithms
to defend against website fingerprinting attacks and, so, that seek to
address leakage resulting from the retrieval of objects hyperlinked in
a webpage, subject to padding overhead constraints.  Distilled down to
our simpler scenario, though, their padding algorithms result in natural
contenders for comparison.  One of these, called \dAlpaca, is
deterministic and so is suitable as a \perObjPaddingTerm algorithm.  In
brief, \dAlpaca sets \distortPubDataAlg{\objStore{\privDataVal}}
to be the smallest multiple of \dAlpacaBinSize that is 
$\ge \objSize{\objStore{\privDataVal}}$, where \dAlpacaBinSize is
an input parameter. For our setting, we set $\dAlpacaBinSize = 
\floor[(\padFactor - 1) \times \objSize{\objStore{\privDataValMin}}]$,
where $\floor: \reals \rightarrow \nats$ is the floor function and 
where \privDataValMin is the identifier of the smallest object in 
the \objStoreTerm. This, then, ensures that \dAlpaca does not violate
\padFactor for any $\privDataVal \in \privDataDomain$. Note that
\dAlpaca is insensitive to the distribution of \privDataRV.

\begin{figure}[t]
  \setcounter{MinNumber}{0}%
  \setcounter{MaxNumber}{1}%
  \iffalse
  \begin{subfigure}[t]{\columnwidth}
    \centering
    %\hspace{-.25em}
    \input{figures/algorithm-demos/d-alpaca/c_1.tex}
    \caption{$\padFactor = 1$}
    \label{fig:d-alpaca-demo:c_1}
  \end{subfigure}
  \fi
  
  \begin{subfigure}[t]{\columnwidth}
    \centering
    %\hspace{-.25em}
    \input{figures/algorithm-demos/d-alpaca/c_2.tex}
    \caption{$\padFactor = 2$}
    \label{fig:d-alpaca-demo:padFactor_2}
  \end{subfigure}
	
  \begin{subfigure}[t]{\columnwidth}
    \centering
    %\hspace{-.25em}
    \input{figures/algorithm-demos/d-alpaca/c_3.tex}
    \caption{$\padFactor = 3$}
    \label{fig:d-alpaca-demo:padFactor_3}
  \end{subfigure}
  \vspace*{-1.0em}
  \caption{\cprob{\big}{\distortPubDataRV = \distortPubDataVal}{\privDataRV = \privDataVal} produced by \dAlpaca~\cite{cherubin:2017:alpaca}.}
  \label{fig:d-alpaca-demo}
\end{figure}

\figref{fig:d-alpaca-demo} illustrates \cprob{\big}{\distortPubDataRV
  = \distortPubDataVal}{\privDataRV = \privDataVal} when \dAlpaca is
applied to the objects listed in \tabref{tbl:objects}.  The values of
\distortPubDataVal (i.e., the column headings) listed in
\figref{fig:d-alpaca-demo} differ from those in
\figsref{fig:pop-demo}{fig:pwod-demo} because \dAlpaca pads objects to
sizes that are not necessarily the sizes of other objects, unlike
\perObjAlg, \perReqAlg, and \noDistAlg.
\figref{fig:d-alpaca-demo:padFactor_2} illustrates that \dAlpaca is
not particularly effective for the objects listed in
\tabref{tbl:objects} when $\padFactor = 2$, as retrievals of six of
the ten objects can be identified immediately based on the sizes to
which they are padded.  Only once $\padFactor = 3$ is a majority of
those objects padded to sizes where multiple objects share the same
padded size (\figref{fig:d-alpaca-demo:padFactor_3}).

\paragraph{\padme}
Nikitin et al.~\cite{nikitin:2019:purbs} proposed a \perObjPaddingTerm
algorithm called \padme.  Like ours, \padme seeks to limit leakage about
which object \objStore{\privDataVal} is returned from the size
\distortPubDataAlg{\objStore{\privDataVal}} of the returned object,
while also limiting padding overhead.  The notion of leakage that
\padme is designed to limit, however, is simply the total number of
valid padding sizes over all objects; in particular, \padme calculates
\perObjPaddingTerm sizes independent of the distribution of
\privDataRV.  Specifically, \padme pads an object
\objStore{\privDataVal} to size
\[
\distortPubDataAlg{\objStore{\privDataVal}} =
\padmeBase - (\padmeBase \bmod 2^{\padmeObjSizeExponent -
  \padmeObjSizeExponentBits})
\]
where $\padmeObjSizeExponent = \floor[\log_2
  \objSize{\objStore{\privDataVal}}]$, $\padmeObjSizeExponentBits =
\floor[\log_2 \padmeObjSizeExponent] + 1$, and $\padmeBase =
\objSize{\objStore{\privDataVal}} + 2^{\padmeObjSizeExponent -
  \padmeObjSizeExponentBits} - 1$.  Nikitin et al.\ show that this
scheme limits ``leakage'' to $\bigO{\log \log \padmeObjSizeMax}$ bits
if the largest object is of size \padmeObjSizeMax, 
with a padding overhead of $\padFactor \approx
\frac{1}{2\log_2 \objSize{\objStore{\privDataVal}}} < 1.12$.

\begin{figure}[t]
  \setcounter{MinNumber}{0}%
  \setcounter{MaxNumber}{1}%
  \centering
  %\hspace{-.25em}
  \input{figures/algorithm-demos/padme/padme.tex}
  \caption{\cprob{\big}{\distortPubDataRV = \distortPubDataVal}{\privDataRV = \privDataVal} produced by \padme~\cite{nikitin:2019:purbs}.}
  \label{fig:padme-demo}
\end{figure}

Unlike the other algorithms considered here, \padme is not tunable to
different padding factors \padFactor.  Moreover, as shown in
\figref{fig:padme-demo}, \padme is unfortunately ineffective when
applied to the objects listed in \tabref{tbl:objects}---each of the
objects is padded to a different size.  In part this is due to the
small number of objects or, more precisely, the low density of object
sizes.  Our evaluation on larger datasets below will provide a better
view of where \padme falls in the security versus overhead tradeoff.

\paragraph{\pAlpaca}
Cherubin et al.~\cite{cherubin:2017:alpaca} also proposed a randomized
algorithm called \pAlpaca that is suitable as a \perReqPaddingTerm 
algorithm.  When applied to our setting, \pAlpaca pads the response 
to a request \privDataVal so that
\begin{align*}
\prob{\distortPubDataAlg{\objStore{\privDataVal}} = \distortPubDataVal}
& = \cprob{\big}{\objSize{\objStore{\privDataRV}} = \distortPubDataVal}{\objSize{\objStore{\privDataVal}} \le \objSize{\objStore{\privDataRV}} \le \padFactor \times \objSize{\objStore{\privDataVal}}} \\
& = \frac{\sum_{\privDataValAlt: \objSize{\objStore{\privDataValAlt}} = \distortPubDataVal} \prob{\privDataRV = \privDataValAlt}}{\sum_{\privDataValAlt: \objSize{\objStore{\privDataVal}} \le \objSize{\objStore{\privDataValAlt}} \le \padFactor \times \objSize{\objStore{\privDataVal}}} \prob{\privDataRV = \privDataValAlt}}
\end{align*}
for each \distortPubDataVal, $\objSize{\objStore{\privDataVal}} \le
\distortPubDataVal \le \padFactor \times
\objSize{\objStore{\privDataVal}}$.  In particular, the most probable
padded size for \objStore{\privDataVal} is the most probable unpadded size in
the interval $[\objSize{\objStore{\privDataVal}}, \padFactor \times
\objSize{\objStore{\privDataVal}}]$.
%probable object \objStore{\privDataValAlt} satisfying
%$\objSize{\objStore{\privDataVal}} \le
%\objSize{\objStore{\privDataValAlt}} \le \padFactor \times
%\objSize{\objStore{\privDataVal}}$.

\begin{figure}[t]
  \setcounter{MinNumber}{0}%
  \setcounter{MaxNumber}{1}%
  \iffalse
  \begin{subfigure}[t]{\columnwidth}
    \centering
    %\hspace{-.25em}
    \input{figures/algorithm-demos/p-alpaca/c_1.tex}
    \caption{$\padFactor = 1$}
    \label{fig:p-alpaca-demo:padFactor_1}
  \end{subfigure}
  \fi
  
  \begin{subfigure}[t]{\columnwidth}
    \centering
    %\hspace{-.25em}
    \input{figures/algorithm-demos/p-alpaca/c_2.tex}
    \caption{$\padFactor = 2$}
    \label{fig:p-alpaca-demo:padFactor_2}
  \end{subfigure}
  
  \begin{subfigure}[t]{\columnwidth}
    \centering
    %\hspace{-.25em}
    \input{figures/algorithm-demos/p-alpaca/c_3.tex}
    \caption{$\padFactor = 3$}
    \label{fig:p-alpaca-demo:padFactor_3}
  \end{subfigure}
  \vspace*{-1.0em}
  \caption{\cprob{\big}{\distortPubDataRV = \distortPubDataVal}{\privDataRV = \privDataVal} produced by \pAlpaca~\cite{cherubin:2017:alpaca}.}
  \label{fig:p-alpaca-demo}
\end{figure}

The effect of \pAlpaca on the objects in \tabref{tbl:objects} is shown
in \figref{fig:p-alpaca-demo}.  Among all of the algorithms we
consider, \pAlpaca permits the widest variety of possible padding
sizes for each object, for the objects in \tabref{tbl:objects}: e.g.,
when $\padFactor = 3$, only the largest object can be padded to only a
single size, and some objects can be padded to up to six distinct
sizes.  As we will see below, however, many possible padded sizes per
object does not necessarily equate to better security.

\subsection{Mutual Information}
\label{sec:security:mutual_information}

To more systematically compare the security offered by the six padding
algorithms we consider, we applied them to the two datasets described
in \secref{sec:security:datasets} and computed the mutual information
\mutInfo{\privDataRV}{\distortPubDataRV} of the distribution resulting
from each.  \figref{fig:mutual-information:nodejs} shows the results
for the NodeJS dataset, and \figref{fig:mutual-information:unsplash}
shows the results for the Unsplash dataset.  Each bar indicates the
value for \mutInfo{\privDataRV}{\distortPubDataRV} that was achieved
by each algorithm, when calculated using the distribution of the given
dataset and the value of \padFactor on the horizontal axis.  Each
error bar extends to \mutInfoInf{\privDataRV}{\distortPubDataRV} for
each algorithm's padding scheme, which upper-bounds the worst-case
\mutInfo{\privDataRV}{\distortPubDataRV} if the distribution for
\privDataRV the defender assumed was incorrect.  The algorithms are
ordered left-to-right roughly according to security, i.e., from lower
security to higher (since lower numbers indicate better security).
The one exception is \padme, for which a bar is added to the left of
the cluster at the value of \padFactor that its padding scheme achieved 
on the corresponding dataset ($\padFactor = 1.09$ in
\figref{fig:mutual-information:nodejs} and $\padFactor = 1.03$ in
\figref{fig:mutual-information:unsplash}). Note that the y-axis in both
\figref{fig:mutual-information:nodejs} and 
\figref{fig:mutual-information:unsplash} begins at 5 \bits.

%%%%%%%%%% Mutual Information for NodeJS starts %%%%%%%%%%%%%
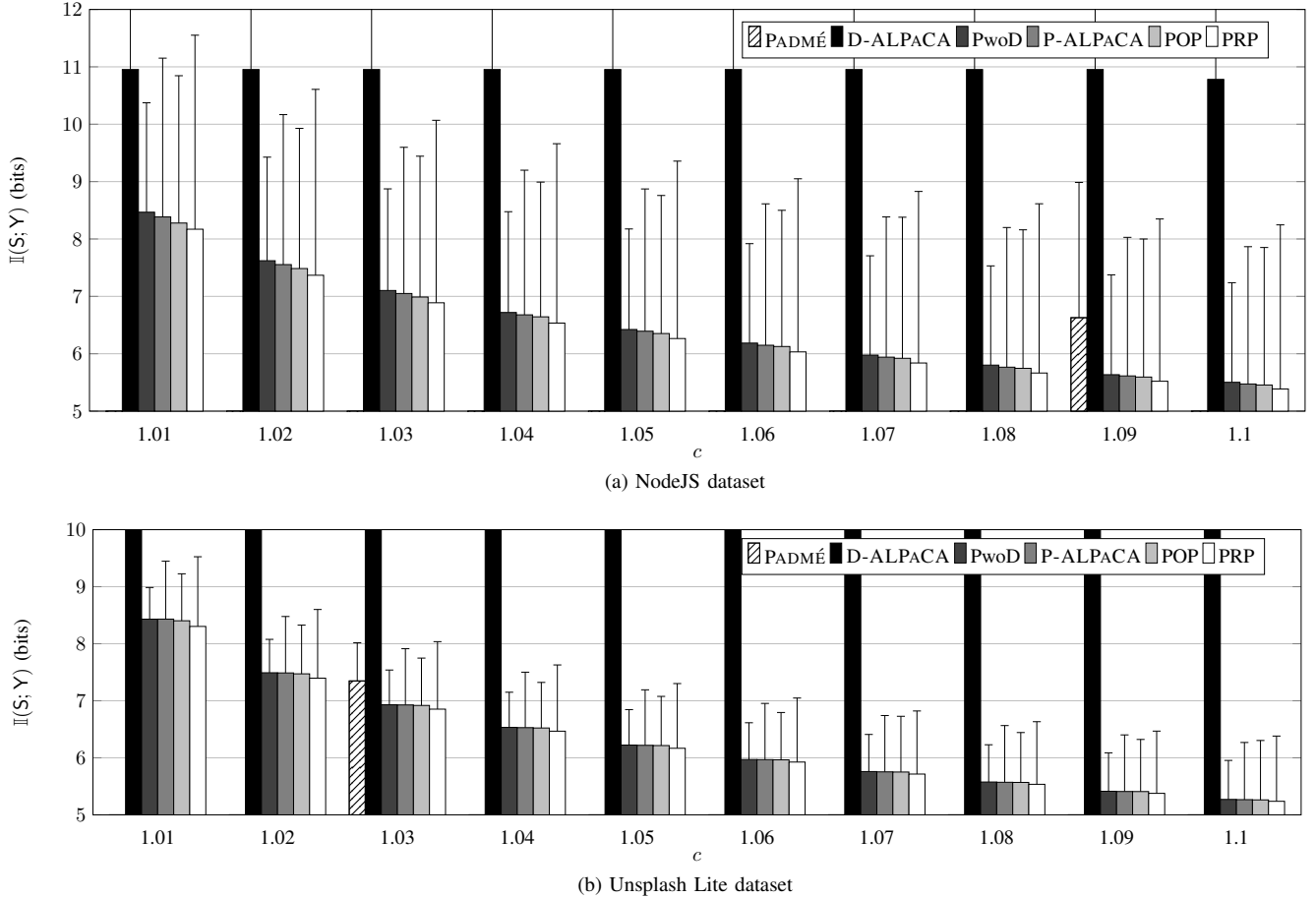
\begin{figure*}[t]
  \begin{subfigure}[t]{\textwidth}
    \centering
    \hspace*{-2.5em}
    \resizebox{!}{18.0em}{\input{figures/mutual-information-nodejs/mutual-information-nodejs.tex}}
    
    %%%%%%%% NodeJS Mutual Information main caption %%%%%%%
    \vspace*{0em} % caption vertical position adjustment
    \caption{\textmd{NodeJS dataset}}
    \label{fig:mutual-information:nodejs}
  \end{subfigure}
  %%%%%%%%%% Mutual Information for NodeJS ends %%%%%%%%%%%%%
\\
  %%%%%%%%%% Mutual Information for Unsplash starts %%%%%%%%%%%%%
  \begin{subfigure}[t]{\textwidth}
    \centering
    \hspace*{-2.5em}
    \resizebox{!}{13.5em}{\input{figures/mutual-information-unsplash/mutual-information-unsplash.tex}}	
	
    %%%%%%%% Unsplash Mutual Information main caption %%%%%%%
    \vspace*{0em} % caption vertical position adjustment
    \caption{\textmd{Unsplash Lite dataset}}
    \label{fig:mutual-information:unsplash}
  \end{subfigure}
  \caption{Per-algorithm mutual information.   Error bars extend to \mutInfoInf{\privDataRV}{\distortPubDataRV}.  Lower values indicate better security.}
  \label{fig:mutual-information}
\end{figure*}
%%%%%%%%%% Mutual Information for Unsplash ends %%%%%%%%%%%%%

As these figures show, \dAlpaca was not competitive with the other
algorithms in our tests.  Indeed, it performed so poorly that 
portions of its bars were clipped in both
\figref{fig:mutual-information:nodejs} and 
\figref{fig:mutual-information:unsplash} so that they would not
distort the graph so as to obscure the differences among the bars for
other algorithms.  \padme performed convincingly better than \dAlpaca
(for the \padFactor values it yielded on the two datasets) but was
nevertheless inferior to other alternatives in terms of mutual
information of the distribution it produced.

Of the remaining algorithms, \perReqAlg consistently produced the
lowest \mutInfo{\privDataRV}{\distortPubDataRV}, which is not
surprising since it is designed specifically to minimize
\mutInfo{\privDataRV}{\distortPubDataRV}.  (So is \perObjAlg, but
\perObjAlg is constrained to produce a padding scheme where only one
padded size is possible for each object, since
\distortPubDataAlg{\cdot} is constrained to be deterministic.)
However, \perReqAlg also produced the highest
\mutInfoInf{\privDataRV}{\distortPubDataRV} (of these four
algorithms), which suggests that by fine-tuning the padding scheme
to the assumed distribution for \privDataRV, there is a risk of making
security more sensitive to errors in that assumption.  On the other
end of the spectrum, \noDistAlg achieves the best
\mutInfoInf{\privDataRV}{\distortPubDataRV} in all cases; again, this
is unsurprising since it was designed to optimize this measure.  And,
while \mutInfo{\privDataRV}{\distortPubDataRV} of the distribution it
produces is generally the worst of these four algorithms, it does not
differ by much.  For this reason, plus the fact that it does not rely
on knowing the distribution for \privDataRV, \noDistAlg appears to be
an attractive choice for these datasets.

A final observation from these graphs is that \perObjAlg strictly
dominated \pAlpaca in these tests, producing a lower
\mutInfoInf{\privDataRV}{\distortPubDataRV} in all but one case
($\padFactor = 1.1$ in \figref{fig:mutual-information:unsplash}) and a
\mutInfo{\privDataRV}{\distortPubDataRV} value that is no higher than
(and often lower than) \pAlpaca's.  This dominance is perhaps
unexpected, since \perObjAlg is restricted to produce a deterministic
padding scheme \distortPubDataAlg{\cdot} whereas that produced by
\pAlpaca need not be.

\subsection{Attacker's Recall and Precision}
\label{sec:security:recall-preciaion}

One of the criticisms sometimes levied against mutual information as a
measure of privacy is that it is difficult to interpret
operationally~\cite{issa:2020:operational}.  In this section,
therefore, we evaluate the extent to which the various algorithms we
consider interfere with a network attacker attempting to achieve a
specific operational goal, namely to infer whether a request
\privDataVal is a member of a target set $\privDataSubset \subseteq
\privDataDomain$, based on the size of the object returned in response
to the request.  Specifically, for the observed size
\distortPubDataVal, the adversary returns \boolTrue iff
$\cprob{\big}{\privDataRV \in
  \privDataSubset}{\distortPubDataAlg{\objStore{\privDataRV}} =
  \distortPubDataVal} \ge \adversaryConfidence$ for a tunable
parameter $0 \le \adversaryConfidence \le 1$.  For a given setting of
\adversaryConfidence, the adversary's \textit{recall} is the fraction
of requests for some member of \privDataSubset for which the adversary
returns \boolTrue, and the adversary's \textit{precision} is the
fraction of \boolTrue declarations by the adversary for which the
requested object is in \privDataSubset.  We consider two different
sets \privDataSubset:

\paragraph{Vulnerable NodeJS Packages}
For the NodeJS packages, we envision an adversary that wants to detect
if \textit{any} vulnerable package was retrieved/installed by a
server.  If so, then the adversary might then hope to exploit the
server via a program such as Metasploit\footnote{https://www.metasploit.com/}. 
To define a subset \privDataSubset of ``vulnerable'' packages, we use 
a dataset provided by Ferenc et al. \cite{ferenc:2019:vulnerablejs}. 
This dataset lists 88 NodeJS packages that, at the time of its publication, 
contained vulnerable functions and whose source code was available 
on GitHub. Note that these packages may no longer be vulnerable.  
Still, we use this dataset as a way to define a subset \privDataSubset 
of ``vulnerable'' packages based on their vulnerability at that time.

\paragraph{Unsplash Lite \textit{Nature} Collection}
The Unsplash Lite dataset includes a table that contains user-created
collections of photos contained within the dataset. We used a
collection of 256 photos named \textit{Nature} to create the
adversary's subset \privDataSubset of interest. We envision that such
an adversary might want to send the user targeted ads, in which case
knowledge of a downloaded picture might help the adversary to tailor
the selection of ads shown to the user.

\bigskip

We allow the adversary to know the object sizes and retrieval
distribution for \privDataRV, as well as the padding algorithm
employed by the object store.  For confidence thresholds
$\adversaryConfidence \in \{0.0, 0.05, 0.10, \dots, 1.0\}$, we
calculate this adversary's recall and precision for each algorithm
using $\padFactor \in \{1.01, 1.03, 1.05, 1.07, 1.09\}$.  The
recall-precision curves that result are depicted in
\figref{fig:recall-precision-nodejs} and
\figref{fig:recall-precision-unsplash}.

%%%%%%%%%% Recall-Precision for NodeJS starts %%%%%%%%%%%%%
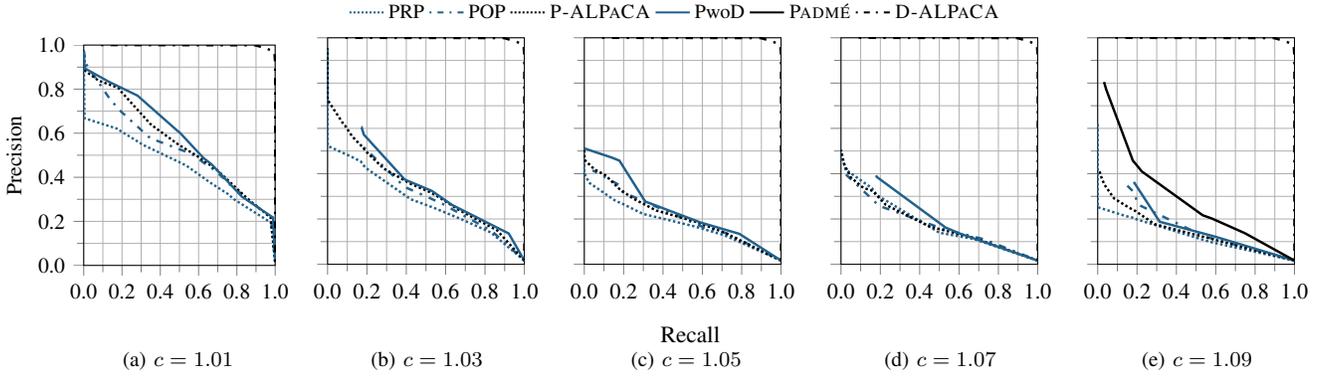
\begin{figure*}[t]
	
	\vspace*{-0.3em}
	%%%%%%%% NodeJS Recall-Precision legend %%%%%%%
	\begin{subfigure}[b]{.1\columnwidth}
		\setlength\figureheight{2in}
		\centering
		\hspace*{12.5em}
		\resizebox{!}{1.3em}{\input{figures/recall-precision-nodejs/recall-precision-nodejs_legend.tex}}
	\end{subfigure}
	
	%\vspace*{0.25em}
	%%%%%%%% NodeJS Recall-Precision %%%%%%%
	%\hspace*{-1.0em}
	%\captionsetup[subfigure]{font=normalsize,labelfont=normalsize, oneside,margin={-0.5in,0in}}
	\begin{subfigure}[b]{0.191\textwidth}
		\setlength\figureheight{1.9in}
		\centering
		\resizebox{!}{10.5em}{\input{figures/recall-precision-nodejs/recall-precision-nodejs_c_1-01.tex}}
		\hspace*{-1.2em}
		\begin{minipage}[t]{15.5em}
			\vspace*{-0.3em}\caption{$\padFactor = 1.01$}
			\label{fig:recall-precision-nodejs:c_1-01}
		\end{minipage}%
	\end{subfigure}
	%
	%\captionsetup[subfigure]{font=normalsize,labelfont=normalsize, oneside,margin={-0.8in,0in}}
	\hspace*{0.5em}
	\begin{subfigure}[b]{0.191\textwidth}
		\setlength\figureheight{1.9in}
		\centering
		\resizebox{!}{10.5em}{\input{figures/recall-precision-nodejs/recall-precision-nodejs_c_1-03.tex}}
		\hspace*{-3.0em}
		\begin{minipage}[t]{15.5em}
			\vspace*{-0.3em}\caption{$\padFactor = 1.03$}
			\label{fig:recall-precision-nodejs:c_1-03}
		\end{minipage}%
	\end{subfigure}
	\hspace*{-1.0em}
	\begin{subfigure}[b]{0.191\textwidth}
		\setlength\figureheight{1.9in}
		\centering
		\resizebox{!}{10.5em}{\input{figures/recall-precision-nodejs/recall-precision-nodejs_c_1-05.tex}}
		\hspace*{-2.9em}
		\begin{minipage}[t]{15.5em}
			\vspace*{-0.3em}\caption{$\padFactor = 1.05$}
			\label{fig:recall-precision-nodejs:c_1-05}
		\end{minipage}%
	\end{subfigure}
	\hspace*{-1.0em}
	\begin{subfigure}[b]{0.191\textwidth}
		\setlength\figureheight{1.9in}
		\centering
		\resizebox{!}{10.5em}{\input{figures/recall-precision-nodejs/recall-precision-nodejs_c_1-07.tex}}
		\hspace*{-3.0em}
		\begin{minipage}[t]{15.5em}
			\vspace*{-0.3em}\caption{$\padFactor = 1.07$}
			\label{fig:recall-precision-nodejs:c_1-07}
		\end{minipage}%
	\end{subfigure}
	\hspace*{-1.0em}
	\begin{subfigure}[b]{0.191\textwidth}
		\setlength\figureheight{1.9in}
		\centering
		\resizebox{!}{10.5em}{\input{figures/recall-precision-nodejs/recall-precision-nodejs_c_1-09.tex}}
		\hspace*{-3.0em}
		\begin{minipage}[t]{15.5em}
			\vspace*{-0.3em}\caption{$\padFactor = 1.09$}
			\label{fig:recall-precision-nodejs:c_1-09}
		\end{minipage}%
	\end{subfigure}
	
	\vspace*{-2.4em}
	\begin{subfigure}[b]{.43\columnwidth}
		\setlength\figureheight{2in}
		\begin{minipage}[b]{1\textwidth}
			\centering
			\hspace*{24.5em}
			\resizebox{!}{1.2em}{\input{figures/recall-precision-nodejs/recall-precision_xlabel.tex}}\vspace*{-0.6em}
		\end{minipage}
	\end{subfigure}
	%%%%%%%% NodeJS Recall-Precision main caption %%%%%%%
	\vspace*{1.0em} % caption vertical position adjustment
	\caption{Adversary's recall and precision for detecting vulnerable NodeJS packages.}
	\label{fig:recall-precision-nodejs}
\end{figure*}
%%%%%%%%%% Recall-Precision for NodeJS ends %%%%%%%%%%%%%

%%%%%%%%%% Recall-Precision for Unsplash starts %%%%%%%%%%%%%
\begin{figure*}[t]
	
	\vspace*{-0.3em}
	%%%%%%%% Unsplash Recall-Precision legend %%%%%%%
	\begin{subfigure}[b]{.1\columnwidth}
		\setlength\figureheight{2in}
		\centering
		\hspace*{12.5em}
		\resizebox{!}{1.3em}{\input{figures/recall-precision-unsplash/recall-precision-unsplash_legend.tex}}
	\end{subfigure}
	
	%\vspace*{0.25em}
	%%%%%%%% Unsplash Recall-Precision %%%%%%%
	%\hspace*{-1.0em}
	%\captionsetup[subfigure]{font=normalsize,labelfont=normalsize, oneside,margin={-0.5in,0in}}
	\begin{subfigure}[b]{0.191\textwidth}
		\setlength\figureheight{1.9in}
		\centering
		\resizebox{!}{10.5em}{\input{figures/recall-precision-unsplash/recall-precision-unsplash_c_1-01.tex}}
		\hspace*{-1.2em}
		\begin{minipage}[t]{15.5em}
			\vspace*{-0.3em}\caption{$\padFactor = 1.01$}
			\label{fig:recall-precision-unsplash:c_1-01}
		\end{minipage}%
	\end{subfigure}
	%
	%\captionsetup[subfigure]{font=normalsize,labelfont=normalsize, oneside,margin={-0.8in,0in}}
	\hspace*{0.5em}
	\begin{subfigure}[b]{0.191\textwidth}
		\setlength\figureheight{1.9in}
		\centering
		\resizebox{!}{10.5em}{\input{figures/recall-precision-unsplash/recall-precision-unsplash_c_1-03.tex}}
		\hspace*{-3.0em}
		\begin{minipage}[t]{15.5em}
			\vspace*{-0.3em}\caption{$\padFactor = 1.03$}
			\label{fig:recall-precision-unsplash:c_1-03}
		\end{minipage}%
	\end{subfigure}
	\hspace*{-1.0em}
	\begin{subfigure}[b]{0.191\textwidth}
		\setlength\figureheight{1.9in}
		\centering
		\resizebox{!}{10.5em}{\input{figures/recall-precision-unsplash/recall-precision-unsplash_c_1-05.tex}}
		\hspace*{-2.9em}
		\begin{minipage}[t]{15.5em}
			\vspace*{-0.3em}\caption{$\padFactor = 1.05$}
			\label{fig:recall-precision-unsplash:c_1-05}
		\end{minipage}%
	\end{subfigure}
	\hspace*{-1.0em}
	\begin{subfigure}[b]{0.191\textwidth}
		\setlength\figureheight{1.9in}
		\centering
		\resizebox{!}{10.5em}{\input{figures/recall-precision-unsplash/recall-precision-unsplash_c_1-07.tex}}
		\hspace*{-3.0em}
		\begin{minipage}[t]{15.5em}
			\vspace*{-0.3em}\caption{$\padFactor = 1.07$}
			\label{fig:recall-precision-unsplash:c_1-07}
		\end{minipage}%
	\end{subfigure}
	\hspace*{-1.0em}
	\begin{subfigure}[b]{0.191\textwidth}
		\setlength\figureheight{1.9in}
		\centering
		\resizebox{!}{10.5em}{\input{figures/recall-precision-unsplash/recall-precision-unsplash_c_1-09.tex}}
		\hspace*{-3.0em}
		\begin{minipage}[t]{15.5em}
			\vspace*{-0.3em}\caption{$\padFactor = 1.09$}
			\label{fig:recall-precision-unsplash:c_1-09}
		\end{minipage}%
	\end{subfigure}
	
	\vspace*{-2.4em}
	\begin{subfigure}[b]{.43\columnwidth}
		\setlength\figureheight{2in}
		\begin{minipage}[b]{1\textwidth}
			\centering
			\hspace*{24.5em}
			\resizebox{!}{1.2em}{\input{figures/recall-precision-unsplash/recall-precision_xlabel.tex}}\vspace*{-0.6em}
		\end{minipage}
	\end{subfigure}
	%%%%%%%% Unsplash Recall-Precision main caption %%%%%%%
	\vspace*{1.0em} % caption vertical position adjustment
	\caption{Adversary's recall and precision for detecting the Unsplash Lite \textit{Nature} collection.}
	\label{fig:recall-precision-unsplash}
\end{figure*}
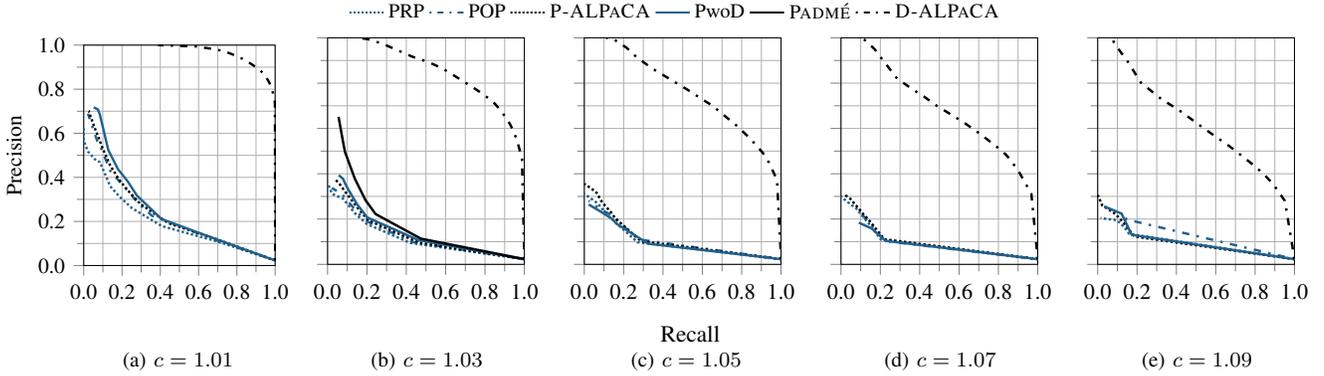
%%%%%%%%%% Recall-Precision for NodeJS ends %%%%%%%%%%%%%

Intuitively, a curve closer to the upper right-hand corner of each
plot indicates that the adversary did better at detecting requests in
\privDataSubset.  In that light, we see that better security provided
by an algorithm when measured using mutual information (i.e.,
\figref{fig:mutual-information}) translated reasonably well into
diminished precision and/or recall in these tests.  Indeed, our
\perReqPaddingTerm algorithm (\perReqAlg) defended as well or better
than the competitor for \perReqPaddingTerm (\pAlpaca), and our
\perObjPaddingTerm algorithm (\perObjAlg) consistently outperformed
the \perObjPaddingTerm competitors (\dAlpaca and \padme).  Only
\noDistAlg was outperformed by \pAlpaca in some cases, but \pAlpaca
did so by leveraging the distribution of \privDataRV; if that
distribution were unknown or incorrect, presumably the protection
offered by \pAlpaca would decay.

We also computed the precision-recall curve per vulnerable NodeJS
package and per \textit{Nature} photo in the Unsplash Lite dataset, to
assess the extent to which individual objects could still be
identified by the attacker based on their sizes when returned.  To
visualize these results, we reduced the precision-recall curve for
each $\privDataVal \in \privDataSubset$ to a single number---the area under the
curve\footnote{We extended each curve to the left to meet the vertical
  axis at its maximum precision.} (AUC)---and then plotted the
distribution of AUC values per algorithm.

\figref{fig:auc-prc-nodejs} shows distributions for the vulnerable
NodeJS packages, and \figref{fig:auc-prc-unsplash} shows distributions
for the \textit{Nature} photos.  Intuitively, the more bottom-heavy
the distribution, the more members of \privDataSubset the
algorithm protects.  As such, it is easy to see from these figures
that our algorithms performed better than \pAlpaca, particularly at
smaller values of \padFactor, i.e., $\padFactor = 1.01$
(\figref{fig:auc-prc-nodejs:padFactor_1-01}) and $\padFactor = 1.03$
(\figref{fig:auc-prc-nodejs:padFactor_1-03}) for the NodeJS dataset,
and $\padFactor = 1.01$ (\figref{fig:auc-prc-unsplash:padFactor_1-01})
for the Unsplash dataset.  \dAlpaca and, to a
lesser extent, \padme were not competitive with the other algorithms.

\begin{figure}[t]
	\centering
	
	\begin{subfigure}[t]{\columnwidth}
		\raisebox{3.75em}{\scalebox{0.75}{\input{figures/auc-prc-nodejs/auc_ylabel.tex}}}
		\hspace{-.25em}
		\scalebox{0.6}{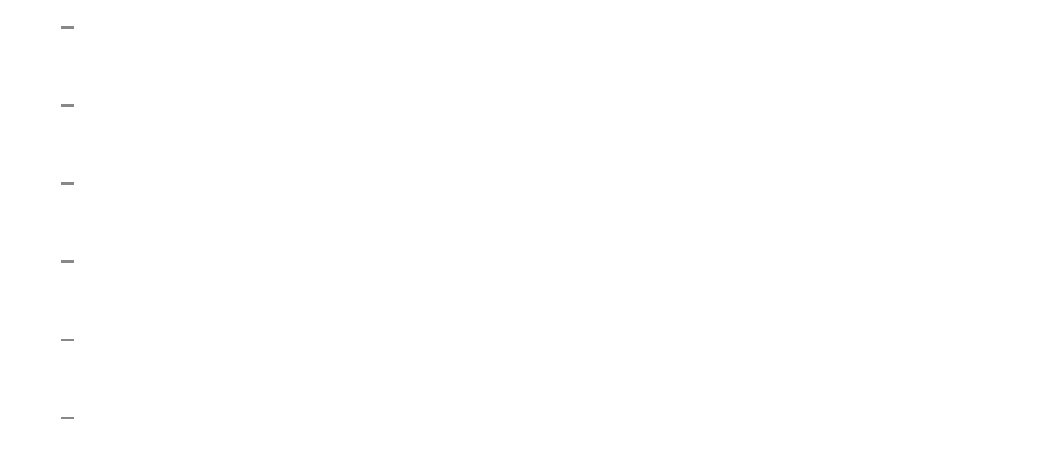}
		\caption{$\padFactor = 1.01$}
		\label{fig:auc-prc-nodejs:padFactor_1-01}
	\end{subfigure}
	
	\begin{subfigure}[t]{\columnwidth}
		\raisebox{3.75em}{\scalebox{0.75}{\input{figures/auc-prc-nodejs/auc_ylabel.tex}}}
		\hspace{-.25em}
		\scalebox{0.6}{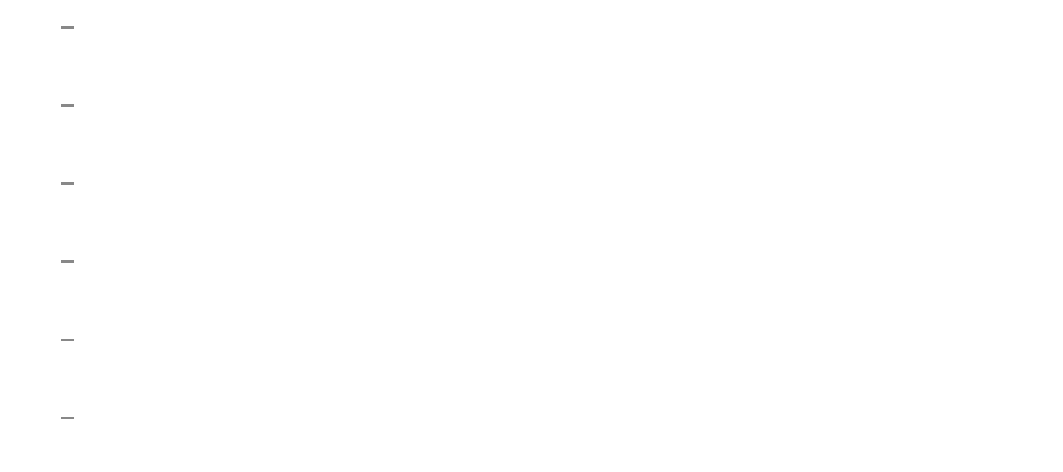}
		\caption{$\padFactor = 1.03$}
		\label{fig:auc-prc-nodejs:padFactor_1-03}
	\end{subfigure}

	\begin{subfigure}[t]{\columnwidth}
		\raisebox{3.75em}{\scalebox{0.75}{\input{figures/auc-prc-nodejs/auc_ylabel.tex}}}
		\hspace{-.25em}
		\scalebox{0.6}{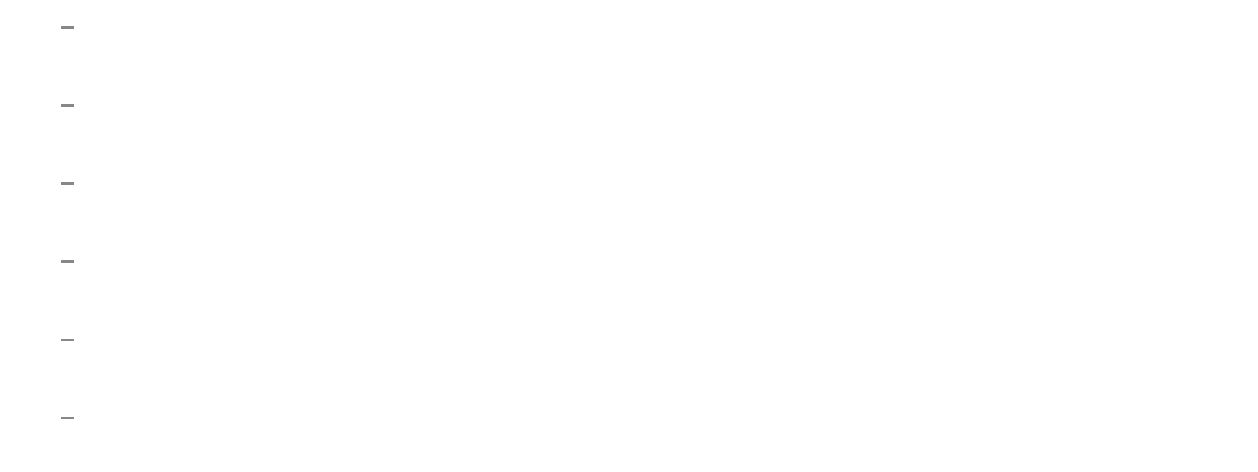}
		\caption{$\padFactor = 1.09$}
		\label{fig:auc-prc-nodejs:padFactor_1-09}
	\end{subfigure}

	\vspace*{-1.0em}
	\caption{Distributions of precision-recall AUC for vulnerable NodeJS packages.}
	\label{fig:auc-prc-nodejs}
\end{figure}

\begin{figure}[t]
	\centering
	
	\begin{subfigure}[t]{\columnwidth}
		\raisebox{3.75em}{\scalebox{0.75}{\input{figures/auc-prc-unsplash/auc_ylabel.tex}}}
		\hspace{-.25em}
		\scalebox{0.6}{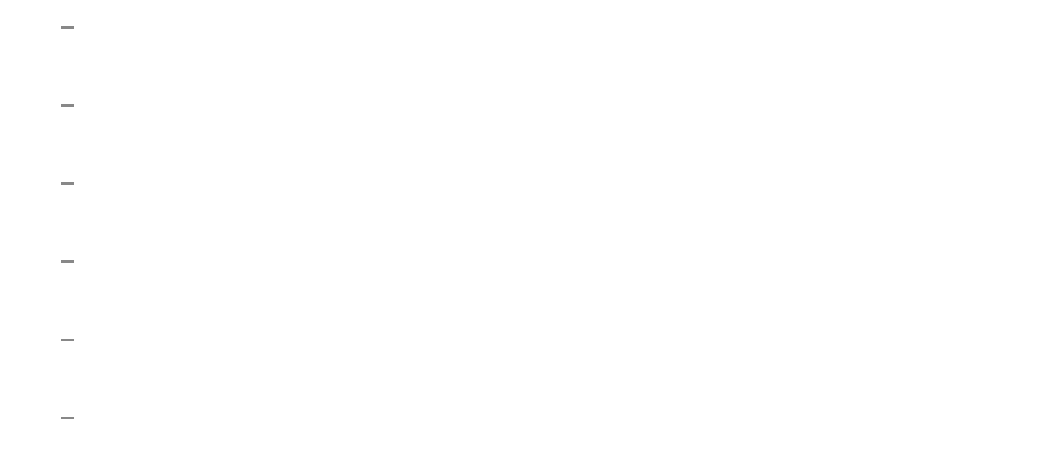}
		\caption{$\padFactor = 1.01$}
		\label{fig:auc-prc-unsplash:padFactor_1-01}
	\end{subfigure}
	
	\begin{subfigure}[t]{\columnwidth}
		\raisebox{3.75em}{\scalebox{0.75}{\input{figures/auc-prc-unsplash/auc_ylabel.tex}}}
		\hspace{-.25em}
		\scalebox{0.6}{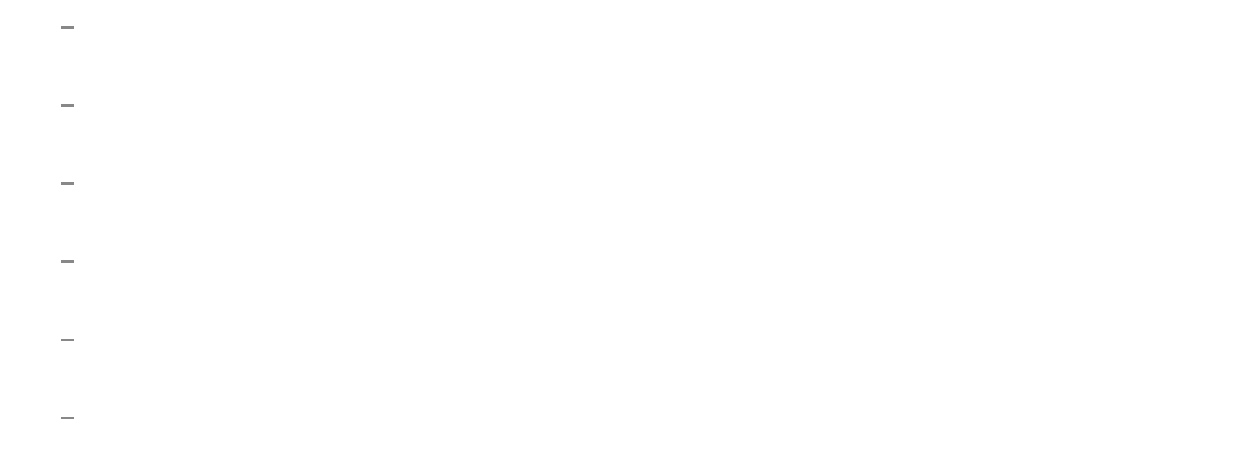}
		\caption{$\padFactor = 1.03$}
		\label{fig:auc-prc-unsplash:padFactor_1-03}
	\end{subfigure}
	
	\vspace*{-1.0em}
	\caption{Distributions of precision-recall AUC for Unsplash Lite \textit{Nature} photos.}
	\label{fig:auc-prc-unsplash}
\end{figure}

%% file: figures/algorithm-demos/d-alpaca/c_2.tex
{\scriptsize
\setlength{\tabcolsep}{.125em}
\begin{tabular}{@{}cc|@{\hspace{.16667em}}|@{\hspace{1.5pt}}*{8}{Y}@{\hspace{1.5pt}}}
& & \multicolumn{8}{c}{\distortPubDataVal} \\
& \privDataVal ~~
& \multicolumn{1}{c}{\tiny 2493855}
& \multicolumn{1}{c}{\tiny 4987710}
& \multicolumn{1}{c}{\tiny 9975420}
& \multicolumn{1}{c}{\tiny 14963130}
& \multicolumn{1}{c}{\tiny 17456985}
& \multicolumn{1}{c}{\tiny 19950840}
& \multicolumn{1}{c}{\tiny 27432405}
& \multicolumn{1}{c}{\tiny 34913970}\\
\cline{2-10} \\[-2.35ex]
& \exampleObject{0} & 1.00 & 0.00 & 0.00 & 0.00 & 0.00 & 0.00 & 0.00 & 0.00 \\
& \exampleObject{1} & 0.00 & 1.00 & 0.00 & 0.00 & 0.00 & 0.00 & 0.00 & 0.00 \\
& \exampleObject{2} & 0.00 & 0.00 & 1.00 & 0.00 & 0.00 & 0.00 & 0.00 & 0.00 \\
& \exampleObject{3} & 0.00 & 0.00 & 0.00 & 1.00 & 0.00 & 0.00 & 0.00 & 0.00 \\
& \exampleObject{4} & 0.00 & 0.00 & 0.00 & 1.00 & 0.00 & 0.00 & 0.00 & 0.00 \\
& \exampleObject{5} & 0.00 & 0.00 & 0.00 & 0.00 & 1.00 & 0.00 & 0.00 & 0.00 \\
& \exampleObject{6} & 0.00 & 0.00 & 0.00 & 0.00 & 1.00 & 0.00 & 0.00 & 0.00 \\
& \exampleObject{7} & 0.00 & 0.00 & 0.00 & 0.00 & 0.00 & 1.00 & 0.00 & 0.00 \\
& \exampleObject{8} & 0.00 & 0.00 & 0.00 & 0.00 & 0.00 & 0.00 & 1.00 & 0.00 \\
& \exampleObject{9} & 0.00 & 0.00 & 0.00 & 0.00 & 0.00 & 0.00 & 0.00 & 1.00 \\
\end{tabular}
}

%% file: figures/algorithm-demos/d-alpaca/c_3.tex
{\scriptsize
\setlength{\tabcolsep}{.125em}
\begin{tabular}{@{}cc|@{\hspace{.16667em}}|@{\hspace{1.5pt}}*{6}{Y}@{\hspace{1.5pt}}}
& & \multicolumn{6}{c}{\distortPubDataVal} \\
& \privDataVal ~~
& \multicolumn{1}{c}{\tiny 4987710}
& \multicolumn{1}{c}{\tiny 9975420}
& \multicolumn{1}{c}{\tiny 14963130}
& \multicolumn{1}{c}{\tiny 19950840}
& \multicolumn{1}{c}{\tiny 29926260}
& \multicolumn{1}{c}{\tiny 34913970}\\
\cline{2-8} \\[-2.35ex]
& \exampleObject{0} & 1.00 & 0.00 & 0.00 & 0.00 & 0.00 & 0.00 \\
& \exampleObject{1} & 1.00 & 0.00 & 0.00 & 0.00 & 0.00 & 0.00 \\
& \exampleObject{2} & 0.00 & 1.00 & 0.00 & 0.00 & 0.00 & 0.00 \\
& \exampleObject{3} & 0.00 & 0.00 & 1.00 & 0.00 & 0.00 & 0.00 \\
& \exampleObject{4} & 0.00 & 0.00 & 1.00 & 0.00 & 0.00 & 0.00 \\
& \exampleObject{5} & 0.00 & 0.00 & 0.00 & 1.00 & 0.00 & 0.00 \\
& \exampleObject{6} & 0.00 & 0.00 & 0.00 & 1.00 & 0.00 & 0.00 \\
& \exampleObject{7} & 0.00 & 0.00 & 0.00 & 1.00 & 0.00 & 0.00 \\
& \exampleObject{8} & 0.00 & 0.00 & 0.00 & 0.00 & 1.00 & 0.00 \\
& \exampleObject{9} & 0.00 & 0.00 & 0.00 & 0.00 & 0.00 & 1.00 \\
\end{tabular}
}

%% file: figures/algorithm-demos/padme/padme.tex
{\scriptsize
\setlength{\tabcolsep}{.125em}
\begin{tabular}{@{}cc|@{\hspace{.16667em}}|@{\hspace{1.5pt}}*{10}{Y}@{\hspace{1.5pt}}}
& & \multicolumn{10}{c}{\distortPubDataVal} \\
& \privDataVal ~~
& \multicolumn{1}{c}{\tiny 2555904}
& \multicolumn{1}{c}{\tiny 3866624}
& \multicolumn{1}{c}{\tiny 7995392}
& \multicolumn{1}{c}{\tiny 13369344}
& \multicolumn{1}{c}{\tiny 13631488}
& \multicolumn{1}{c}{\tiny 16252928}
& \multicolumn{1}{c}{\tiny 16777216}
& \multicolumn{1}{c}{\tiny 19922944}
& \multicolumn{1}{c}{\tiny 26214400}
& \multicolumn{1}{c}{\tiny 34603008}\\
\cline{2-12} \\[-2.35ex]
& \exampleObject{0} & 1.00 & 0.00 & 0.00 & 0.00 & 0.00 & 0.00 & 0.00 & 0.00 & 0.00 & 0.00 \\
& \exampleObject{1} & 0.00 & 1.00 & 0.00 & 0.00 & 0.00 & 0.00 & 0.00 & 0.00 & 0.00 & 0.00 \\
& \exampleObject{2} & 0.00 & 0.00 & 1.00 & 0.00 & 0.00 & 0.00 & 0.00 & 0.00 & 0.00 & 0.00 \\
& \exampleObject{3} & 0.00 & 0.00 & 0.00 & 1.00 & 0.00 & 0.00 & 0.00 & 0.00 & 0.00 & 0.00 \\
& \exampleObject{4} & 0.00 & 0.00 & 0.00 & 0.00 & 1.00 & 0.00 & 0.00 & 0.00 & 0.00 & 0.00 \\
& \exampleObject{5} & 0.00 & 0.00 & 0.00 & 0.00 & 0.00 & 1.00 & 0.00 & 0.00 & 0.00 & 0.00 \\
& \exampleObject{6} & 0.00 & 0.00 & 0.00 & 0.00 & 0.00 & 0.00 & 1.00 & 0.00 & 0.00 & 0.00 \\
& \exampleObject{7} & 0.00 & 0.00 & 0.00 & 0.00 & 0.00 & 0.00 & 0.00 & 1.00 & 0.00 & 0.00 \\
& \exampleObject{8} & 0.00 & 0.00 & 0.00 & 0.00 & 0.00 & 0.00 & 0.00 & 0.00 & 1.00 & 0.00 \\
& \exampleObject{9} & 0.00 & 0.00 & 0.00 & 0.00 & 0.00 & 0.00 & 0.00 & 0.00 & 0.00 & 1.00 \\
\end{tabular}
}

%% file: figures/algorithm-demos/p-alpaca/c_2.tex
{\scriptsize
\setlength{\tabcolsep}{.125em}
\begin{tabular}{@{}cc|@{\hspace{.16667em}}|@{\hspace{1.5pt}}*{10}{Y}@{\hspace{1.5pt}}}
& & \multicolumn{10}{c}{\distortPubDataVal} \\
& \privDataVal ~~
& \multicolumn{1}{c}{\tiny 2493855}
& \multicolumn{1}{c}{\tiny 3833489}
& \multicolumn{1}{c}{\tiny 7929946}
& \multicolumn{1}{c}{\tiny 13322074}
& \multicolumn{1}{c}{\tiny 13589747}
& \multicolumn{1}{c}{\tiny 16235142}
& \multicolumn{1}{c}{\tiny 16719886}
& \multicolumn{1}{c}{\tiny 19437984}
& \multicolumn{1}{c}{\tiny 25905442}
& \multicolumn{1}{c}{\tiny 34389677}\\
\cline{2-12} \\[-2.35ex]
& \exampleObject{0} & 0.08 & 0.92 & 0.00 & 0.00 & 0.00 & 0.00 & 0.00 & 0.00 & 0.00 & 0.00 \\
& \exampleObject{1} & 0.00 & 1.00 & 0.00 & 0.00 & 0.00 & 0.00 & 0.00 & 0.00 & 0.00 & 0.00 \\
& \exampleObject{2} & 0.00 & 0.00 & 0.29 & 0.66 & 0.06 & 0.00 & 0.00 & 0.00 & 0.00 & 0.00 \\
& \exampleObject{3} & 0.00 & 0.00 & 0.00 & 0.34 & 0.03 & 0.15 & 0.29 & 0.14 & 0.06 & 0.00 \\
& \exampleObject{4} & 0.00 & 0.00 & 0.00 & 0.00 & 0.04 & 0.23 & 0.43 & 0.21 & 0.09 & 0.00 \\
& \exampleObject{5} & 0.00 & 0.00 & 0.00 & 0.00 & 0.00 & 0.24 & 0.45 & 0.22 & 0.10 & 0.00 \\
& \exampleObject{6} & 0.00 & 0.00 & 0.00 & 0.00 & 0.00 & 0.00 & 0.59 & 0.28 & 0.13 & 0.00 \\
& \exampleObject{7} & 0.00 & 0.00 & 0.00 & 0.00 & 0.00 & 0.00 & 0.00 & 0.44 & 0.20 & 0.37 \\
& \exampleObject{8} & 0.00 & 0.00 & 0.00 & 0.00 & 0.00 & 0.00 & 0.00 & 0.00 & 0.35 & 0.65 \\
& \exampleObject{9} & 0.00 & 0.00 & 0.00 & 0.00 & 0.00 & 0.00 & 0.00 & 0.00 & 0.00 & 1.00 \\
\end{tabular}
}

%% file: figures/algorithm-demos/p-alpaca/c_3.tex
{\scriptsize
\setlength{\tabcolsep}{.125em}
\begin{tabular}{@{}cc|@{\hspace{.16667em}}|@{\hspace{1.5pt}}*{10}{Y}@{\hspace{1.5pt}}}
& & \multicolumn{10}{c}{\distortPubDataVal} \\
& \privDataVal ~~
& \multicolumn{1}{c}{\tiny 2493855}
& \multicolumn{1}{c}{\tiny 3833489}
& \multicolumn{1}{c}{\tiny 7929946}
& \multicolumn{1}{c}{\tiny 13322074}
& \multicolumn{1}{c}{\tiny 13589747}
& \multicolumn{1}{c}{\tiny 16235142}
& \multicolumn{1}{c}{\tiny 16719886}
& \multicolumn{1}{c}{\tiny 19437984}
& \multicolumn{1}{c}{\tiny 25905442}
& \multicolumn{1}{c}{\tiny 34389677}\\
\cline{2-12} \\[-2.35ex]
& \exampleObject{0} & 0.08 & 0.92 & 0.00 & 0.00 & 0.00 & 0.00 & 0.00 & 0.00 & 0.00 & 0.00 \\
& \exampleObject{1} & 0.00 & 0.84 & 0.16 & 0.00 & 0.00 & 0.00 & 0.00 & 0.00 & 0.00 & 0.00 \\
& \exampleObject{2} & 0.00 & 0.00 & 0.13 & 0.31 & 0.03 & 0.14 & 0.27 & 0.13 & 0.00 & 0.00 \\
& \exampleObject{3} & 0.00 & 0.00 & 0.00 & 0.30 & 0.03 & 0.13 & 0.26 & 0.12 & 0.06 & 0.10 \\
& \exampleObject{4} & 0.00 & 0.00 & 0.00 & 0.00 & 0.04 & 0.19 & 0.37 & 0.18 & 0.08 & 0.15 \\
& \exampleObject{5} & 0.00 & 0.00 & 0.00 & 0.00 & 0.00 & 0.20 & 0.38 & 0.18 & 0.08 & 0.15 \\
& \exampleObject{6} & 0.00 & 0.00 & 0.00 & 0.00 & 0.00 & 0.00 & 0.48 & 0.23 & 0.10 & 0.19 \\
& \exampleObject{7} & 0.00 & 0.00 & 0.00 & 0.00 & 0.00 & 0.00 & 0.00 & 0.44 & 0.20 & 0.37 \\
& \exampleObject{8} & 0.00 & 0.00 & 0.00 & 0.00 & 0.00 & 0.00 & 0.00 & 0.00 & 0.35 & 0.65 \\
& \exampleObject{9} & 0.00 & 0.00 & 0.00 & 0.00 & 0.00 & 0.00 & 0.00 & 0.00 & 0.00 & 1.00 \\
\end{tabular}
}

%% file: figures/mutual-information-nodejs/mutual-information-nodejs.tex
\begin{tikzpicture}

%\pgfplotsset{compat=1.12}

\begin{axis}[
xlabel={$\padFactor$},
ylabel={$\mutInfo{\privDataRV}{\distortPubDataRV}$ (bits)},
%width  = 0.50*\textwidth,
%height = 5cm,
x=2.1cm,
y=1cm,
ymax=12,
major x tick style = transparent,
ybar=0pt,
bar width=8pt,
%ybar=2*\pgflinewidth,
%bar width=25pt,
ymajorgrids = true,
symbolic x coords={1.01,1.02,1.03,1.04,1.05,1.06,1.07,1.08,1.09,1.1},
xtick = data,
scaled y ticks = false,
enlarge x limits=0.06,
ymin=5,
legend cell align=left,
legend pos=north east,
legend style={legend columns=-1},
legend image code/.code={\draw [#1] (0cm,-0.1cm) rectangle (0.2cm,0.25cm); },
]

\addplot[style={fill=white,postaction={pattern=north east lines}},error bars/.cd, y dir=both, y explicit]
coordinates {%
(1.01, -0.01)
(1.02, -0.01)
(1.03, -0.01)
(1.04, -0.01)
(1.05, -0.01)
(1.06, -0.01)
(1.07, -0.01)
(1.08, -0.01)
(1.09, 6.628412525306687) += (0,2.3574294116966543)
(1.1, -0.01)
};
\addplot[black!100!white,fill=black!100!white,error bars/.cd, y dir=both, y explicit]
coordinates {%
(1.01, 10.953143421467676) += (0,6.006507662013547)
(1.02, 10.953143421467676) += (0,6.006507662013547)
(1.03, 10.953143421467676) += (0,6.006507662013547)
(1.04, 10.953143421467676) += (0,6.006507662013547)
(1.05, 10.953143421467676) += (0,6.006507662013547)
(1.06, 10.953143421467676) += (0,6.006507662013547)
(1.07, 10.953143421467676) += (0,6.006507662013547)
(1.08, 10.953143421467676) += (0,6.006507662013547)
(1.09, 10.953143421467676) += (0,6.006507662013547)
(1.1, 10.781219563659695) += (0,5.875456676713112)
};
\addplot[black!100!white,fill=black!75!white,error bars/.cd, y dir=both, y explicit]
coordinates {%
(1.01, 8.466755720176913) += (0,1.90719693519328)
(1.02, 7.621422380874755) += (0,1.8048423738273431)
(1.03, 7.102572023482344) += (0,1.7708720890330332)
(1.04, 6.718974198085197) += (0,1.7567592328812012)
(1.05, 6.422920842825744) += (0,1.7520048396749344)
(1.06, 6.185841241769393) += (0,1.7330219955052018)
(1.07, 5.976910480828468) += (0,1.7304486512524138)
(1.08, 5.802044672368523) += (0,1.7293367881477888)
(1.09, 5.635089418967113) += (0,1.7399500123798113)
(1.1, 5.5026607294148855) += (0,1.735744009910193)
};
\addplot[black!100!white,fill=black!50!white,error bars/.cd, y dir=both, y explicit]
coordinates {%
(1.01, 8.384779127380826) += (0,2.7667990646333624)
(1.02, 7.553898501019978) += (0,2.61352551159065)
(1.03, 7.050270728455452) += (0,2.5474730304371267)
(1.04, 6.678183370495728) += (0,2.520766227816842)
(1.05, 6.39276533414659) += (0,2.4792180762230114)
(1.06, 6.148508543375182) += (0,2.464025642023967)
(1.07, 5.940749661509789) += (0,2.443771935929596)
(1.08, 5.765990159717148) += (0,2.4342033359897)
(1.09, 5.613771380966534) += (0,2.4147565094766987)
(1.1, 5.4732699394599225) += (0,2.3927007393059405)
};
\addplot[black!100!white,fill=black!25!white,error bars/.cd, y dir=both, y explicit]
coordinates {%
(1.01, 8.277937059655287) += (0,2.567552991289089)
(1.02, 7.484942386009778) += (0,2.4428355760725635)
(1.03, 6.9884904978593845) += (0,2.454452997989341)
(1.04, 6.64343207332388) += (0,2.3480897727518153)
(1.05, 6.352539844983313) += (0,2.40568336974341)
(1.06, 6.126396878322534) += (0,2.373449008760671)
(1.07, 5.920926910138978) += (0,2.4584514569322833)
(1.08, 5.746208516870693) += (0,2.413662819907696)
(1.09, 5.594034202809926) += (0,2.405965797190074)
(1.1, 5.4553692155593945) += (0,2.3963798258566626)
};
\addplot[style={fill=white},error bars/.cd, y dir=both, y explicit]
coordinates {%
(1.01, 8.170424334614124) += (0,3.381000260312117)
(1.02, 7.3681231415836015) += (0,3.2396176890616433)
(1.03, 6.888286509224871) += (0,3.179997540674105)
(1.04, 6.534558928010268) += (0,3.125913436643435)
(1.05, 6.265147198470579) += (0,3.0919764712719946)
(1.06, 6.032419213042403) += (0,3.016738609650111)
(1.07, 5.839072111570392) += (0,2.989559268273606)
(1.08, 5.664432996257888) += (0,2.949499986280035)
(1.09, 5.521025345830834) += (0,2.827825563228645)
(1.1, 5.386305614519569) += (0,2.859780295892125)
};

\legend{\padme, \dAlpaca, \noDistAlg, \pAlpaca, \perObjAlg, \perReqAlg}

\end{axis}
\end{tikzpicture}

%% file: figures/mutual-information-unsplash/mutual-information-unsplash.tex
\begin{tikzpicture}

%\pgfplotsset{compat=1.12}

\begin{axis}[
xlabel={$\padFactor$},
ylabel={$\mutInfo{\privDataRV}{\distortPubDataRV}$ (bits)},
%width  = 0.50*\textwidth,
%height = 5cm,
x=2.1cm,
y=1cm,
ymax=10,
major x tick style = transparent,
ybar=0pt,
bar width=8pt,
%ybar=2*\pgflinewidth,
%bar width=25pt,
ymajorgrids = true,
symbolic x coords={1.01,1.02,1.03,1.04,1.05,1.06,1.07,1.08,1.09,1.1},
xtick = data,
scaled y ticks = false,
enlarge x limits=0.06,
ymin=5,
legend cell align=left,
legend pos=north east,
legend style={legend columns=-1},
legend image code/.code={\draw [#1] (0cm,-0.1cm) rectangle (0.2cm,0.25cm); },
]

\addplot[style={fill=white,postaction={pattern=north east lines}},error bars/.cd, y dir=both, y explicit]
coordinates {%
(1.01, -0.01)
(1.02, -0.01)
(1.03, 7.347807198866242) += (0,0.6690010888203126)
(1.04, -0.01)
(1.05, -0.01)
(1.06, -0.01)
(1.07, -0.01)
(1.08, -0.01)
(1.09, -0.01)
(1.1, -0.01)
};
\addplot[black!100!white,fill=black!100!white,error bars/.cd, y dir=both, y explicit]
coordinates {%
(1.01, 12.905728977214935) += (0,0.9761497303855116)
(1.02, 12.485657287345207) += (0,0.886527609141142)
(1.03, 12.173573125579464) += (0,0.8221940252983373)
(1.04, 11.919305927593564) += (0,0.7861104712891347)
(1.05, 11.703131816864765) += (0,0.7723483097307113)
(1.06, 11.513800496420775) += (0,0.7681295305346705)
(1.07, 11.352582652609794) += (0,0.7595310054880713)
(1.08, 11.201057188580416) += (0,0.7679695135703533)
(1.09, 11.074178253528908) += (0,0.7646317321796534)
(1.1, 10.94829504092756) += (0,0.7706658710566501)
};
\addplot[black!100!white,fill=black!75!white,error bars/.cd, y dir=both, y explicit]
coordinates {%
(1.01, 8.431252356986443) += (0,0.5545895800168985)
(1.02, 7.491257190823013) += (0,0.5855584062278192)
(1.03, 6.930375973428737) += (0,0.6087828376792936)
(1.04, 6.533106313673134) += (0,0.6166408058315493)
(1.05, 6.223260928354093) += (0,0.6222291225902818)
(1.06, 5.970230336370048) += (0,0.6444795077451593)
(1.07, 5.7592664164293215) += (0,0.6501245197083811)
(1.08, 5.5757394758585095) += (0,0.6530792146373701)
(1.09, 5.413441629272037) += (0,0.6740212119783013)
(1.1, 5.270329058627385) += (0,0.6838672517594899)
};
\addplot[black!100!white,fill=black!50!white,error bars/.cd, y dir=both, y explicit]
coordinates {%
(1.01, 8.431937501532838) += (0,1.0150294508937012)
(1.02, 7.488183029894812) += (0,0.9895232156139597)
(1.03, 6.927125138806439) += (0,0.9858195110530179)
(1.04, 6.528758235078951) += (0,0.9724524312810612)
(1.05, 6.218886337267028) += (0,0.9732019035179214)
(1.06, 5.968229986403829) += (0,0.9854985909072207)
(1.07, 5.753558234408785) += (0,0.9879352881099335)
(1.08, 5.57122400599217) += (0,0.9938660159320438)
(1.09, 5.409305987711459) += (0,0.9922995033376685)
(1.1, 5.266636164962187) += (0,1.002196871203692)
};
\addplot[black!100!white,fill=black!25!white,error bars/.cd, y dir=both, y explicit]
coordinates {%
(1.01, 8.402787295271834) += (0,0.8236248975169538)
(1.02, 7.472212462477684) += (0,0.854217024644619)
(1.03, 6.918159818803513) += (0,0.8300330307859465)
(1.04, 6.522677469755046) += (0,0.7992506251323164)
(1.05, 6.21493975042396) += (0,0.8618758466268721)
(1.06, 5.9642480725700375) += (0,0.8301677937800687)
(1.07, 5.750789142306028) += (0,0.9771313122571721)
(1.08, 5.568531222291302) += (0,0.8744122735574251)
(1.09, 5.4082453216475095) += (0,0.9136827732398531)
(1.1, 5.262197002485307) += (0,1.0415837456917965)
};
\addplot[style={fill=white},error bars/.cd, y dir=both, y explicit]
coordinates {%
(1.01, 8.303639767789576) += (0,1.2216979882769206)
(1.02, 7.396466136757662) += (0,1.2041252379728489)
(1.03, 6.853125522089223) += (0,1.1835082092319782)
(1.04, 6.4668132820437085) += (0,1.1613127799726737)
(1.05, 6.167182758323941) += (0,1.1345412741696874)
(1.06, 5.925495356221929) += (0,1.1225803447938523)
(1.07, 5.714914690624754) += (0,1.1062724361274423)
(1.08, 5.5353877727317045) += (0,1.0977719399509551)
(1.09, 5.376783273024744) += (0,1.0914126122704069)
(1.1, 5.2375102124788775) += (0,1.143830704233288)
};

\legend{\padme, \dAlpaca, \noDistAlg, \pAlpaca, \perObjAlg, \perReqAlg}

\end{axis}
\end{tikzpicture}

%% file: figures/recall-precision-nodejs/recall-precision-nodejs_legend.tex
\newenvironment{customlegend}[1][]{%
    \begingroup
    % inits/clears the lists (which might be populated from previous
    % axes):
    \csname pgfplots@init@cleared@structures\endcsname
    \pgfplotsset{#1}%
}{%
    % draws the legend:
    \csname pgfplots@createlegend\endcsname
    \endgroup
}%

\def\addlegendimage{\csname pgfplots@addlegendimage\endcsname}

\begin{tikzpicture}

\begin{customlegend}[
    legend style={{font={\fontsize{10pt}{12}\selectfont}},{draw=none}},
    legend columns=6,
    legend cell align={left},
    legend entries={{\perReqAlg},{\perObjAlg},{\pAlpaca},{\noDistAlg},{\padme},{\dAlpaca}}]
%legend cell align={left},
%legend style={at={(0.97,0.03)}, anchor=south east, draw=white!80.0!black, nodes={scale=0.618, transform shape}}
%%]
%\addlegendimage{line width=1pt, dashed, curve_color}
\addlegendimage{line width=1pt, densely dotted, curvecolor}
\addlegendimage{line width=1pt, dash pattern=on 1pt off 3pt on 3pt off 3pt, curvecolor}
\addlegendimage{line width=1pt, densely dotted, black}
\addlegendimage{line width=1pt, solid, curvecolor}
\addlegendimage{line width=1pt, solid, black}
\addlegendimage{line width=1pt, dash pattern=on 1pt off 3pt on 3pt off 3pt, black}

\end{customlegend}

\end{tikzpicture}

%% file: figures/recall-precision-nodejs/recall-precision-nodejs_c_1-01.tex
\begin{tikzpicture}

\pgfplotsset{every axis/.append style={
                  ylabel={Precision},
                    compat=1.3,
                    x label style={yshift=-1.5em, align=center},
                    label style={font=\small},
                    tick label style={font=\small}  
                    }}
%\draw [line width=1.5pt] (0.25,4.14) rectangle (0.5,2.78);
\begin{axis}[
xmin=0, xmax=1.0,
ymin=0, ymax=1.0,
width=1.25\figurewidth,
height=\figureheight,
xtick={0.0,0.2,0.4,0.6,0.8,1.0},
xticklabels={0.0,0.2,0.4,0.6,0.8,1.0},
ytick={0.0,0.2,0.4,0.6,0.8,1.0},
yticklabels={0.0,0.2,0.4,0.6,0.8,1.0},
tick align=outside,
tick pos=left,
minor tick num=1,
xmajorgrids,
x grid style={lightgray!92.02614379084967!black},
ymajorgrids,
y grid style={lightgray!92.02614379084967!black},
grid=both,
]
\addplot [line width=1pt, densely dotted, curvecolor]
table {%
1.0 0.016471778847454096
0.9787673403942141 0.19391219908514867
0.7837192248066039 0.3008242902084101
0.7516289881345307 0.3250831835294759
0.6623602869516084 0.3747901932341018
0.5388069134417781 0.4512166872401678
0.5366036148079325 0.45260908830332003
0.5115725937496538 0.4636760150814195
0.3051417837148541 0.5490290941317921
0.2955549170936443 0.5550706878507726
0.19261071179494052 0.6098254325848391
0.1787125527870154 0.6195572239373315
0.1787116871392433 0.6195573663429609
0.004108822984344563 0.6678621350284989
0.00024968098053456566 0.9477549511098924
0.00023957175938102237 0.9611289736089972
0.00023908113732727735 0.9615401826852206
0.00023908113390015971 0.9615401852838252
0.00023907272873060107 0.9615431913171744
0.00021573922139739242 0.9629049986243752
};
\addplot [line width=1pt, dash pattern=on 1pt off 3pt on 3pt off 3pt, curvecolor]
table {%
1.0 0.01647177884745359
0.9773493678278355 0.21655642096561828
0.7629593090073011 0.3675628556104657
0.7515180821990843 0.379826593052399
0.6818279616133909 0.433238541820604
0.6354253474607221 0.4601058076963003
0.5660372700747385 0.5052966119399384
0.565611119826076 0.5055059886185532
0.5629192381897092 0.506356909814776
0.3470478356917074 0.5775332327107694
0.19340385568035448 0.6993818226354016
0.1812690444073692 0.7163754448181967
0.18021648432632054 0.717266333055533
0.1479428510818882 0.7404313706385544
0.1479428510818882 0.7404313706385544
0.011682003167357987 0.9145353440829256
0.011682003167357987 0.9145353440829256
0.011682003167357987 0.9145353440829256
0.011682003167357987 0.9145353440829256
0.00023907272873060025 0.9731047531916664
};
\addplot [line width=1pt, densely dotted, black]
table {%
1.0 0.01647177884745364
0.9760398528681672 0.2178108452705684
0.8631205501274385 0.29700782946272075
0.74499641062177 0.3935696370019246
0.6961334609097561 0.43303072795924014
0.6598394319079796 0.4557915052244492
0.5948601787926552 0.4950830343636417
0.4980461940021979 0.5465889111054927
0.476782074241909 0.5593690762604563
0.35206496343184673 0.638837078498042
0.3500752269148013 0.6401132821604768
0.19712019128263614 0.7871173562565251
0.19477104436812792 0.7907005217753635
0.18592507698287525 0.8006262174989079
0.1796891716768267 0.8052745951482733
0.17795703670057425 0.8063117371655378
0.07664604018895482 0.8389543316443251
0.010460066513132148 0.8812299004742831
0.004257067124129189 0.9091424742227
0.00023512891920366224 0.9767322350324588
};
\addplot [line width=1pt, solid, curvecolor]
table {%
1.0 0.01647177884745359
0.9886109293131994 0.217406125977401
0.8252284487285775 0.31278675546295864
0.6981809924330638 0.4311263516342601
0.6680889281690494 0.4584927530738567
0.6249694320302843 0.488067229264913
0.5011932372983504 0.6031511304504134
0.500767087049688 0.6035730459329561
0.490829772552084 0.6099315989166287
0.4908133134900661 0.6099416077427677
0.2816803797122976 0.7699935202824181
0.2816803797122976 0.7699935202824181
0.2816803797122976 0.7699935202824181
0.2816803797122976 0.7699935202824181
0.2668536868029455 0.7765001777821505
0.13059283888841533 0.8336835896167086
0.13059283888841533 0.8336835896167086
0.005742866597440356 0.8941817690126506
0.003426532022787076 0.9143647576179741
0.00023907272873060025 0.9731047531916664
};
%\addplot [line width=1pt, solid, black]
%table {%
%};
\addplot [line width=1pt, dash pattern=on 1pt off 3pt on 3pt off 3pt, black]
table {%
1.0 0.01647177884745359
0.9999312596623078 0.9170891522789832
0.9957577329565871 0.9695780719247776
0.9957577329565871 0.9695780719247776
0.9957577329565871 0.9695780719247776
0.9957570663186004 0.9695803116253254
0.9957536743827902 0.9695881494739864
0.995750070450992 0.9695957039082312
0.995750070450992 0.9695957039082312
0.995750070450992 0.9695957039082312
0.995750070450992 0.9695957039082312
0.994694087784111 0.970535670675582
0.994694087784111 0.970535670675582
0.9945990318480372 0.9705852037542186
0.9945990318480372 0.9705852037542186
0.9476290781186136 0.9860870789804644
0.9476290781186136 0.9860870789804644
0.8971396768235588 0.998667190053948
0.8971047511230571 0.99867220967799
0.8971007666201477 0.99867265244238
0.06994341492480823 1.0
};

\end{axis}

\end{tikzpicture}

%% file: figures/recall-precision-nodejs/recall-precision-nodejs_c_1-03.tex
\begin{tikzpicture}

\pgfplotsset{every axis/.append style={
                    compat=1.3,
                    x label style={yshift=-1.5em, align=center},
                    label style={font=\small},
                    tick label style={font=\small}  
                    }}
%\draw [line width=1.5pt] (0.25,4.14) rectangle (0.5,2.78);
\begin{axis}[
xmin=0, xmax=1.0,
ymin=0, ymax=1.0,
width=1.25\figurewidth,
height=\figureheight,
xtick={0.0,0.2,0.4,0.6,0.8,1.0},
xticklabels={0.0,0.2,0.4,0.6,0.8,1.0},
ytick={0.0,0.2,0.4,0.6,0.8,1.0},
yticklabels={},
tick align=outside,
tick pos=left,
minor tick num=1,
xmajorgrids,
x grid style={lightgray!92.02614379084967!black},
ymajorgrids,
y grid style={lightgray!92.02614379084967!black},
grid=both,
]
\addplot [line width=1pt, densely dotted, curvecolor]
table {%
1.0 0.016471778847454238
0.8299136263823129 0.1403929408807079
0.674200706335155 0.20113216858560126
0.5378598232210653 0.2500689347184802
0.423585851074464 0.2867873533895882
0.1923425862923869 0.4264082680312671
0.1818725849234981 0.43824132153413853
0.1687874043874461 0.453830228417352
0.1593799042552257 0.4576129057197574
0.13664991758386302 0.4684227000886369
0.001256473448788544 0.5224909136409339
5.548168403660205e-05 0.5570623790261451
4.521182920846754e-12 0.8181796936686528
4.521182920846754e-12 0.8181796936686528
4.521182920846754e-12 0.8181796936686528
4.421382467939041e-12 0.8202975183452553
4.421382467939041e-12 0.8202975183452553
8.386941082194781e-14 0.9149585639647962
4.2494011049333267e-14 0.9348285801704582
6.853865786124552e-15 0.9529299543463976
};
\addplot [line width=1pt, dash pattern=on 1pt off 3pt on 3pt off 3pt, curvecolor]
table {%
1.0 0.01647177884745359
0.8352802921431404 0.14638941122365948
0.629862036356888 0.24459768907364546
0.5839701843938295 0.2669593512123392
0.5194048395000979 0.29243641343913285
0.3916500087839133 0.3414378232013449
0.1836002467970452 0.5170554837089152
0.16877355388769316 0.5475031615519096
0.16877355388769316 0.5475031615519096
0.16877355388769316 0.5475031615519096
0.16877355388769316 0.5475031615519096
};
\addplot [line width=1pt, densely dotted, black]
table {%
1.0 0.016471778847454606
0.8614991938644677 0.14998947522268286
0.6661640556609931 0.2406028174706309
0.5661865308628751 0.2960147139944625
0.5200496778787235 0.3185761017503884
0.4227480752888496 0.3548594378615855
0.2847865584168369 0.4229805650406505
0.20088679470368576 0.4934850096474819
0.17857400166792672 0.5166415944866357
0.12489439254513002 0.5646105681332977
0.09419768510188377 0.6053970565577991
0.08946352952654314 0.6104678458519771
0.08572107530061931 0.612221865654176
0.00042557477841648396 0.7250645855393131
0.00010758112663341626 0.8449437672320765
0.00010757127920810136 0.8449595930896179
0.00010288650786879459 0.8488095283490301
1.7729568906857884e-05 0.9368159718822668
1.7674560697838996e-05 0.937086691449528
3.5459137813715782e-06 0.953759072867
};
\addplot [line width=1pt, solid, curvecolor]
table {%
1.0 0.01647177884745359
0.920926742863014 0.13576160031840234
0.6324454577807269 0.2621346237720049
0.5298377316553914 0.3239985480382026
0.5298377316553914 0.3239985480382026
0.39214558644160297 0.3740745810485766
0.18409582445473488 0.5726930729925949
0.17196101318174964 0.6069628354808466
0.17196101318174964 0.6069628354808466
0.17196101318174964 0.6069628354808466
0.17196101318174964 0.6069628354808466
0.17172194045301906 0.6070540619975877
0.17172194045301906 0.6070540619975877
};
%\addplot [line width=1pt, solid, black]
%table {%
%};
\addplot [line width=1pt, dash pattern=on 1pt off 3pt on 3pt off 3pt, black]
table {%
1.0 0.01647177884745359
0.9999312596623078 0.9170891522789832
0.9957577329565871 0.9695780719247776
0.9957577329565871 0.9695780719247776
0.9957577329565871 0.9695780719247776
0.9957570663186004 0.9695803116253254
0.9957536743827902 0.9695881494739864
0.995750070450992 0.9695957039082312
0.995750070450992 0.9695957039082312
0.995750070450992 0.9695957039082312
0.995750070450992 0.9695957039082312
0.994694087784111 0.970535670675582
0.994694087784111 0.970535670675582
0.9945990318480372 0.9705852037542186
0.9945990318480372 0.9705852037542186
0.9476290781186136 0.9860870789804644
0.9476290781186136 0.9860870789804644
0.8971396768235588 0.998667190053948
0.8971047511230571 0.99867220967799
0.8971007666201477 0.99867265244238
0.06994341492480823 1.0
};

\end{axis}

\end{tikzpicture}

%% file: figures/recall-precision-nodejs/recall-precision-nodejs_c_1-05.tex
\begin{tikzpicture}

\pgfplotsset{every axis/.append style={
                    compat=1.3,
                    x label style={yshift=-1.5em, align=center},
                    label style={font=\small},
                    tick label style={font=\small}  
                    }}
%\draw [line width=1.5pt] (0.25,4.14) rectangle (0.5,2.78);
\begin{axis}[
xmin=0, xmax=1.0,
ymin=0, ymax=1.0,
width=1.25\figurewidth,
height=\figureheight,
xtick={0.0,0.2,0.4,0.6,0.8,1.0},
xticklabels={0.0,0.2,0.4,0.6,0.8,1.0},
ytick={0.0,0.2,0.4,0.6,0.8,1.0},
yticklabels={},
tick align=outside,
tick pos=left,
minor tick num=1,
xmajorgrids,
x grid style={lightgray!92.02614379084967!black},
ymajorgrids,
y grid style={lightgray!92.02614379084967!black},
grid=both,
]
\addplot [line width=1pt, densely dotted, curvecolor]
table {%
1.0 0.01647177884745387
0.7134238775700285 0.12770865913620222
0.5620016924706243 0.16692075991044045
0.298387005905829 0.2224746741963047
0.18464420401964884 0.2712239246038322
0.16314949400480375 0.277437224870771
0.024313565371999416 0.36029012527772986
0.023805186100778572 0.3614759020899471
6.80875458183023e-05 0.4023867870616633
4.385989267649086e-14 0.4636752399482574
};
\addplot [line width=1pt, dash pattern=on 1pt off 3pt on 3pt off 3pt, curvecolor]
table {%
1.0 0.01647177884745359
0.7234018435319269 0.13680724182702628
0.6057486139874703 0.16332955473261526
0.19552961607835606 0.3186074179437653
0.17849450732314984 0.3494949660743568
0.17849450732314984 0.3494949660743568
0.17849450732314984 0.3494949660743568
0.042210947742036216 0.4149649292387783
0.042210947742036216 0.4149649292387783
};
\addplot [line width=1pt, densely dotted, black]
table {%
1.0 0.01647177884745338
0.7135602893055492 0.13942771502725032
0.5693159476994839 0.18690976574829632
0.3649255371302873 0.23772827572563576
0.24685475326679995 0.29082236156678765
0.17666467919729004 0.3338600069046126
0.10961755306049664 0.3915116958309759
0.09809486838567172 0.397928676340165
0.05284091972623815 0.4181562665122489
0.0011054690381025812 0.4611107484863422
3.114248189328118e-07 0.5066174920658569
};
\addplot [line width=1pt, solid, curvecolor]
table {%
1.0 0.01647177884745359
0.7881576655439061 0.13436305208147006
0.5918782528597372 0.1841964869931356
0.30965315839553803 0.2772525160501745
0.18189832767935352 0.4529425151219482
0.18189832767935352 0.4529425151219482
0.17871086838529704 0.4576705418674333
0.17871086838529704 0.4576705418674333
0.17871086838529704 0.4576705418674333
0.1364999206432608 0.4721269270260594
0.00023907272873060025 0.5116234579035688
};
%\addplot [line width=1pt, solid, black]
%table {%
%};
\addplot [line width=1pt, dash pattern=on 1pt off 3pt on 3pt off 3pt, black]
table {%
1.0 0.01647177884745359
0.9999312596623078 0.9170891522789832
0.9957577329565871 0.9695780719247776
0.9957577329565871 0.9695780719247776
0.9957577329565871 0.9695780719247776
0.9957570663186004 0.9695803116253254
0.9957536743827902 0.9695881494739864
0.995750070450992 0.9695957039082312
0.995750070450992 0.9695957039082312
0.995750070450992 0.9695957039082312
0.995750070450992 0.9695957039082312
0.994694087784111 0.970535670675582
0.994694087784111 0.970535670675582
0.9945990318480372 0.9705852037542186
0.9945990318480372 0.9705852037542186
0.9476290781186136 0.9860870789804644
0.9476290781186136 0.9860870789804644
0.8971396768235588 0.998667190053948
0.8971047511230571 0.99867220967799
0.8971007666201477 0.99867265244238
0.06994341492480823 1.0
};

\end{axis}

\end{tikzpicture}

%% file: figures/recall-precision-nodejs/recall-precision-nodejs_c_1-07.tex
\begin{tikzpicture}

\pgfplotsset{every axis/.append style={
                    compat=1.3,
                    x label style={yshift=-1.5em, align=center},
                    label style={font=\small},
                    tick label style={font=\small}  
                    }}
%\draw [line width=1.5pt] (0.25,4.14) rectangle (0.5,2.78);
\begin{axis}[
xmin=0, xmax=1.0,
ymin=0, ymax=1.0,
width=1.25\figurewidth,
height=\figureheight,
xtick={0.0,0.2,0.4,0.6,0.8,1.0},
xticklabels={0.0,0.2,0.4,0.6,0.8,1.0},
ytick={0.0,0.2,0.4,0.6,0.8,1.0},
yticklabels={},
tick align=outside,
tick pos=left,
minor tick num=1,
xmajorgrids,
x grid style={lightgray!92.02614379084967!black},
ymajorgrids,
y grid style={lightgray!92.02614379084967!black},
grid=both,
]
\addplot [line width=1pt, densely dotted, curvecolor]
table {%
1.0 0.01647177884745364
0.6736683005290948 0.11214485534957344
0.5005960718852713 0.13903547400108174
0.18335239499188166 0.30489936752100383
0.1799261091778133 0.30965886201782183
0.15627081395190842 0.33652573745055137
0.15626516713190788 0.33652898975204
0.023047570256647897 0.4188308435179883
0.018550215438757383 0.4374011613238138
0.0005677997804397608 0.4727583009667102
1.5289420741964001e-25 0.5058105350351483
};
\addplot [line width=1pt, dash pattern=on 1pt off 3pt on 3pt off 3pt, curvecolor]
table {%
1.0 0.01647177884745359
0.7252446888983523 0.11095298570395244
0.5298541907174094 0.1529004150655365
0.1938105252856261 0.2579069831704326
0.14621462147415198 0.297431746905106
0.14621462147415198 0.297431746905106
0.0099537735596218 0.4087212435893953
0.0099537735596218 0.4087212435893953
0.0099537735596218 0.4087212435893953
};
\addplot [line width=1pt, densely dotted, black]
table {%
1.0 0.016471778847453652
0.6656279087056844 0.1191131518654653
0.4521457755872416 0.1658522427884648
0.22852869081825994 0.2621256867291777
0.1762288556145194 0.30935422658899625
0.1602631227508484 0.3242436220769048
0.1134551829930973 0.3375859283433596
0.04986379904165551 0.3838244525675898
0.017997412319724106 0.4382382562388043
0.010741657327917388 0.4622902759826733
3.4061202080154074e-06 0.505071595214437
};
\addplot [line width=1pt, solid, curvecolor]
table {%
1.0 0.01647177884745359
0.6022833020238673 0.13623176525431402
0.5315505826143986 0.1622184514337932
0.1816592549506229 0.38586134474673456
0.17847179565656646 0.39262383584870497
0.17847179565656646 0.39262383584870497
0.17847179565656646 0.39262383584870497
0.17847179565656646 0.39262383584870497
};
%\addplot [line width=1pt, solid, black]
%table {%
%};
\addplot [line width=1pt, dash pattern=on 1pt off 3pt on 3pt off 3pt, black]
table {%
1.0 0.01647177884745359
0.9999312596623078 0.9170891522789832
0.9957577329565871 0.9695780719247776
0.9957577329565871 0.9695780719247776
0.9957577329565871 0.9695780719247776
0.9957570663186004 0.9695803116253254
0.9957536743827902 0.9695881494739864
0.995750070450992 0.9695957039082312
0.995750070450992 0.9695957039082312
0.995750070450992 0.9695957039082312
0.995750070450992 0.9695957039082312
0.994694087784111 0.970535670675582
0.994694087784111 0.970535670675582
0.9945990318480372 0.9705852037542186
0.9945990318480372 0.9705852037542186
0.9476290781186136 0.9860870789804644
0.9476290781186136 0.9860870789804644
0.8971396768235588 0.998667190053948
0.8971047511230571 0.99867220967799
0.8971007666201477 0.99867265244238
0.06994341492480823 1.0
};

\end{axis}

\end{tikzpicture}

%% file: figures/recall-precision-nodejs/recall-precision-nodejs_c_1-09.tex
\begin{tikzpicture}

\pgfplotsset{every axis/.append style={
                    compat=1.3,
                    x label style={yshift=-1.5em, align=center},
                    label style={font=\small},
                    tick label style={font=\small}  
                    }}
%\draw [line width=1.5pt] (0.25,4.14) rectangle (0.5,2.78);
\begin{axis}[
xmin=0, xmax=1.0,
ymin=0, ymax=1.0,
width=1.25\figurewidth,
height=\figureheight,
xtick={0.0,0.2,0.4,0.6,0.8,1.0},
xticklabels={0.0,0.2,0.4,0.6,0.8,1.0},
ytick={0.0,0.2,0.4,0.6,0.8,1.0},
yticklabels={},
tick align=outside,
tick pos=left,
minor tick num=1,
xmajorgrids,
x grid style={lightgray!92.02614379084967!black},
ymajorgrids,
y grid style={lightgray!92.02614379084967!black},
grid=both,
]
\addplot [line width=1pt, densely dotted, curvecolor]
table {%
1.0 0.016471778847453336
0.569274020281592 0.09933502843361892
0.1866659116824209 0.20132273925657776
0.18241381174862464 0.2050844514872749
0.14242991952760586 0.2128125137839492
1.2741689850007602e-05 0.25229722474799315
6.325724809329133e-11 0.3227211626785837
3.886998778835514e-27 0.4483658443865262
3.886998778835514e-27 0.4483658443865262
2.411769820279982e-31 0.5952193254414182
2.411769820279982e-31 0.5952193254414182
2.411769820279982e-31 0.5952193254414182
6.692961708484821e-32 0.6231596309683504
};
\addplot [line width=1pt, dash pattern=on 1pt off 3pt on 3pt off 3pt, curvecolor]
table {%
1.0 0.01647177884745359
0.5427331104628875 0.12262520908499415
0.20574634384247933 0.263457846718164
0.18357257238100486 0.3186965892772252
0.18357257238100486 0.3186965892772252
0.1413616246389686 0.3548740316273283
0.1413616246389686 0.3548740316273283
0.1413616246389686 0.3548740316273283
};
\addplot [line width=1pt, densely dotted, black]
table {%
1.0 0.016471778847453468
0.589410802428772 0.10969371638800987
0.2897817027069779 0.17707925172156647
0.1928171682693015 0.23456749428335344
0.1596877830100577 0.249868350689324
0.08282009473263739 0.2910408054987743
0.033508937017812435 0.3524543248407277
0.013598308038817244 0.39562739889298615
0.008512822290954012 0.4188410534043602
};
\addplot [line width=1pt, solid, curvecolor]
table {%
1.0 0.01647177884745359
0.5974608734776705 0.12183503418098947
0.31523577901347144 0.1878453664375629
0.18357252385180275 0.3634667414896838
0.18357252385180275 0.3634667414896838
0.18357252385180275 0.3634667414896838
0.18357252385180275 0.3634667414896838
0.18357252385180275 0.3634667414896838
};
\addplot [line width=1pt, solid, black]
table {%
1.0 0.01647177884745359
0.7483207229477992 0.13718472340072632
0.5797134071189882 0.202285609509966
0.5338215551559298 0.2163803715594049
0.22261721445276367 0.4108928859431677
0.21992533281639687 0.4145543283747796
0.19084567965828228 0.4453444754895136
0.17871086838529704 0.4576034634660433
0.17871086838529704 0.4576034634660433
0.04245002047076682 0.773806102276098
0.04245002047076682 0.773806102276098
0.04245002047076682 0.773806102276098
0.042210947742036216 0.7755567544260519
0.042210947742036216 0.7755567544260519
0.032273633244432345 0.8046167072604542
0.032273633244432345 0.8046167072604542
0.032273633244432345 0.8046167072604542
};
\addplot [line width=1pt, dash pattern=on 1pt off 3pt on 3pt off 3pt, black]
table {%
1.0 0.01647177884745359
0.9999312596623078 0.9170891522789832
0.9957577329565871 0.9695780719247776
0.9957577329565871 0.9695780719247776
0.9957577329565871 0.9695780719247776
0.9957570663186004 0.9695803116253254
0.9957536743827902 0.9695881494739864
0.995750070450992 0.9695957039082312
0.995750070450992 0.9695957039082312
0.995750070450992 0.9695957039082312
0.995750070450992 0.9695957039082312
0.994694087784111 0.970535670675582
0.994694087784111 0.970535670675582
0.9945990318480372 0.9705852037542186
0.9945990318480372 0.9705852037542186
0.9476290781186136 0.9860870789804644
0.9476290781186136 0.9860870789804644
0.8971396768235588 0.998667190053948
0.8971047511230571 0.99867220967799
0.8971007666201477 0.99867265244238
0.06994341492480823 1.0
};

\end{axis}

\end{tikzpicture}

%% file: figures/recall-precision-nodejs/recall-precision_xlabel.tex
\begin{tikzpicture}
\node at (0,0)[
  scale=1,
  anchor=south,
  text=black,
  rotate=0
]{Recall};
\end{tikzpicture}

%% file: figures/recall-precision-unsplash/recall-precision-unsplash_legend.tex
\newenvironment{customlegend}[1][]{%
    \begingroup
    % inits/clears the lists (which might be populated from previous
    % axes):
    \csname pgfplots@init@cleared@structures\endcsname
    \pgfplotsset{#1}%
}{%
    % draws the legend:
    \csname pgfplots@createlegend\endcsname
    \endgroup
}%

\def\addlegendimage{\csname pgfplots@addlegendimage\endcsname}

\begin{tikzpicture}

\begin{customlegend}[
    legend style={{font={\fontsize{10pt}{12}\selectfont}},{draw=none}},
    legend columns=6,
    legend cell align={left},
    legend entries={{\perReqAlg},{\perObjAlg},{\pAlpaca},{\noDistAlg},{\padme},{\dAlpaca}}]
%legend cell align={left},
%legend style={at={(0.97,0.03)}, anchor=south east, draw=white!80.0!black, nodes={scale=0.618, transform shape}}
%%]
%\addlegendimage{line width=1pt, dashed, curve_color}
\addlegendimage{line width=1pt, densely dotted, curvecolor}
\addlegendimage{line width=1pt, dash pattern=on 1pt off 3pt on 3pt off 3pt, curvecolor}
\addlegendimage{line width=1pt, densely dotted, black}
\addlegendimage{line width=1pt, solid, curvecolor}
\addlegendimage{line width=1pt, solid, black}
\addlegendimage{line width=1pt, dash pattern=on 1pt off 3pt on 3pt off 3pt, black}

\end{customlegend}

\end{tikzpicture}

%% file: figures/recall-precision-unsplash/recall-precision-unsplash_c_1-01.tex
\begin{tikzpicture}

\pgfplotsset{every axis/.append style={
                  ylabel={Precision},
                    compat=1.3,
                    x label style={yshift=-1.5em, align=center},
                    label style={font=\small},
                    tick label style={font=\small}  
                    }}
%\draw [line width=1.5pt] (0.25,4.14) rectangle (0.5,2.78);
\begin{axis}[
xmin=0, xmax=1.0,
ymin=0, ymax=1.0,
width=1.25\figurewidth,
height=\figureheight,
xtick={0.0,0.2,0.4,0.6,0.8,1.0},
xticklabels={0.0,0.2,0.4,0.6,0.8,1.0},
ytick={0.0,0.2,0.4,0.6,0.8,1.0},
yticklabels={0.0,0.2,0.4,0.6,0.8,1.0},
tick align=outside,
tick pos=left,
minor tick num=1,
xmajorgrids,
x grid style={lightgray!92.02614379084967!black},
ymajorgrids,
y grid style={lightgray!92.02614379084967!black},
grid=both,
]
\addplot [line width=1pt, densely dotted, curvecolor]
table {%
1.0 0.022950723488767158
0.698846197034616 0.10883481442046956
0.4076288654551399 0.17907301007722318
0.2518997681464976 0.26023359113957284
0.1876654359196966 0.311895579736168
0.13890420167405956 0.3591272161236139
0.08782626257068486 0.4624480206257793
0.08571045186273486 0.4686286724079335
0.08571045017720266 0.468628675332105
0.06315489323929094 0.4774733172863467
0.02203710577866832 0.5173015780227662
6.93421838949847e-05 0.5665172682446775
1.4271668148360017e-07 0.6266770732654935
};
\addplot [line width=1pt, dash pattern=on 1pt off 3pt on 3pt off 3pt, curvecolor]
table {%
1.0 0.022950723488767286
0.6908439955176687 0.11475084303567501
0.3562797898673091 0.2260067896486075
0.2805187143575065 0.2864375663348727
0.22622351101386215 0.3433686848706949
0.21216784395493946 0.3550898015902038
0.13724954661655786 0.4313836101600134
0.08571054003352385 0.5434896632260234
0.08571054003352385 0.5434896632260234
0.08571054003352385 0.5434896632260234
0.0750899743733008 0.5530293451648869
0.02202616358012168 0.6889094991701651
0.02202616358012168 0.6889094991701651
0.02202616358012168 0.6889094991701651
};
\addplot [line width=1pt, densely dotted, black]
table {%
1.0 0.022950723488767283
0.6943468270258167 0.11970601660756627
0.4432110351380246 0.19166776108725586
0.2948682759489062 0.28028163166853576
0.2221869403389826 0.34862682726950306
0.18491053349079392 0.3894845365471564
0.1385801167761226 0.453526894100745
0.11010205036847412 0.5091795228755169
0.09228213240437656 0.5462434424650471
0.08028467562444068 0.5752573244735687
0.059533727221275236 0.6192973939507751
0.05016889625024492 0.6415437549296416
0.028016351072365952 0.6957572759444717
0.02643868620110711 0.7008082716675038
0.025841802306080627 0.7011675727669543
};
\addplot [line width=1pt, solid, curvecolor]
table {%
1.0 0.02295072348876729
0.7079367853943422 0.11445893034362
0.4057711576338534 0.2111583974270468
0.2772958682019797 0.31802349100203703
0.2327550418423756 0.377722394871678
0.2186993747834529 0.39414114139453776
0.17919024624729618 0.4369389083187135
0.1284986999366356 0.5234235436970356
0.08444012997512246 0.6850255124650232
0.08444012997512246 0.6850255124650232
0.08444012997512246 0.6850255124650232
0.08444012997512246 0.6850255124650232
0.0738195643148994 0.7086948052920993
0.0738195643148994 0.7086948052920993
0.05179340073477773 0.7174575721543616
};
%\addplot [line width=1pt, solid, black]
%table {%
%};
\addplot [line width=1pt, dash pattern=on 1pt off 3pt on 3pt off 3pt, black]
table {%
1.0 0.02295072348876726
0.9987061794368368 0.7366469508839867
0.996109174108606 0.7549746506329652
0.989743077513742 0.7791024239222576
0.9804912460319194 0.8047101258996722
0.9718558911754884 0.8247508019633865
0.9643067837399998 0.8379302769015452
0.9478226669871304 0.8603700985739502
0.9401143681192624 0.8697485907683777
0.9359925585808864 0.8737852478458636
0.9205915678483249 0.8860505080244429
0.9026903352072142 0.8981835630861258
0.8650790182570723 0.9199832696817754
0.8404339995442756 0.9327434885255671
0.8248550889755252 0.9397360284063732
0.7738045815918521 0.9590855719141049
0.7385233902156577 0.9706308279899408
0.7164051453168067 0.9758248299319728
0.6530531980310208 0.985260044832068
0.5855560307020969 0.9936649186927272
0.3879729312136944 1.0
};

\end{axis}

\end{tikzpicture}

%% file: figures/recall-precision-unsplash/recall-precision-unsplash_c_1-03.tex
\begin{tikzpicture}

\pgfplotsset{every axis/.append style={
                    compat=1.3,
                    x label style={yshift=-1.5em, align=center},
                    label style={font=\small},
                    tick label style={font=\small}  
                    }}
%\draw [line width=1.5pt] (0.25,4.14) rectangle (0.5,2.78);
\begin{axis}[
xmin=0, xmax=1.0,
ymin=0, ymax=1.0,
width=1.25\figurewidth,
height=\figureheight,
xtick={0.0,0.2,0.4,0.6,0.8,1.0},
xticklabels={0.0,0.2,0.4,0.6,0.8,1.0},
ytick={0.0,0.2,0.4,0.6,0.8,1.0},
yticklabels={},
tick align=outside,
tick pos=left,
minor tick num=1,
xmajorgrids,
x grid style={lightgray!92.02614379084967!black},
ymajorgrids,
y grid style={lightgray!92.02614379084967!black},
grid=both,
]
\addplot [line width=1pt, densely dotted, curvecolor]
table {%
1.0 0.02295072348876724
0.4206758886227688 0.0944167506037388
0.19928058603674784 0.17837014290079242
0.13797448118843014 0.226673349731445
0.08386370752504901 0.2841083671035898
0.07889607405548348 0.28961595630820064
0.023354604889812967 0.31101969923498496
3.004577290900808e-06 0.3584740147689435
};
\addplot [line width=1pt, dash pattern=on 1pt off 3pt on 3pt off 3pt, curvecolor]
table {%
1.0 0.022950723488767283
0.4337186182183667 0.09838564695377996
0.20752632121085374 0.1866077994397696
0.1560294535363909 0.22484358062063967
0.08658297151097641 0.29570070144766863
0.07596240585075335 0.3079005800970413
0.02202616358012168 0.3372846115240303
};
\addplot [line width=1pt, densely dotted, black]
table {%
1.0 0.0229507234887674
0.4657566763478558 0.09575384551521003
0.20753911495856306 0.19423418932319053
0.14646911670656712 0.2556086903356776
0.1176602810619174 0.2891992413522592
0.08730094740762398 0.3226159793015629
0.06141376808582207 0.3543635715821947
0.03567190888250836 0.3776557840988853
};
\addplot [line width=1pt, solid, curvecolor]
table {%
1.0 0.0229507234887673
0.4616676395031978 0.10717852023727907
0.20121484912710036 0.20724254221976465
0.1505233028164398 0.2646568922037545
0.10710930208603203 0.32758003465344177
0.08960448732555272 0.3560761110904378
0.07898392166532966 0.3771846794809687
0.05695775808520799 0.3935746869843737
};
\addplot [line width=1pt, solid, black]
table {%
1.0 0.02295072348876731
0.47423595915959404 0.11350824338484013
0.243929656115292 0.22215604865131472
0.19280423511638145 0.2832462449586952
0.13940487375495134 0.376447681181409
0.13174807799755908 0.3938544227907042
0.08833407726715134 0.4968790876928074
0.07771351160692827 0.5450783781416937
0.05568734802680659 0.6514697827277707
0.05568734802680659 0.6514697827277707
0.05568734802680659 0.6514697827277707
0.05568734802680659 0.6514697827277707
0.05568734802680659 0.6514697827277707
0.05568734802680659 0.6514697827277707
};
\addplot [line width=1pt, dash pattern=on 1pt off 3pt on 3pt off 3pt, black]
table {%
1.0 0.02295072348876721
0.9880668853728044 0.466814531232949
0.9619844616676396 0.5501094185643859
0.9405029824983068 0.5959526533825299
0.9247617590965418 0.6226448403482852
0.8946964281812535 0.6612819687599852
0.8708301989268625 0.6884165434944072
0.8515852033124304 0.7079358859918624
0.83032222218754 0.7246201328241196
0.7896940715237742 0.7514539373633685
0.7535825875786983 0.7728539393464061
0.6993139161404507 0.8028514245476226
0.6199725941024258 0.8459924822438266
0.5968570617007264 0.8578295686503219
0.5513657688477975 0.877843266855845
0.5116943274786264 0.8931598203121383
0.4315617303796561 0.9165572835872344
0.29319160093517177 0.9688997312895524
0.2679191806947614 0.9771516393442622
0.2024587272801051 0.9943964922501418
0.1561948877700401 1.0
};

\end{axis}

\end{tikzpicture}

%% file: figures/recall-precision-unsplash/recall-precision-unsplash_c_1-05.tex
\begin{tikzpicture}

\pgfplotsset{every axis/.append style={
                    compat=1.3,
                    x label style={yshift=-1.5em, align=center},
                    label style={font=\small},
                    tick label style={font=\small}  
                    }}
%\draw [line width=1.5pt] (0.25,4.14) rectangle (0.5,2.78);
\begin{axis}[
xmin=0, xmax=1.0,
ymin=0, ymax=1.0,
width=1.25\figurewidth,
height=\figureheight,
xtick={0.0,0.2,0.4,0.6,0.8,1.0},
xticklabels={0.0,0.2,0.4,0.6,0.8,1.0},
ytick={0.0,0.2,0.4,0.6,0.8,1.0},
yticklabels={},
tick align=outside,
tick pos=left,
minor tick num=1,
xmajorgrids,
x grid style={lightgray!92.02614379084967!black},
ymajorgrids,
y grid style={lightgray!92.02614379084967!black},
grid=both,
]
\addplot [line width=1pt, densely dotted, curvecolor]
table {%
1.0 0.02295072348876748
0.2707491670482637 0.0975645402895514
0.15070461469894428 0.19125506701912512
0.11904307383572255 0.215733030253376
0.06887279128770663 0.26876984667684845
0.06887086691981176 0.2687723001405075
0.0002886864294023391 0.3113412913046528
};
\addplot [line width=1pt, dash pattern=on 1pt off 3pt on 3pt off 3pt, curvecolor]
table {%
1.0 0.0229507234887673
0.2767324553873814 0.1098767834303334
0.17652144870446268 0.18558920437390886
0.15577505938464706 0.2088909073405132
0.07985635314278221 0.2446893979245373
0.022026163580121667 0.2806100131228377
};
\addplot [line width=1pt, densely dotted, black]
table {%
1.0 0.02295072348876721
0.288430855672733 0.10258768379794224
0.15753068068577403 0.20302633984924345
0.13748605974547734 0.2284356482947129
0.0914432984871157 0.2729190020007849
0.06201314329218709 0.3164791353750952
0.05016918568232745 0.3254417319165738
0.0015343929321231901 0.3554882747044945
};
\addplot [line width=1pt, solid, curvecolor]
table {%
1.0 0.022950723488767286
0.3245569952148928 0.09143969727056313
0.1644853278563157 0.1791077243095162
0.1437389385365001 0.19957959975296066
0.0865829715109764 0.2313450263133753
0.022026163580121667 0.2645112922875082
};
%\addplot [line width=1pt, solid, black]
%table {%
%};
\addplot [line width=1pt, dash pattern=on 1pt off 3pt on 3pt off 3pt, black]
table {%
1.0 0.022950723488767217
0.9796484678604984 0.3728202094499565
0.9381572499211848 0.4497002686472702
0.9089705932846186 0.4904869060211418
0.868147355081172 0.5310746276066579
0.8352805670960675 0.5615177038874272
0.7764671363334156 0.6109794347440579
0.7472258551857389 0.6314023690644807
0.6786705330415865 0.6809437534352544
0.641199679120015 0.7061538223550664
0.6129759745793134 0.721126093374755
0.5185895014811047 0.7743144236684626
0.4832240322877682 0.7925845906524065
0.4445592426233501 0.8129375469118784
0.34447465266614435 0.8649942978519949
0.26443413688590345 0.9158639552860032
0.22320199519928574 0.9497732057352716
0.21472114967428174 0.9564727720191044
0.1483164725675705 0.9945891636751824
0.1483164725675705 0.9945891636751824
0.11284331505233025 1.0
};

\end{axis}

\end{tikzpicture}

%% file: figures/recall-precision-unsplash/recall-precision-unsplash_c_1-07.tex
\begin{tikzpicture}

\pgfplotsset{every axis/.append style={
                    compat=1.3,
                    x label style={yshift=-1.5em, align=center},
                    label style={font=\small},
                    tick label style={font=\small}  
                    }}
%\draw [line width=1.5pt] (0.25,4.14) rectangle (0.5,2.78);
\begin{axis}[
xmin=0, xmax=1.0,
ymin=0, ymax=1.0,
width=1.25\figurewidth,
height=\figureheight,
xtick={0.0,0.2,0.4,0.6,0.8,1.0},
xticklabels={0.0,0.2,0.4,0.6,0.8,1.0},
ytick={0.0,0.2,0.4,0.6,0.8,1.0},
yticklabels={},
tick align=outside,
tick pos=left,
minor tick num=1,
xmajorgrids,
x grid style={lightgray!92.02614379084967!black},
ymajorgrids,
y grid style={lightgray!92.02614379084967!black},
grid=both,
]
\addplot [line width=1pt, densely dotted, curvecolor]
table {%
1.0 0.02295072348876731
0.2027162451287827 0.1094563674614438
0.1425027904815588 0.17926968900348228
0.09597774948352633 0.22690678015455146
0.08371387288527038 0.2427039354342137
0.024353474751676638 0.2810021360356221
0.00028377824209749054 0.3009410042479748
};
\addplot [line width=1pt, dash pattern=on 1pt off 3pt on 3pt off 3pt, curvecolor]
table {%
1.0 0.02295072348876729
0.2220221057592963 0.10279247420758507
0.12126953606622365 0.20674883191247065
0.1106489704060006 0.217032739962959
0.08862280682587892 0.23771030521728576
};
\addplot [line width=1pt, densely dotted, black]
table {%
1.0 0.02295072348876733
0.220513673343636 0.11086526201355833
0.14007633008406484 0.19741540213522524
0.11465165078932488 0.22067546691412784
0.07786095038975749 0.2589514356649305
0.045119251940054325 0.2916304297922774
0.017573107443039568 0.3080392430624781
};
\addplot [line width=1pt, solid, curvecolor]
table {%
1.0 0.0229507234887673
0.22693206895798285 0.10376068620035107
0.15634939710146736 0.1581247711287614
0.09047691880300526 0.18412285020088612
};
%\addplot [line width=1pt, solid, black]
%table {%
%};
\addplot [line width=1pt, dash pattern=on 1pt off 3pt on 3pt off 3pt, black]
table {%
1.0 0.02295072348876729
0.9674656411825114 0.3113520659297386
0.9086428462179549 0.4021581908884059
0.8611398106558376 0.4620609493534673
0.837261095798907 0.4849843917953187
0.7912360434374113 0.5196188863618427
0.7244443126519733 0.564168227456862
0.681948003708224 0.5917628560632838
0.6202582022605184 0.6285832231452017
0.5744235553377511 0.6533837736505558
0.4877297741042361 0.6995167287077475
0.4217215148781562 0.7406707928798154
0.4035299919780004 0.7497665651732638
0.2862730164279316 0.8191257050225741
0.1850007335291492 0.933060980313442
0.1850007335291492 0.933060980313442
0.17267120102132233 0.9486152790877131
0.15297048091419588 0.96492316173937
0.1219468800039954 0.9897147489486752
0.11607708611007933 0.993216083757328
0.08034329164182548 1.0
};

\end{axis}

\end{tikzpicture}

%% file: figures/recall-precision-unsplash/recall-precision-unsplash_c_1-09.tex
\begin{tikzpicture}

\pgfplotsset{every axis/.append style={
                    compat=1.3,
                    x label style={yshift=-1.5em, align=center},
                    label style={font=\small},
                    tick label style={font=\small}  
                    }}
%\draw [line width=1.5pt] (0.25,4.14) rectangle (0.5,2.78);
\begin{axis}[
xmin=0, xmax=1.0,
ymin=0, ymax=1.0,
width=1.25\figurewidth,
height=\figureheight,
xtick={0.0,0.2,0.4,0.6,0.8,1.0},
xticklabels={0.0,0.2,0.4,0.6,0.8,1.0},
ytick={0.0,0.2,0.4,0.6,0.8,1.0},
yticklabels={},
tick align=outside,
tick pos=left,
minor tick num=1,
xmajorgrids,
x grid style={lightgray!92.02614379084967!black},
ymajorgrids,
y grid style={lightgray!92.02614379084967!black},
grid=both,
]
\addplot [line width=1pt, densely dotted, curvecolor]
table {%
1.0 0.02295072348876739
0.16021842865264793 0.1344289547250427
0.13774402952784526 0.1614994366096182
0.09785987447573374 0.19514041241723626
0.014476864603190824 0.20437358528562288
};
\addplot [line width=1pt, dash pattern=on 1pt off 3pt on 3pt off 3pt, curvecolor]
table {%
1.0 0.022950723488767304
0.1472458321498023 0.1982762787628381
0.1472458321498023 0.1982762787628381
0.12477642967952582 0.22549746858649825
0.12477642967952582 0.22549746858649825
};
\addplot [line width=1pt, densely dotted, black]
table {%
1.0 0.022950723488767574
0.18607024747507206 0.1229383975107125
0.1394428507785444 0.1755934404379932
0.1057851767520616 0.20826354920686524
0.07937762967379433 0.2297687977964423
0.023315482118812927 0.2602512965022459
0.0006632270087705165 0.3054546323425139
};
\addplot [line width=1pt, solid, curvecolor]
table {%
1.0 0.022950723488767304
0.17271646133052826 0.13058782575860384
0.12126953606622365 0.22408134779875316
0.12126953606622365 0.22408134779875316
0.12126953606622365 0.22408134779875316
0.032646729240344716 0.2561283212930085
};
%\addplot [line width=1pt, solid, black]
%table {%
%};
\addplot [line width=1pt, dash pattern=on 1pt off 3pt on 3pt off 3pt, black]
table {%
1.0 0.02295072348876725
0.951537133742653 0.2798751940051992
0.8931856702739653 0.3529888508670861
0.8421523305937841 0.4002811468417342
0.7878696128526792 0.4410231231850377
0.7023604031601062 0.4998617138376133
0.6685478307826289 0.5222626455728303
0.5886961597407989 0.5704623891061648
0.5290024939991073 0.6072917581167349
0.42698731774922055 0.6721923701557958
0.4132110160471208 0.6802811979691258
0.34713252530675565 0.7170383695364836
0.3219084867761863 0.7305392827037098
0.21927527320059065 0.8028136017416446
0.14540420577521546 0.9052273610571318
0.14540420577521546 0.9052273610571318
0.12154890142304653 0.931668919646382
0.08604296920113992 0.9741836301950806
0.07132712590793741 1.0
0.07132712590793741 1.0
0.07132712590793741 1.0
};

\end{axis}

\end{tikzpicture}

%% file: figures/recall-precision-unsplash/recall-precision_xlabel.tex
\begin{tikzpicture}
\node at (0,0)[
  scale=1,
  anchor=south,
  text=black,
  rotate=0
]{Recall};
\end{tikzpicture}

%% file: figures/auc-prc-nodejs/auc_ylabel.tex
\begin{tikzpicture}
\node at (0,0)[
  scale=1,
  anchor=south,
  text=black,
  rotate=90
]{AUC};
\end{tikzpicture}

%% file: figures/auc-prc-nodejs/nodejs_c_1-01.pdf_tex
%% Creator: Inkscape inkscape 0.92.5, www.inkscape.org
%% Accompanies image file 'nodejs_c_1-01.pdf' (pdf, eps, ps)
\begingroup%
  \makeatletter%
  \providecommand\color[2][]{%
    \errmessage{(Inkscape) Color is used for the text in Inkscape, but the package 'color.sty' is not loaded}%
    \renewcommand\color[2][]{}%
  }%
  \providecommand\transparent[1]{%
    \errmessage{(Inkscape) Transparency is used (non-zero) for the text in Inkscape, but the package 'transparent.sty' is not loaded}%
    \renewcommand\transparent[1]{}%
  }%
  \providecommand\rotatebox[2]{#2}%
  \newcommand*\fsize{\dimexpr\f@size pt\relax}%
  \newcommand*\lineheight[1]{\fontsize{\fsize}{#1\fsize}\selectfont}%
  \ifx\svgwidth\undefined%
    \setlength{\unitlength}{306bp}%
    \ifx\svgscale\undefined%
      \relax%
    \else%
      \setlength{\unitlength}{\unitlength * \real{\svgscale}}%
    \fi%
  \else%
    \setlength{\unitlength}{\svgwidth}%
  \fi%
  \global\let\svgwidth\undefined%
  \global\let\svgscale\undefined%
  \makeatother%
  \begin{picture}(1,0.44362745)%
    \lineheight{1}%
    \setlength\tabcolsep{0pt}%
    \put(0,0){\includegraphics[width=\unitlength,page=1]{figures/auc-prc-nodejs/nodejs_c_1-01.pdf}}%
    \put(0.05269608,0.04289216){\color[rgb]{0,0,0}\makebox(0,0)[rt]{\lineheight{1.25}\smash{\begin{tabular}[t]{r}0.0\end{tabular}}}}%
    \put(0.05269608,0.11642157){\color[rgb]{0,0,0}\makebox(0,0)[rt]{\lineheight{1.25}\smash{\begin{tabular}[t]{r}0.2\end{tabular}}}}%
    \put(0.05269608,0.18995098){\color[rgb]{0,0,0}\makebox(0,0)[rt]{\lineheight{1.25}\smash{\begin{tabular}[t]{r}0.4\end{tabular}}}}%
    \put(0.05269608,0.26348039){\color[rgb]{0,0,0}\makebox(0,0)[rt]{\lineheight{1.25}\smash{\begin{tabular}[t]{r}0.6\end{tabular}}}}%
    \put(0.05269608,0.3370098){\color[rgb]{0,0,0}\makebox(0,0)[rt]{\lineheight{1.25}\smash{\begin{tabular}[t]{r}0.8\end{tabular}}}}%
    \put(0.05269608,0.41053922){\color[rgb]{0,0,0}\makebox(0,0)[rt]{\lineheight{1.25}\smash{\begin{tabular}[t]{r}1.0\end{tabular}}}}%
    \put(0,0){\includegraphics[width=\unitlength,page=2]{figures/auc-prc-nodejs/nodejs_c_1-01.pdf}}%
    \put(0.16053922,0.01715686){\color[rgb]{0,0,0}\makebox(0,0)[t]{\lineheight{1.25}\smash{\begin{tabular}[t]{c}\perReqAlg\end{tabular}}}}%
    \put(0,0){\includegraphics[width=\unitlength,page=3]{figures/auc-prc-nodejs/nodejs_c_1-01.pdf}}%
    \put(0.34436275,0.01715686){\color[rgb]{0,0,0}\makebox(0,0)[t]{\lineheight{1.25}\smash{\begin{tabular}[t]{c}\perObjAlg\end{tabular}}}}%
    \put(0,0){\includegraphics[width=\unitlength,page=4]{figures/auc-prc-nodejs/nodejs_c_1-01.pdf}}%
    \put(0.52818627,0.01715686){\color[rgb]{0,0,0}\makebox(0,0)[t]{\lineheight{1.25}\smash{\begin{tabular}[t]{c}\noDistAlg\end{tabular}}}}%
    \put(0,0){\includegraphics[width=\unitlength,page=5]{figures/auc-prc-nodejs/nodejs_c_1-01.pdf}}%
    \put(0.7120098,0.01715686){\color[rgb]{0,0,0}\makebox(0,0)[t]{\lineheight{1.25}\smash{\begin{tabular}[t]{c}\pAlpaca\end{tabular}}}}%
    \put(0,0){\includegraphics[width=\unitlength,page=6]{figures/auc-prc-nodejs/nodejs_c_1-01.pdf}}%
    \put(0.89583333,0.01715686){\color[rgb]{0,0,0}\makebox(0,0)[t]{\lineheight{1.25}\smash{\begin{tabular}[t]{c}\dAlpaca\end{tabular}}}}%
    \put(0,0){\includegraphics[width=\unitlength,page=7]{figures/auc-prc-nodejs/nodejs_c_1-01.pdf}}%
  \end{picture}%
\endgroup%

%% file: figures/auc-prc-nodejs/nodejs_c_1-03.pdf_tex
%% Creator: Inkscape inkscape 0.92.5, www.inkscape.org
%% Accompanies image file 'nodejs_c_1-03.pdf' (pdf, eps, ps)
\begingroup%
  \makeatletter%
  \providecommand\color[2][]{%
    \errmessage{(Inkscape) Color is used for the text in Inkscape, but the package 'color.sty' is not loaded}%
    \renewcommand\color[2][]{}%
  }%
  \providecommand\transparent[1]{%
    \errmessage{(Inkscape) Transparency is used (non-zero) for the text in Inkscape, but the package 'transparent.sty' is not loaded}%
    \renewcommand\transparent[1]{}%
  }%
  \providecommand\rotatebox[2]{#2}%
  \newcommand*\fsize{\dimexpr\f@size pt\relax}%
  \newcommand*\lineheight[1]{\fontsize{\fsize}{#1\fsize}\selectfont}%
  \ifx\svgwidth\undefined%
    \setlength{\unitlength}{306bp}%
    \ifx\svgscale\undefined%
      \relax%
    \else%
      \setlength{\unitlength}{\unitlength * \real{\svgscale}}%
    \fi%
  \else%
    \setlength{\unitlength}{\svgwidth}%
  \fi%
  \global\let\svgwidth\undefined%
  \global\let\svgscale\undefined%
  \makeatother%
  \begin{picture}(1,0.44362745)%
    \lineheight{1}%
    \setlength\tabcolsep{0pt}%
    \put(0,0){\includegraphics[width=\unitlength,page=1]{figures/auc-prc-nodejs/nodejs_c_1-03.pdf}}%
    \put(0.05269608,0.04289216){\color[rgb]{0,0,0}\makebox(0,0)[rt]{\lineheight{1.25}\smash{\begin{tabular}[t]{r}0.0\end{tabular}}}}%
    \put(0.05269608,0.11642157){\color[rgb]{0,0,0}\makebox(0,0)[rt]{\lineheight{1.25}\smash{\begin{tabular}[t]{r}0.2\end{tabular}}}}%
    \put(0.05269608,0.18995098){\color[rgb]{0,0,0}\makebox(0,0)[rt]{\lineheight{1.25}\smash{\begin{tabular}[t]{r}0.4\end{tabular}}}}%
    \put(0.05269608,0.26348039){\color[rgb]{0,0,0}\makebox(0,0)[rt]{\lineheight{1.25}\smash{\begin{tabular}[t]{r}0.6\end{tabular}}}}%
    \put(0.05269608,0.3370098){\color[rgb]{0,0,0}\makebox(0,0)[rt]{\lineheight{1.25}\smash{\begin{tabular}[t]{r}0.8\end{tabular}}}}%
    \put(0.05269608,0.41053922){\color[rgb]{0,0,0}\makebox(0,0)[rt]{\lineheight{1.25}\smash{\begin{tabular}[t]{r}1.0\end{tabular}}}}%
    \put(0,0){\includegraphics[width=\unitlength,page=2]{figures/auc-prc-nodejs/nodejs_c_1-03.pdf}}%
    \put(0.16053922,0.01715686){\color[rgb]{0,0,0}\makebox(0,0)[t]{\lineheight{1.25}\smash{\begin{tabular}[t]{c}\perReqAlg\end{tabular}}}}%
    \put(0,0){\includegraphics[width=\unitlength,page=3]{figures/auc-prc-nodejs/nodejs_c_1-03.pdf}}%
    \put(0.34436275,0.01715686){\color[rgb]{0,0,0}\makebox(0,0)[t]{\lineheight{1.25}\smash{\begin{tabular}[t]{c}\perObjAlg\end{tabular}}}}%
    \put(0,0){\includegraphics[width=\unitlength,page=4]{figures/auc-prc-nodejs/nodejs_c_1-03.pdf}}%
    \put(0.52818627,0.01715686){\color[rgb]{0,0,0}\makebox(0,0)[t]{\lineheight{1.25}\smash{\begin{tabular}[t]{c}\noDistAlg\end{tabular}}}}%
    \put(0,0){\includegraphics[width=\unitlength,page=5]{figures/auc-prc-nodejs/nodejs_c_1-03.pdf}}%
    \put(0.7120098,0.01715686){\color[rgb]{0,0,0}\makebox(0,0)[t]{\lineheight{1.25}\smash{\begin{tabular}[t]{c}\pAlpaca\end{tabular}}}}%
    \put(0,0){\includegraphics[width=\unitlength,page=6]{figures/auc-prc-nodejs/nodejs_c_1-03.pdf}}%
    \put(0.89583333,0.01715686){\color[rgb]{0,0,0}\makebox(0,0)[t]{\lineheight{1.25}\smash{\begin{tabular}[t]{c}\dAlpaca\end{tabular}}}}%
    \put(0,0){\includegraphics[width=\unitlength,page=7]{figures/auc-prc-nodejs/nodejs_c_1-03.pdf}}%
  \end{picture}%
\endgroup%

%% file: figures/auc-prc-nodejs/nodejs_c_1-09.pdf_tex
%% Creator: Inkscape inkscape 0.92.5, www.inkscape.org
%% Accompanies image file 'nodejs_c_1-09.pdf' (pdf, eps, ps)
\begingroup%
  \makeatletter%
  \providecommand\color[2][]{%
    \errmessage{(Inkscape) Color is used for the text in Inkscape, but the package 'color.sty' is not loaded}%
    \renewcommand\color[2][]{}%
  }%
  \providecommand\transparent[1]{%
    \errmessage{(Inkscape) Transparency is used (non-zero) for the text in Inkscape, but the package 'transparent.sty' is not loaded}%
    \renewcommand\transparent[1]{}%
  }%
  \providecommand\rotatebox[2]{#2}%
  \newcommand*\fsize{\dimexpr\f@size pt\relax}%
  \newcommand*\lineheight[1]{\fontsize{\fsize}{#1\fsize}\selectfont}%
  \ifx\svgwidth\undefined%
    \setlength{\unitlength}{362.25bp}%
    \ifx\svgscale\undefined%
      \relax%
    \else%
      \setlength{\unitlength}{\unitlength * \real{\svgscale}}%
    \fi%
  \else%
    \setlength{\unitlength}{\svgwidth}%
  \fi%
  \global\let\svgwidth\undefined%
  \global\let\svgscale\undefined%
  \makeatother%
  \begin{picture}(1,0.3747412)%
    \lineheight{1}%
    \setlength\tabcolsep{0pt}%
    \put(0,0){\includegraphics[width=\unitlength,page=1]{figures/auc-prc-nodejs/nodejs_c_1-09.pdf}}%
    \put(0.04451346,0.03623188){\color[rgb]{0,0,0}\makebox(0,0)[rt]{\lineheight{1.25}\smash{\begin{tabular}[t]{r}0.0\end{tabular}}}}%
    \put(0.04451346,0.09834369){\color[rgb]{0,0,0}\makebox(0,0)[rt]{\lineheight{1.25}\smash{\begin{tabular}[t]{r}0.2\end{tabular}}}}%
    \put(0.04451346,0.16045549){\color[rgb]{0,0,0}\makebox(0,0)[rt]{\lineheight{1.25}\smash{\begin{tabular}[t]{r}0.4\end{tabular}}}}%
    \put(0.04451346,0.22256729){\color[rgb]{0,0,0}\makebox(0,0)[rt]{\lineheight{1.25}\smash{\begin{tabular}[t]{r}0.6\end{tabular}}}}%
    \put(0.04451346,0.28467909){\color[rgb]{0,0,0}\makebox(0,0)[rt]{\lineheight{1.25}\smash{\begin{tabular}[t]{r}0.8\end{tabular}}}}%
    \put(0.04451346,0.34679089){\color[rgb]{0,0,0}\makebox(0,0)[rt]{\lineheight{1.25}\smash{\begin{tabular}[t]{r}1.0\end{tabular}}}}%
    \put(0,0){\includegraphics[width=\unitlength,page=2]{figures/auc-prc-nodejs/nodejs_c_1-09.pdf}}%
    \put(0.13561077,0.01449275){\color[rgb]{0,0,0}\makebox(0,0)[t]{\lineheight{1.25}\smash{\begin{tabular}[t]{c}\perReqAlg\end{tabular}}}}%
    \put(0,0){\includegraphics[width=\unitlength,page=3]{figures/auc-prc-nodejs/nodejs_c_1-09.pdf}}%
    \put(0.29089027,0.01449275){\color[rgb]{0,0,0}\makebox(0,0)[t]{\lineheight{1.25}\smash{\begin{tabular}[t]{c}\perObjAlg\end{tabular}}}}%
    \put(0,0){\includegraphics[width=\unitlength,page=4]{figures/auc-prc-nodejs/nodejs_c_1-09.pdf}}%
    \put(0.44616977,0.01449275){\color[rgb]{0,0,0}\makebox(0,0)[t]{\lineheight{1.25}\smash{\begin{tabular}[t]{c}\noDistAlg\end{tabular}}}}%
    \put(0,0){\includegraphics[width=\unitlength,page=5]{figures/auc-prc-nodejs/nodejs_c_1-09.pdf}}%
    \put(0.60144928,0.01449275){\color[rgb]{0,0,0}\makebox(0,0)[t]{\lineheight{1.25}\smash{\begin{tabular}[t]{c}\pAlpaca\end{tabular}}}}%
    \put(0,0){\includegraphics[width=\unitlength,page=6]{figures/auc-prc-nodejs/nodejs_c_1-09.pdf}}%
    \put(0.75672878,0.01449275){\color[rgb]{0,0,0}\makebox(0,0)[t]{\lineheight{1.25}\smash{\begin{tabular}[t]{c}\dAlpaca\end{tabular}}}}%
    \put(0,0){\includegraphics[width=\unitlength,page=7]{figures/auc-prc-nodejs/nodejs_c_1-09.pdf}}%
    \put(0.91200828,0.01449275){\color[rgb]{0,0,0}\makebox(0,0)[t]{\lineheight{1.25}\smash{\begin{tabular}[t]{c}\padme\end{tabular}}}}%
    \put(0,0){\includegraphics[width=\unitlength,page=8]{figures/auc-prc-nodejs/nodejs_c_1-09.pdf}}%
  \end{picture}%
\endgroup%

%% file: figures/auc-prc-unsplash/auc_ylabel.tex
\begin{tikzpicture}
\node at (0,0)[
  scale=1,
  anchor=south,
  text=black,
  rotate=90
]{AUC};
\end{tikzpicture}

%% file: figures/auc-prc-unsplash/unsplash_c_1-01.pdf_tex
%% Creator: Inkscape inkscape 0.92.5, www.inkscape.org
%% Accompanies image file 'unsplash_c_1-01.pdf' (pdf, eps, ps)
\begingroup%
  \makeatletter%
  \providecommand\color[2][]{%
    \errmessage{(Inkscape) Color is used for the text in Inkscape, but the package 'color.sty' is not loaded}%
    \renewcommand\color[2][]{}%
  }%
  \providecommand\transparent[1]{%
    \errmessage{(Inkscape) Transparency is used (non-zero) for the text in Inkscape, but the package 'transparent.sty' is not loaded}%
    \renewcommand\transparent[1]{}%
  }%
  \providecommand\rotatebox[2]{#2}%
  \newcommand*\fsize{\dimexpr\f@size pt\relax}%
  \newcommand*\lineheight[1]{\fontsize{\fsize}{#1\fsize}\selectfont}%
  \ifx\svgwidth\undefined%
    \setlength{\unitlength}{306bp}%
    \ifx\svgscale\undefined%
      \relax%
    \else%
      \setlength{\unitlength}{\unitlength * \real{\svgscale}}%
    \fi%
  \else%
    \setlength{\unitlength}{\svgwidth}%
  \fi%
  \global\let\svgwidth\undefined%
  \global\let\svgscale\undefined%
  \makeatother%
  \begin{picture}(1,0.44362745)%
    \lineheight{1}%
    \setlength\tabcolsep{0pt}%
    \put(0,0){\includegraphics[width=\unitlength,page=1]{figures/auc-prc-unsplash/unsplash_c_1-01.pdf}}%
    \put(0.05269608,0.04289216){\color[rgb]{0,0,0}\makebox(0,0)[rt]{\lineheight{1.25}\smash{\begin{tabular}[t]{r}0.0\end{tabular}}}}%
    \put(0.05269608,0.11642157){\color[rgb]{0,0,0}\makebox(0,0)[rt]{\lineheight{1.25}\smash{\begin{tabular}[t]{r}0.2\end{tabular}}}}%
    \put(0.05269608,0.18995098){\color[rgb]{0,0,0}\makebox(0,0)[rt]{\lineheight{1.25}\smash{\begin{tabular}[t]{r}0.4\end{tabular}}}}%
    \put(0.05269608,0.26348039){\color[rgb]{0,0,0}\makebox(0,0)[rt]{\lineheight{1.25}\smash{\begin{tabular}[t]{r}0.6\end{tabular}}}}%
    \put(0.05269608,0.3370098){\color[rgb]{0,0,0}\makebox(0,0)[rt]{\lineheight{1.25}\smash{\begin{tabular}[t]{r}0.8\end{tabular}}}}%
    \put(0.05269608,0.41053922){\color[rgb]{0,0,0}\makebox(0,0)[rt]{\lineheight{1.25}\smash{\begin{tabular}[t]{r}1.0\end{tabular}}}}%
    \put(0,0){\includegraphics[width=\unitlength,page=2]{figures/auc-prc-unsplash/unsplash_c_1-01.pdf}}%
    \put(0.16053922,0.01715686){\color[rgb]{0,0,0}\makebox(0,0)[t]{\lineheight{1.25}\smash{\begin{tabular}[t]{c}\perReqAlg\end{tabular}}}}%
    \put(0,0){\includegraphics[width=\unitlength,page=3]{figures/auc-prc-unsplash/unsplash_c_1-01.pdf}}%
    \put(0.34436275,0.01715686){\color[rgb]{0,0,0}\makebox(0,0)[t]{\lineheight{1.25}\smash{\begin{tabular}[t]{c}\perObjAlg\end{tabular}}}}%
    \put(0,0){\includegraphics[width=\unitlength,page=4]{figures/auc-prc-unsplash/unsplash_c_1-01.pdf}}%
    \put(0.52818627,0.01715686){\color[rgb]{0,0,0}\makebox(0,0)[t]{\lineheight{1.25}\smash{\begin{tabular}[t]{c}\noDistAlg\end{tabular}}}}%
    \put(0,0){\includegraphics[width=\unitlength,page=5]{figures/auc-prc-unsplash/unsplash_c_1-01.pdf}}%
    \put(0.7120098,0.01715686){\color[rgb]{0,0,0}\makebox(0,0)[t]{\lineheight{1.25}\smash{\begin{tabular}[t]{c}\pAlpaca\end{tabular}}}}%
    \put(0,0){\includegraphics[width=\unitlength,page=6]{figures/auc-prc-unsplash/unsplash_c_1-01.pdf}}%
    \put(0.89583333,0.01715686){\color[rgb]{0,0,0}\makebox(0,0)[t]{\lineheight{1.25}\smash{\begin{tabular}[t]{c}\dAlpaca\end{tabular}}}}%
    \put(0,0){\includegraphics[width=\unitlength,page=7]{figures/auc-prc-unsplash/unsplash_c_1-01.pdf}}%
  \end{picture}%
\endgroup%

%% file: figures/auc-prc-unsplash/unsplash_c_1-03.pdf_tex
%% Creator: Inkscape inkscape 0.92.5, www.inkscape.org
%% Accompanies image file 'unsplash_c_1-03.pdf' (pdf, eps, ps)
\begingroup%
  \makeatletter%
  \providecommand\color[2][]{%
    \errmessage{(Inkscape) Color is used for the text in Inkscape, but the package 'color.sty' is not loaded}%
    \renewcommand\color[2][]{}%
  }%
  \providecommand\transparent[1]{%
    \errmessage{(Inkscape) Transparency is used (non-zero) for the text in Inkscape, but the package 'transparent.sty' is not loaded}%
    \renewcommand\transparent[1]{}%
  }%
  \providecommand\rotatebox[2]{#2}%
  \newcommand*\fsize{\dimexpr\f@size pt\relax}%
  \newcommand*\lineheight[1]{\fontsize{\fsize}{#1\fsize}\selectfont}%
  \ifx\svgwidth\undefined%
    \setlength{\unitlength}{362.25bp}%
    \ifx\svgscale\undefined%
      \relax%
    \else%
      \setlength{\unitlength}{\unitlength * \real{\svgscale}}%
    \fi%
  \else%
    \setlength{\unitlength}{\svgwidth}%
  \fi%
  \global\let\svgwidth\undefined%
  \global\let\svgscale\undefined%
  \makeatother%
  \begin{picture}(1,0.3747412)%
    \lineheight{1}%
    \setlength\tabcolsep{0pt}%
    \put(0,0){\includegraphics[width=\unitlength,page=1]{figures/auc-prc-unsplash/unsplash_c_1-03.pdf}}%
    \put(0.04451346,0.03623188){\color[rgb]{0,0,0}\makebox(0,0)[rt]{\lineheight{1.25}\smash{\begin{tabular}[t]{r}0.0\end{tabular}}}}%
    \put(0.04451346,0.09834369){\color[rgb]{0,0,0}\makebox(0,0)[rt]{\lineheight{1.25}\smash{\begin{tabular}[t]{r}0.2\end{tabular}}}}%
    \put(0.04451346,0.16045549){\color[rgb]{0,0,0}\makebox(0,0)[rt]{\lineheight{1.25}\smash{\begin{tabular}[t]{r}0.4\end{tabular}}}}%
    \put(0.04451346,0.22256729){\color[rgb]{0,0,0}\makebox(0,0)[rt]{\lineheight{1.25}\smash{\begin{tabular}[t]{r}0.6\end{tabular}}}}%
    \put(0.04451346,0.28467909){\color[rgb]{0,0,0}\makebox(0,0)[rt]{\lineheight{1.25}\smash{\begin{tabular}[t]{r}0.8\end{tabular}}}}%
    \put(0.04451346,0.34679089){\color[rgb]{0,0,0}\makebox(0,0)[rt]{\lineheight{1.25}\smash{\begin{tabular}[t]{r}1.0\end{tabular}}}}%
    \put(0,0){\includegraphics[width=\unitlength,page=2]{figures/auc-prc-unsplash/unsplash_c_1-03.pdf}}%
    \put(0.13561077,0.01449275){\color[rgb]{0,0,0}\makebox(0,0)[t]{\lineheight{1.25}\smash{\begin{tabular}[t]{c}\perReqAlg\end{tabular}}}}%
    \put(0,0){\includegraphics[width=\unitlength,page=3]{figures/auc-prc-unsplash/unsplash_c_1-03.pdf}}%
    \put(0.29089027,0.01449275){\color[rgb]{0,0,0}\makebox(0,0)[t]{\lineheight{1.25}\smash{\begin{tabular}[t]{c}\perObjAlg\end{tabular}}}}%
    \put(0,0){\includegraphics[width=\unitlength,page=4]{figures/auc-prc-unsplash/unsplash_c_1-03.pdf}}%
    \put(0.44616977,0.01449275){\color[rgb]{0,0,0}\makebox(0,0)[t]{\lineheight{1.25}\smash{\begin{tabular}[t]{c}\noDistAlg\end{tabular}}}}%
    \put(0,0){\includegraphics[width=\unitlength,page=5]{figures/auc-prc-unsplash/unsplash_c_1-03.pdf}}%
    \put(0.60144928,0.01449275){\color[rgb]{0,0,0}\makebox(0,0)[t]{\lineheight{1.25}\smash{\begin{tabular}[t]{c}\pAlpaca\end{tabular}}}}%
    \put(0,0){\includegraphics[width=\unitlength,page=6]{figures/auc-prc-unsplash/unsplash_c_1-03.pdf}}%
    \put(0.75672878,0.01449275){\color[rgb]{0,0,0}\makebox(0,0)[t]{\lineheight{1.25}\smash{\begin{tabular}[t]{c}\dAlpaca\end{tabular}}}}%
    \put(0,0){\includegraphics[width=\unitlength,page=7]{figures/auc-prc-unsplash/unsplash_c_1-03.pdf}}%
    \put(0.91200828,0.01449275){\color[rgb]{0,0,0}\makebox(0,0)[t]{\lineheight{1.25}\smash{\begin{tabular}[t]{c}\padme\end{tabular}}}}%
    \put(0,0){\includegraphics[width=\unitlength,page=8]{figures/auc-prc-unsplash/unsplash_c_1-03.pdf}}%
  \end{picture}%
\endgroup%

%% file: performance.tex
\section{Performance Evaluation}
\label{sec:performance}

In this section, we describe how we implemented and then evaluated
the runtime performance of \perReqAlg, \perObjAlg, \noDistAlg, and 
\pAlpaca, and we present the results that we obtained. Note that 
we did not evaluate the runtime performance of either \dAlpaca or 
\padme, as they do not need to compute a padding scheme prior to 
usage. We conclude this section with an analysis of the overall
bandwidth increase incurred by all of the padding algorithms.

\subsection{Implementation}
\label{sec:performance:implementation}

\paragraph{Code Overview}
We implemented all four algorithms in Python. Additionally, 
we used Cython to optimize \perObjAlg, \perReqAlg, and \pAlpaca. 
By using Cython, we were able to implement each algorithm's core 
routines in ``pure C'' using only C arrays and C data types. 
Furthermore, we leveraged the multi-processing API OpenMP
to parallelize both \perReqAlg and \pAlpaca so that they used all of
the available processors on our evaluation platform (see
\secref{sec:performance:platform}).

\paragraph{\perReqAlg Improvements}
Our implementation of \perReqAlg improved upon the underlying algorithm
in the following two ways:
\begin{itemize}
\item \textbf{\mutInfoEst{\blahutIdx}{\privDataRV}{\distortPubDataRV}
  calculations.} In our tests, calculating
  \mutInfoEst{\blahutIdx}{\privDataRV}{\distortPubDataRV} for every
  $\blahutIdx \ge 0$ (and specifically the needed logarithms) was a
  performance bottleneck, and so our implementation of \perReqAlg
  calculated only
  \mutInfoEst{10\blahutIdx}{\privDataRV}{\distortPubDataRV}, instead.
  As such, the algorithm terminated after the first iteration
  $10(\blahutIdx+1)$ for which
  $\mutInfoEst{10\blahutIdx}{\privDataRV}{\distortPubDataRV} -
  \mutInfoEst{10(\blahutIdx+1)}{\privDataRV}{\distortPubDataRV} <
  10\blahutCutoff$.
\item \textbf{Incremental update.} Despite this optimization,
  \perReqAlg was still considerably slower than the alternatives (see
  \secref{sec:performance:runtimes}), so much so that recalculating
  the distributions for use by \distortPubDataAlg{\cdot}, e.g., after
  changes to (the sizes of) objects, could be a considerable
  disruption to serving those changed objects.  We anticipated,
  however, that incrementally updating these distributions after only
  a few object changes would be much faster, suggesting that their
  continuous maintenance could be adequately responsive, even if
  computing them from scratch would not be.

  To this end, we implemented two versions of \perReqAlg. The first,
  \perReqAlgInit, implemented \perReqAlg as described in
  \secref{sec:algs:per-request} and was used to calculate the initial
  padding scheme for a given dataset and a given \padFactor. The
  second, \perReqAlgInc, allowed for faster updates to a padding
  scheme as object sizes were changed.  \perReqAlgInc did this by
  retaining the final \blahutProb{\blahutIdx}[\distortPubDataVal] and
  \blahutCProb{\blahutIdx}[\distortPubDataVal][\privDataVal] values
  from \perReqAlgInit (or a previous invocation of \perReqAlgInc) for
  values \privDataVal and \distortPubDataVal unaffected by the newly
  changed objects, using these values to initialize
  \blahutProb{0}[\distortPubDataVal] and
  \blahutCProb{0}[\distortPubDataVal][\privDataVal].  For other values
  of \distortPubDataVal and \privDataVal,
  \blahutProb{0}[\distortPubDataVal] and
  \blahutCProb{0}[\distortPubDataVal][\privDataVal] were initialized
  as in \secref{sec:algs:per-request}.  After this initialization
  step, \perReqAlgInc iterated in the same way as \perReqAlgInit,
  stopping once
  $\mutInfoEst{10\blahutIdx}{\privDataRV}{\distortPubDataRV} -
  \mutInfoEst{10(\blahutIdx+1)}{\privDataRV}{\distortPubDataRV} <
  10\blahutCutoff$.
\end{itemize}

\paragraph{Inputs}
Our two datasets were each converted to a table, implemented as a
Pandas DataFrame, with columns \textit{Object ID}, \textit{Object
  Size}, and \textit{Retrievals} (per unit time). For each run of an
algorithm, we supplied the algorithm with the table for the dataset
being tested and a value for \padFactor.  We also provided
\perReqAlgInit and \perReqAlgInc with a third input: the value for
\blahutCutoff to indicate when these algorithms should halt iterating
and return their padding scheme.  For the NodeJS dataset we set
$\blahutCutoff = \num{1e-3}$, and for the Unsplash dataset we set
$\blahutCutoff = \num{5e-4}$. We arrived at these values for 
\blahutCutoff by observing, for each dataset, at what point smaller
values for \blahutCutoff yielded minimal reductions to
\mutInfo{\privDataRV}{\distortPubDataRV} at the expense of increasing 
runtimes.

\paragraph{Outputs}
\perReqAlgInit, \perReqAlgInc, and \pAlpaca each returned their
scheme for \cprob{\big}{\distortPubDataRV =
  \distortPubDataVal}{\privDataRV = \privDataVal} as a compact, linear
array that only contained the scheme's nonzero values. These
algorithms also returned the auxiliary information needed to
support sampling from this array on a per-request basis.  \perObjAlg
and \noDistAlg each returned a Python dictionary that mapped each
object's original size to its padded size.

\subsection{Evaluation Platform}
\label{sec:performance:platform}

Our evaluation platform consisted of an Ubuntu 20.04.1 LTS virtual
machine running within VMware Workstation 15.5 Pro. The host machine
was outfitted with a quad-core (eight logical processors) Intel Core
i7-7700HQ CPU running at 2.80\gigahertz. We provided the virtual
machine with six of these logical processors. Furthermore, the host
machine was outfitted with 32\gigabytes of RAM, of which 
20\gigabytes were allocated to the virtual machine; however, memory 
usage was not a factor for any of the algorithms with the values of 
\padFactor tested.

\subsection{Runtime Test Procedure}
\label{sec:performance:procedure}

For \perReqAlgInit, \perObjAlg, \pAlpaca, and \noDistAlg, for each
dataset and for $\padFactor \in \{1.01, 1.02, \dots, 1.1\}$, we measured
the time it took for each algorithm to calculate its padding scheme and
return its outputs. This test was conducted ten times, and we report the
average of these measurements as the runtime reported for each
value of \padFactor.

To test \perReqAlgInc, for each dataset we first created a list of the
ten most frequently retrieved objects. Then, for each of these
objects, we simulated the object's size being increased by 25\% and
measured the time it took for \perReqAlgInc to update the padding scheme
provided by \perReqAlgInit. We then took the average of these ten
measurements as the runtime reported for each value of \padFactor.
We took this approach to test \perReqAlgInc because we found that, 
in general, updating the size of frequently requested objects led to
longer runtimes than when updating the sizes of infrequently 
requested objects. Thus, our test for \perReqAlgInc was designed to
gauge its worst-case runtime.

\subsection{Runtime Results}
\label{sec:performance:runtimes}

The results of our runtime tests are depicted in
\figref{fig:runtimes-nodejs} and \figref{fig:runtimes-unsplash}.
The relative standard deviations for all algorithms were less than 
18\%, except for \pAlpaca, which ranged up to 124\%. Our use of OpenMP 
in \pAlpaca led to high variance since its runtimes were so low.

An immediate observation is that, although each algorithm took less time
on the Unsplash dataset, \figref{fig:runtimes-unsplash} looks largely
similar to \figref{fig:runtimes-nodejs}. In other words, the relative
performance between each algorithm remained unchanged between the two
datasets.  The absolute differences between runtimes for the NodeJS
and Unsplash datasets is due primarily to the sizes of these
datasets---the NodeJS dataset contains about $17\times$ more objects.

Overall, \noDistAlg outperformed the other algorithms, as its runtime
remained constant and negligible for a given dataset. We then see 
\pAlpaca, \perObjAlg, \perReqAlgInc, and \perReqAlgInit, in that order.
Furthermore, we see from both \figref{fig:runtimes-nodejs}
and \figref{fig:runtimes-unsplash} that \perReqAlgInit required
significantly more time to compute its initial padding scheme for each
value of \padFactor than the other algorithms required. With the 
addition of \perReqAlgInc, though, \perReqAlg required considerably
less time to maintain its padding scheme as object sizes were changed.
The addition of \perReqAlgInc therefore put \perReqAlg's steady-state
runtime much closer to that of the other algorithms.

Finally, we attribute the slight decrease in \perReqAlgInit's runtime
between $\padFactor = 1.09$ and $\padFactor = 1.1$ in
\figref{fig:runtimes-unsplash} to our decision to calculate only
\mutInfoEst{10\blahutIdx}{\privDataRV}{\distortPubDataRV}, rather than
each \mutInfoEst{\blahutIdx}{\privDataRV}{\distortPubDataRV}.  To
confirm this, we conducted the Unsplash runtime test a second time
while calculating each
\mutInfoEst{\blahutIdx}{\privDataRV}{\distortPubDataRV}. Though not
depicted, this test resulted in increasing runtimes for
\perReqAlgInit, from a low of 0.75 seconds at $\padFactor = 1.01$, to
a high of 3.25 seconds at $\padFactor = 1.1$. Additionally, this test
revealed that, by calculating only
\mutInfoEst{10\blahutIdx}{\privDataRV}{\distortPubDataRV},
\perReqAlgInit took 14 extra iterations at $\padFactor = 1.09$,
whereas it only took 6 extra iterations at $\padFactor = 1.1$. Thus,
we conclude that, although calculating
\mutInfoEst{10\blahutIdx}{\privDataRV}{\distortPubDataRV} resulted in
overall reduced runtimes for \perReqAlgInit, it had the effect of
causing \perReqAlgInit to iterate more times than necessary to reach
its termination condition, which impacted some values of \padFactor
more than others.

%%%%%%%%%%%  Runtime graphs for NodeJS start here %%%%%%%%%%%%
\begin{figure}[t]
	\captionsetup[subfigure]{font=normalsize,labelfont=normalsize}
	\begin{subfigure}[t]{.1\columnwidth}
		\setlength\figureheight{2in}
		\begin{minipage}[t]{1\columnwidth}
			\centering
			\hspace*{2.75em}
			\resizebox{!}{1.3em}{\input{figures/runtimes-nodejs/runtimes_legend.tex}}
		\end{minipage}
	\end{subfigure}
	
	\hspace*{-0.25em}
	\begin{subfigure}[b]{.49\columnwidth}
		\setlength\figureheight{2.3in}
		\begin{minipage}[b]{1\textwidth}
			\centering
			\vspace*{0em}\resizebox{!}{10.25em}{\input{figures/runtimes-nodejs/runtimes-nodejs.tex}}
		\end{minipage}
	\end{subfigure}%
	
	\vspace*{0.0em}
	\begin{subfigure}[b]{.43\columnwidth}
		\setlength\figureheight{2in}
		\begin{minipage}[b]{1\textwidth}
			\centering
			\hspace*{13.0em}
			\resizebox{!}{1.4em}{\input{figures/runtimes-nodejs/runtimes_xlabel.tex}}
			\vspace*{-1.0em}
		\end{minipage}
	\end{subfigure}
	\vspace*{-0.75em}
	\caption{Runtimes on the NodeJS dataset.}
	\label{fig:runtimes-nodejs}
\end{figure}
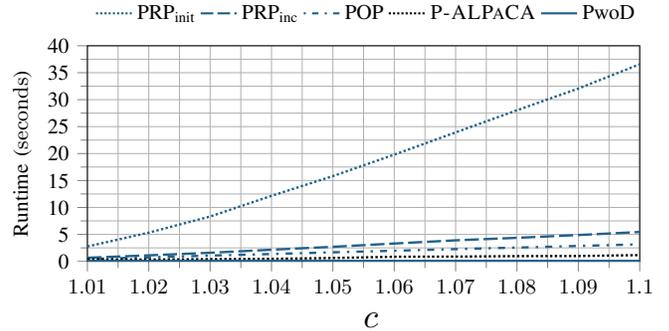
%%%%%%%%%%%  Runtime graphs for NodeJS end here %%%%%%%%%%%%

%%%%%%%%%%%  Runtime graphs for Unsplash start here %%%%%%%%%%%%
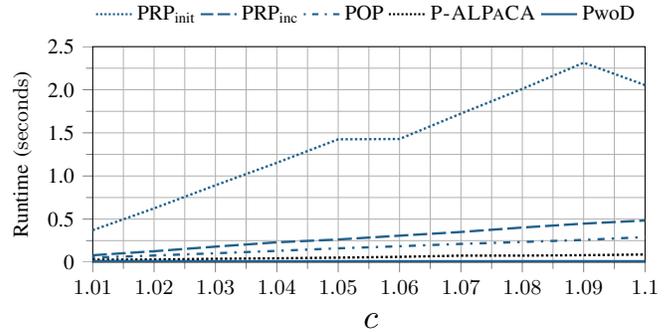
\begin{figure}[t]
	\captionsetup[subfigure]{font=normalsize,labelfont=normalsize}
	\begin{subfigure}[t]{.1\columnwidth}
		\setlength\figureheight{2in}
		\begin{minipage}[t]{1\columnwidth}
			\centering
			\hspace*{2.75em}
			\resizebox{!}{1.3em}{\input{figures/runtimes-unsplash/runtimes_legend.tex}}
		\end{minipage}
	\end{subfigure}
	
	\hspace*{-0.25em}
	\begin{subfigure}[b]{.49\columnwidth}
		\setlength\figureheight{2.3in}
		\begin{minipage}[b]{1\textwidth}
			\centering
			\vspace*{0em}\resizebox{!}{10.25em}{\input{figures/runtimes-unsplash/runtimes-unsplash.tex}}
		\end{minipage}
	\end{subfigure}%
	
	\vspace*{0.0em}
	\begin{subfigure}[b]{.43\columnwidth}
		\setlength\figureheight{2in}
		\begin{minipage}[b]{1\textwidth}
			\centering
			\hspace*{13.0em}
			\resizebox{!}{1.4em}{\input{figures/runtimes-unsplash/runtimes_xlabel.tex}}
			\vspace*{-1.0em}
		\end{minipage}
	\end{subfigure}
	\vspace*{-0.75em}
	\caption{Runtimes on the Unsplash Lite dataset.}
	\label{fig:runtimes-unsplash}
\end{figure}
%%%%%%%%%%%  Runtime graphs for Unsplash end here %%%%%%%%%%%%

\subsection{Bandwidth Increase Analysis}
\label{sec:performance:avg-overhead}

Our final evaluation analyzed the overall bandwidth increase incurred 
by each of the padding algorithms for 
$\padFactor \in \{1.01, 1.02, \dots, 1.1\}$. For this analysis, we 
calculated the multiplicative increase to bandwidth that the 
\objStoreTerm would have incurred over an arbitrary time interval if 
all of its objects had been retrieved according to the distribution 
\privDataRV. The results for the NodeJS dataset are depicted in 
\figref{fig:avg-overhead-nodejs}. Note that 
\figref{fig:avg-overhead-nodejs}'s y-axis starts at 1, which 
represents sending objects with no padding, i.e., this is the 
\objStoreTerm's baseline average bandwidth. We do not depict the 
Unsplash results, as the only significant difference was that \padme 
yielded a point at $(1.031, 1.011)$, compared to the point at 
$(1.093, 1.022)$ in \figref{fig:avg-overhead-nodejs}.

An immediate observation from \figref{fig:avg-overhead-nodejs} is that
\dAlpaca had a negligible effect on the \objStoreTerm's overall
bandwidth, due to the fact that, for all values of \padFactor, the
chosen input parameter \dAlpacaBinSize (see
\secref{sec:security:algos}) was relatively small compared to the
majority of object sizes in the dataset. Furthermore, we see that
\perReqAlg, \noDistAlg, \perObjAlg, and \pAlpaca resulted in similar
increases to the \objStoreTerm's overall bandwidth, despite producing
different padding schemes.

Finally, as mentioned in \secref{sec:security:algos}, \padme is not a
tunable padding algorithm, and so it yielded a point rather than a
line for this analysis. Thus, at approximately $\padFactor=1.09$,
\padme caused this \objStoreTerm's bandwidth to increase by 2.25\%,
whereas our algorithms increased this \objStoreTerm's bandwidth by
4.5\%. Note, though, that \padme's 2.25\% savings in bandwidth (4.5\%
minus 2.25\%) came with a \mutInfo{\privDataRV}{\distortPubDataRV}
that was roughly 18\% higher than our algorithms (see
\figref{fig:mutual-information:nodejs}).

%%%%%%%%%%%  Average Overhead graph for NodeJS start here %%%%%%%%%%%%
\begin{figure}[t]
	\captionsetup[subfigure]{font=normalsize,labelfont=normalsize}
	\begin{subfigure}[t]{.1\columnwidth}
		\setlength\figureheight{2in}
		\begin{minipage}[t]{1\columnwidth}
			\centering
			\hspace*{0.6em}
			\resizebox{!}{1.3em}{\input{figures/avg-overhead-nodejs/avg-overhead_legend.tex}}
		\end{minipage}
	\end{subfigure}
	
	\hspace*{-0.25em}
	\begin{subfigure}[b]{.49\columnwidth}
		\setlength\figureheight{2.3in}
		\begin{minipage}[b]{1\textwidth}
			\centering
			\vspace*{0em}\resizebox{!}{10.0em}{\input{figures/avg-overhead-nodejs/avg-overhead-nodejs.tex}}
		\end{minipage}
	\end{subfigure}%
	
	\vspace*{0.0em}
	\begin{subfigure}[b]{.43\columnwidth}
		\setlength\figureheight{2in}
		\begin{minipage}[b]{1\textwidth}
			\centering
			\hspace*{13.5em}
			\resizebox{!}{1.4em}{\input{figures/avg-overhead-nodejs/avg-overhead_xlabel.tex}}
			\vspace*{-1.0em}
		\end{minipage}
	\end{subfigure}
	\vspace*{-0.75em}
	\caption{Bandwidth increase for the NodeJS dataset.}
	\label{fig:avg-overhead-nodejs}
\end{figure}
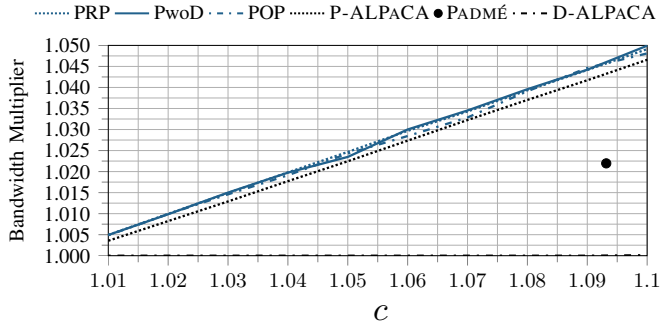
%%%%%%%%%%%  Average Overhead graph for NodeJS end here %%%%%%%%%%%%

%% file: figures/runtimes-nodejs/runtimes_legend.tex
\newenvironment{customlegend}[1][]{%
    \begingroup
    % inits/clears the lists (which might be populated from previous
    % axes):
    \csname pgfplots@init@cleared@structures\endcsname
    \pgfplotsset{#1}%
}{%
    % draws the legend:
    \csname pgfplots@createlegend\endcsname
    \endgroup
}%

\def\addlegendimage{\csname pgfplots@addlegendimage\endcsname}

\begin{tikzpicture}

\begin{customlegend}[
	legend style={{font={\fontsize{10pt}{12}\selectfont}},{draw=none}},
	legend cell align={left},
	legend columns=5,
	legend entries={{\perReqAlgInit},{\perReqAlgInc},{\perObjAlg},{\pAlpaca},{\noDistAlg}}]
%legend cell align={center},
%legend style={at={(0.97,0.03)}, anchor=south east, draw=white!80.0!black, nodes={scale=0.618, transform shape}}
%]
\addlegendimage{line width=1pt, densely dotted, curvecolor}
\addlegendimage{line width=1pt, dash pattern=on 6pt off 2pt, curvecolor}
\addlegendimage{line width=1pt, dash pattern=on 1pt off 3pt on 3pt off 3pt, curvecolor}
\addlegendimage{line width=1pt, densely dotted, black}
\addlegendimage{line width=1pt, solid, curvecolor}

\end{customlegend}

\end{tikzpicture}

%% file: figures/runtimes-nodejs/runtimes-nodejs.tex
\begin{tikzpicture}

\pgfplotsset{every axis/.append style={
					xlabel={},
					ylabel={Runtime (\secs)},
					compat=1.3,
                    label style={font=\small},
                    tick label style={font=\small}  
                    }}

\begin{axis}[
xmin=1.01, xmax=1.1,
ymin=0, ymax=40,
xtick={1.01,1.02,1.03,1.04,1.05,1.06,1.07,1.08,1.09,1.1},
ytick={0,5,10,15,20,25,30,35,40},
yticklabels={0,5,10,15,20,25,30,35,40},
width=2.25\figurewidth,
height=1.1\figurewidth,
tick align=outside,
tick pos=left,
xmajorgrids,
minor tick num=1,
x grid style={lightgray!92.026143790849673!black},
ymajorgrids,
y grid style={lightgray!92.026143790849673!black},
grid=both,
%legend entries={{1 responder},{32 responders},{64 responders},{96 responders},{128 responders}},
%legend cell align={left},
%legend style={at={(0.03,0.97)}, anchor=north west, draw=white!80.0!black, nodes={scale=0.618, transform shape}}
]
%\addlegendimage{line width=1pt, color0}
%\addlegendimage{line width=1pt, color0, dotted}
%\addlegendimage{line width=1pt, color0, dash pattern=on 1pt off 3pt on 3pt off 3pt}
%\addlegendimage{line width=1pt, color0, dashed}
%\addlegendimage{line width=2pt, color0, dashed}
\addplot [line width=1pt, densely dotted, curvecolor]
table {%
1.01 2.7570274114608764
1.02 5.325932741165161
1.03 8.320385837554932
1.04 12.1544593334198
1.05 15.816859245300293
1.06 19.801548576354982
1.07 23.922722697257996
1.08 28.037268543243407
1.09 32.06004500389099
1.1 36.56845226287842
};
\addplot [line width=1pt, dash pattern=on 6pt off 2pt, curvecolor]
table{%
1.01 0.663546753
1.02 1.115452385
1.03 1.590560555
1.04 2.154585886
1.05 2.696105242
1.06 3.308991027
1.07 3.900485086
1.08 4.363044095
1.09 4.884126973
1.1 5.447285628
};
\addplot [line width=1pt, dash pattern=on 1pt off 3pt on 3pt off 3pt, curvecolor]
table {%
1.01 0.4366678476333618
1.02 0.7394137859344483
1.03 1.0672799110412599
1.04 1.3580400466918945
1.05 1.6761354684829712
1.06 1.9659568309783935
1.07 2.267652678489685
1.08 2.5529378414154054
1.09 2.859099507331848
1.1 3.1534498929977417
};
\addplot [line width=1pt, densely dotted, black]
table {%
1.01 0.36461827754974363
1.02 0.3275375604629517
1.03 0.401826286315918
1.04 0.46978049278259276
1.05 0.6228283166885376
1.06 0.8478482484817504
1.07 0.8817314147949219
1.08 0.9686981201171875
1.09 0.9867099285125732
1.1 1.1620088338851928
};
\addplot [line width=1pt, solid, curvecolor]
table {%
1.01 0.0783278226852417
1.02 0.07841565608978271
1.03 0.07803113460540771
1.04 0.0774298906326294
1.05 0.07856214046478271
1.06 0.07871861457824707
1.07 0.0772998571395874
1.08 0.08339133262634277
1.09 0.07844724655151367
1.1 0.07804665565490723
};
\end{axis}
\end{tikzpicture}

%% file: figures/runtimes-nodejs/runtimes_xlabel.tex
\begin{tikzpicture}
\node at (0,0)[
  scale=1,
  anchor=south,
  text=black,
  rotate=0
]{$\padFactor$};
\end{tikzpicture}

%% file: figures/runtimes-unsplash/runtimes_legend.tex
\newenvironment{customlegend}[1][]{%
    \begingroup
    % inits/clears the lists (which might be populated from previous
    % axes):
    \csname pgfplots@init@cleared@structures\endcsname
    \pgfplotsset{#1}%
}{%
    % draws the legend:
    \csname pgfplots@createlegend\endcsname
    \endgroup
}%

\def\addlegendimage{\csname pgfplots@addlegendimage\endcsname}

\begin{tikzpicture}

\begin{customlegend}[
	legend style={{font={\fontsize{10pt}{12}\selectfont}},{draw=none}},
	legend cell align={left},
	legend columns=5,
	legend entries={{\perReqAlgInit},{\perReqAlgInc},{\perObjAlg},{\pAlpaca},{\noDistAlg}}]
%legend cell align={center},
%legend style={at={(0.97,0.03)}, anchor=south east, draw=white!80.0!black, nodes={scale=0.618, transform shape}}
%]
\addlegendimage{line width=1pt, densely dotted, curvecolor}
\addlegendimage{line width=1pt, dash pattern=on 6pt off 2pt, curvecolor}
\addlegendimage{line width=1pt, dash pattern=on 1pt off 3pt on 3pt off 3pt, curvecolor}
\addlegendimage{line width=1pt, densely dotted, black}
\addlegendimage{line width=1pt, solid, curvecolor}

\end{customlegend}

\end{tikzpicture}

%% file: figures/runtimes-unsplash/runtimes-unsplash.tex
\begin{tikzpicture}

\pgfplotsset{every axis/.append style={
					xlabel={},
					ylabel={Runtime (\secs)},
					compat=1.3,
                    label style={font=\small},
                    tick label style={font=\small}  
                    }}

\begin{axis}[
xmin=1.01, xmax=1.1,
ymin=0, ymax=2.5,
xtick={1.01,1.02,1.03,1.04,1.05,1.06,1.07,1.08,1.09,1.1},
ytick={0,0.5,1.0,1.5,2.0,2.5},
yticklabels={0,0.5,1.0,1.5,2.0,2.5},
width=2.25\figurewidth,
height=1.1\figurewidth,
tick align=outside,
tick pos=left,
xmajorgrids,
minor tick num=1,
x grid style={lightgray!92.026143790849673!black},
ymajorgrids,
y grid style={lightgray!92.026143790849673!black},
grid=both,
%legend entries={{1 responder},{32 responders},{64 responders},{96 responders},{128 responders}},
%legend cell align={left},
%legend style={at={(0.03,0.97)}, anchor=north west, draw=white!80.0!black, nodes={scale=0.618, transform shape}}
]
%\addlegendimage{line width=1pt, color0}
%\addlegendimage{line width=1pt, color0, dotted}
%\addlegendimage{line width=1pt, color0, dash pattern=on 1pt off 3pt on 3pt off 3pt}
%\addlegendimage{line width=1pt, color0, dashed}
%\addlegendimage{line width=2pt, color0, dashed}
\addplot [line width=1pt, densely dotted, curvecolor]
table {%
1.01 0.37108001708984373
1.02 0.6244914054870605
1.03 0.8920913934707642
1.04 1.152976655960083
1.05 1.4250144004821776
1.06 1.4288505554199218
1.07 1.723721480369568
1.08 2.00927791595459
1.09 2.313656234741211
1.1 2.0539715051651
};
\addplot [line width=1pt, dash pattern=on 6pt off 2pt, curvecolor]
table{%
1.01 0.080764365
1.02 0.127295756
1.03 0.180532479
1.04 0.230577874
1.05 0.26346643
1.06 0.307756495
1.07 0.350497413
1.08 0.402058458
1.09 0.447527242
1.1 0.483524466
};
\addplot [line width=1pt, dash pattern=on 1pt off 3pt on 3pt off 3pt, curvecolor]
table {%
1.01 0.052715492248535153
1.02 0.07773644924163818
1.03 0.10368852615356446
1.04 0.13136625289916992
1.05 0.16096112728118897
1.06 0.1850059986114502
1.07 0.21232705116271972
1.08 0.23453474044799805
1.09 0.2584112882614136
1.1 0.2913126707077026
};
\addplot [line width=1pt, densely dotted, black]
table {%
1.01 0.02849256992340088
1.02 0.03099067211151123
1.03 0.03931279182434082
1.04 0.045276832580566403
1.05 0.05295209884643555
1.06 0.062122154235839847
1.07 0.07618544101715088
1.08 0.07616939544677734
1.09 0.08149478435516358
1.1 0.08909015655517578
};
\addplot [line width=1pt, solid, curvecolor]
table {%
1.01 0.013563084602355956
1.02 0.012564396858215332
1.03 0.011486434936523437
1.04 0.010989856719970704
1.05 0.01132669448852539
1.06 0.01148672103881836
1.07 0.011106443405151368
1.08 0.011681175231933594
1.09 0.011910176277160645
1.1 0.011432170867919922
};
\end{axis}
\end{tikzpicture}

%% file: figures/runtimes-unsplash/runtimes_xlabel.tex
\begin{tikzpicture}
\node at (0,0)[
  scale=1,
  anchor=south,
  text=black,
  rotate=0
]{$\padFactor$};
\end{tikzpicture}

%% file: figures/avg-overhead-nodejs/avg-overhead_legend.tex
\newenvironment{customlegend}[1][]{%
    \begingroup
    % inits/clears the lists (which might be populated from previous
    % axes):
    \csname pgfplots@init@cleared@structures\endcsname
    \pgfplotsset{#1}%
}{%
    % draws the legend:
    \csname pgfplots@createlegend\endcsname
    \endgroup
}%

\def\addlegendimage{\csname pgfplots@addlegendimage\endcsname}

\begin{tikzpicture}

\begin{customlegend}[
	legend style={{font={\fontsize{10pt}{12}\selectfont}},{draw=none}},
	legend cell align={left},
	legend columns=6,
	legend entries={{\perReqAlg},{\noDistAlg},{\perObjAlg},{\pAlpaca},{\padme},{\dAlpaca}}]
%legend cell align={center},
%legend style={at={(0.97,0.03)}, anchor=south east, draw=white!80.0!black, nodes={scale=0.618, transform shape}}
%]
\addlegendimage{line width=1pt, densely dotted, curvecolor}
\addlegendimage{line width=1pt, solid, curvecolor}
\addlegendimage{line width=1pt, dash pattern=on 1pt off 3pt on 3pt off 3pt, curvecolor}
\addlegendimage{line width=1pt, densely dotted, black}
\addlegendimage{only marks, mark=*, mark options={fill=black}}
\addlegendimage{line width=1pt, dash pattern=on 1pt off 3pt on 3pt off 3pt, black}

\end{customlegend}

\end{tikzpicture}

%% file: figures/avg-overhead-nodejs/avg-overhead-nodejs.tex
\begin{tikzpicture}

\pgfplotsset{every axis/.append style={
					xlabel={},
					ylabel={Bandwidth Multiplier},
					compat=1.3,
                    label style={font=\small},
                    tick label style={font=\small}  
                    }}

\begin{axis}[
xmin=1.01, xmax=1.1,
ymin=1.000, ymax=1.050,
xtick={1.01,1.02,1.03,1.04,1.05,1.06,1.07,1.08,1.09,1.1},
ytick={1.000,1.005,1.010,1.015,1.020,1.025,1.030,1.035,1.040,1.045,1.050},
yticklabels={1.000,1.005,1.010,1.015,1.020,1.025,1.030,1.035,1.040,1.045,1.050},
width=2.25\figurewidth,
height=1.1\figurewidth,
tick align=outside,
tick pos=left,
xmajorgrids,
minor tick num=1,
x grid style={lightgray!92.026143790849673!black},
ymajorgrids,
y grid style={lightgray!92.026143790849673!black},
grid=both,
%legend entries={{1 responder},{32 responders},{64 responders},{96 responders},{128 responders}},
%legend cell align={left},
%legend style={at={(0.03,0.97)}, anchor=north west, draw=white!80.0!black, nodes={scale=0.618, transform shape}}
]
%\addlegendimage{line width=1pt, color0}
%\addlegendimage{line width=1pt, color0, dotted}
%\addlegendimage{line width=1pt, color0, dash pattern=on 1pt off 3pt on 3pt off 3pt}
%\addlegendimage{line width=1pt, color0, dashed}
%\addlegendimage{line width=2pt, color0, dashed}
\addplot [line width=1pt, densely dotted, curvecolor]
table {%
1.01 1.0049392989948984
1.02 1.0099439666205694
1.03 1.0149141185283068
1.04 1.0198063939968378
1.05 1.0246701286255333
1.06 1.029675359348177
1.07 1.034286156069867
1.08 1.0393225684924945
1.09 1.0441422344979014
1.1 1.0491113072465186
};
\addplot [line width=1pt, solid, curvecolor]
table {%
1.01 1.0049210473838086
1.02 1.0098616922912456
1.03 1.0150028253769114
1.04 1.0197476947884114
1.05 1.0234881792391946
1.06 1.0300084820854916
1.07 1.0345250702043935
1.08 1.0395829582228546
1.09 1.0442172187849508
1.1 1.049924328340712
};
\addplot [line width=1pt, dash pattern=on 1pt off 3pt on 3pt off 3pt, curvecolor]
table {%
1.01 1.0048645020815459
1.02 1.00990792561299
1.03 1.0145895886573717
1.04 1.0191273054941026
1.05 1.0242753069848978
1.06 1.0284523907292025
1.07 1.0328623746461143
1.08 1.0391064624619293
1.09 1.044629129908443
1.1 1.04809973009054
};
\addplot [line width=1pt, densely dotted, black]
table {%
1.01 1.0035545989439274
1.02 1.0082305763087036
1.03 1.0128972159663123
1.04 1.017723250384823
1.05 1.0224549718448173
1.06 1.0273587446747643
1.07 1.0322510219802992
1.08 1.037035801621319
1.09 1.0417122646388317
1.1 1.0465744982841605
};
%\addplot [scatter, only marks, mark=*, mark options={fill=black}]
\addplot [only marks, mark=*, mark options={fill=black}]
table {%
1.093167701863354 1.021955264756098
};
\addplot [line width=1pt, dash pattern=on 1pt off 3pt on 3pt off 3pt, black]
table {%
1.01 0.9999999999981196
1.02 0.9999999999981196
1.03 0.9999999999981196
1.04 0.9999999999981196
1.05 0.9999999999981196
1.06 0.9999999999981196
1.07 0.9999999999981196
1.08 0.9999999999981196
1.09 0.9999999999981196
1.1 1.0001156157372673
};
\end{axis}
\end{tikzpicture}

%% file: figures/avg-overhead-nodejs/avg-overhead_xlabel.tex
\begin{tikzpicture}
\node at (0,0)[
  scale=1,
  anchor=south,
  text=black,
  rotate=0
]{$\padFactor$};
\end{tikzpicture}

%% file: discussion.tex
\section{Discussion}
\label{sec:discussion}

We see primarily two opportunities for improving on our results.
First, perhaps the most significant limitation of our results is that
the privacy metric we optimize measures privacy only for
\textit{independent} object retrievals.  In several contexts, most
notably web browsing, there are dependencies among objects retrieved
due to hyperlinking, a fact that has been used in several previous
works to fingerprint webpages (e.g.,~\cite{miller:2014:clinic, 
cheng:1998:traffic, hintz:2002:fingerprinting}).  
One way to eliminate some of this leakage is by
transcluding \textit{at the server} those objects that the client-side
browser would typically transclude when assembling a webpage (images,
scripts, stylesheets, etc.); the assembled page could then be padded at
the server as a single object.  Some statistical dependencies among
pages would nevertheless remain due to hyperlinks followed manually by
the user, however.

Second, while our performance analysis of \perReqAlg in
\secref{sec:performance} touched on the need to update padding
distributions in response to object (size) changes, we believe this is
a topic that requires further consideration.  Any \objStoreTerm that
supports updates to its objects' sizes will need to recalculate object
padding distributions in response to those updates before serving the
new objects, if it is to have any hope of protecting its clients'
privacy from the attacker we consider.  If objects are updated
frequently, then these updates might therefore need to be batched and
remain hidden from clients until after the \objStoreTerm recalculates
the necessary padding distributions.  Evaluating the best balance
among padding algorithm, batch size, and length of the recalculation
time window for a given type of \objStoreTerm is a topic of future
work.

%% file: conclusion.tex
\section{Conclusion}
\label{sec:conclusion}

Object size is a particularly potent feature for traffic analysis, and
one that encryption does nothing to obscure.  In this paper we
provided algorithms for computing padding schemes suitable for various
scenarios, in which the \objStoreTerm responds to every request with
the same (padded) object copy; in which the \objStoreTerm pads each
object anew before serving it; and in which the \objStoreTerm has no
knowledge of (or little confidence in its knowledge of) the
distribution of object requests it will receive.  In each case we
provided an algorithm for constructing a padding scheme
\distortPubDataAlg{\cdot}, subject to a padding overhead constraint,
that minimizes the information gain
\mutInfo{\privDataRV}{\distortPubDataRV} about the object identity
\privDataRV based on the padded size \distortPubDataRV of the object
returned or, in the last case, an upper bound
\mutInfoInf{\privDataRV}{\distortPubDataRV} on the information gain
that is robust to any query distribution.  Our empirical analysis
using datasets of NodeJS packages and of Unsplash Lite photos
suggested that our algorithms provide better privacy than competitors
from recent literature, and do so efficiently.